\documentclass[11pt]{article}
\usepackage[margin=1.25in]{geometry}
\RequirePackage{amsmath,amsfonts,amsthm,amssymb}
\RequirePackage{accents}
\RequirePackage{algorithm,algorithmic}
\RequirePackage{xcolor}
\RequirePackage{bbm}
\RequirePackage{gensymb}
\RequirePackage{centernot}
\RequirePackage{caption}
\RequirePackage{bm}
\RequirePackage{enumerate}
\RequirePackage{enumitem}
\RequirePackage{graphicx}
\RequirePackage{multirow}
\RequirePackage{booktabs}
\RequirePackage{mathtools}
\RequirePackage[authoryear]{natbib}

\theoremstyle{plain}
\newtheorem{theorem}{Theorem}[section]
\newtheorem{corollary}{Corollary}[section]

\newtheorem{lemma}{Lemma}[section]
\newtheorem{proposition}{Proposition}[section]
\newtheorem{assumption}{Assumption}[section]

\usepackage{url}

\theoremstyle{remark}
\newtheorem{definition}{Definition}[section]
\newtheorem{remark}{Remark}[section]
\newtheorem{example}{Example}[section]

\newcommand{\Var}{\textnormal{Var}}

\newcommand{\Cov}{\textnormal{Cov}}

\newcommand{\Unif}{\textnormal{Unif}}

\newcommand{\Bern}{\textnormal{Bern}}

\newcommand{\aipw}{\textnormal{AIPW}}
\newcommand{\pl}{\textnormal{PL}}
\newcommand{\epl}{\textnormal{EPL}}

\newcommand{\diag}{\textnormal{diag}}

\newcommand{\ate}{\textnormal{ATE}}

\newcommand{\expit}{\textnormal{expit}}

\newcommand{\ovl}{\overline}

\newcommand{\indic}{\bm{1}}

\newcommand{\simiid}{\stackrel{\mathrm{iid}}{\sim}} 

\newcommand{\ce}{\mathcal{E}}
\newcommand{\cf}{\mathcal{F}}

\newcommand{\ci}{\mathcal{I}}

\newcommand{\cn}{\mathcal{N}}

\newcommand{\cs}{\mathcal{S}}
\newcommand{\ct}{\mathcal{T}}
\newcommand{\cw}{\mathcal{W}}

\renewcommand{\le}{\leqslant}
\renewcommand{\leq}{\leqslant}

\renewcommand{\geq}{\geqslant}
\renewcommand{\Pr}{\textnormal{Pr}}
\renewcommand{\preceq}{\preccurlyeq}
\renewcommand{\succeq}{\succcurlyeq}
\newcommand{\dunif}{\mathbb{U}}

\newcommand{\fromart}[1]{\begingroup\color{blue}#1\endgroup}

\newcommand{\e}{\mathbb{E}}
\newcommand{\real}{\mathbb{R}}
\newcommand{\ess}{\mathbb{S}}

\newcommand{\mrd}{\mathrm{d}} 
\newcommand{\giv}{\!\mid\!}

\newcommand{\cg}{\mathcal{G}}
\newcommand{\cm}{\mathcal{M}}
\newcommand{\cx}{\mathcal{X}}

\newcommand{\phe}{\phantom{=}}

\newcommand{\tr}{\textnormal{tr}}

\newcommand{\sconv}{\textnormal{sconv}}
\newcommand{\indep}{\perp \!\!\! \perp}
\newcommand{\tod}{\stackrel{\mathrm{d\,\,}}\rightarrow}

\newcommand{\la}{\mathrm{LA}}
\DeclareMathOperator*{\argmin}{arg\,min}
\DeclareMathOperator*{\argmax}{arg\,max}

\title{Double machine learning and design in batch adaptive experiments}
\author{Harrison H. Li\\Stanford University
 \and
Art B. Owen\\
Stanford University
}
\date{September 2023}
\begin{document}
\maketitle
\begin{abstract}
We consider an experiment with at least two stages or batches and $O(N)$ subjects per batch.
First, we propose a semiparametric treatment effect estimator that efficiently pools information across the batches,
and show it asymptotically dominates alternatives that aggregate single batch estimates.
Then,
we consider the design problem of learning propensity scores for assigning treatment in the later batches of the experiment 
to maximize the asymptotic precision of this estimator.
For two common causal estimands,
we estimate this precision using observations from previous batches,
and then solve a finite-dimensional concave maximization problem
to adaptively learn flexible propensity scores that converge to suitably defined optima in each batch at rate $O_p(N^{-1/4})$.
By extending the framework of double machine learning,
we show this rate suffices for our pooled estimator
to attain the targeted precision after each batch,
as long as nuisance function estimates converge at rate $o_p(N^{-1/4})$.
These relatively weak rate requirements enable the investigator to avoid the common practice of discretizing the covariate space for design and estimation in batch adaptive experiments
while maintaining the advantages of pooling.
Our numerical study shows that
such discretization often leads to substantial asymptotic and finite sample precision losses 
outweighing any gains from design.
\end{abstract}

\section{Introduction}\label{sec:introduction}

In sequential experimentation, we can use 
earlier observations to adjust
our treatment allocation policy for subsequent observations and thereby gain improved
estimation of causal effects in the overall study.
For instance, for an experiment with one treatment arm and one control arm,~\citet{neyman1934} showed that choosing the number of subjects in each arm to be proportional to the outcome standard deviation of that arm
minimizes the variance of the treatment effect estimate based on the difference in means.
While these standard deviations are unknown,
they can be estimated using the initial data.
Then the Neyman allocation can be approximated to improve the sample efficiency of the remainder of the experiment~\citep{hahn2011adaptive,blackwell2022batch, zhao2023adaptive,dai2023clip}.

We study a version of this design problem 
for an experiment
divided into a small number of stages or \emph{batches}.
The design, or treatment assignment mechanism,
can be updated \emph{adaptively} for later batches based on the observations from earlier batches
to improve
the precision of the causal estimate computed at the end of the experiment.
A salient feature of our setting is knowledge of pre-treatment covariates that can further improve precision.
Thus, we conceptualize our design problem as choosing a \emph{propensity score} for each batch.
The propensity score specifies the probability that a subject receives treatment
given their covariates
(throughout, we consider the setting of a binary treatment).
The propensity score is well known to be a key mathematical object to be estimated in the causal analysis of observational data (e.g.~\citet{rosenbaum1983central}).
In a randomized experiment it is known and under the control of the investigator.
Hence, it can be exploited for design.

Our specific design objective is to minimize an appropriate scalarization of the asymptotic covariance matrix of an estimator that efficiently pools information across all batches of the experiment.
After describing our mathematical notation and setup in Section~\ref{sec:setup},
we present and study an oracle version of this pooled estimator in Section~\ref{sec:pooled_estimation}.
That oracle is given knowledge of some nuisance parameters, 
typically infinite-dimensional mean or variance functions.
It pertains to 
a so-called
``non-adaptive batch experiment"
where treatment is assigned with possibly varying but nonrandom propensity scores across batches.
In certain cases,
our oracle pooled estimator asymptotically dominates the best possible 
alternative that aggregates single batch estimates,
regardless of the per-batch propensities.
This justifies designing for the pooled estimator instead of an
aggregation-based alternative.

For such a design procedure to be useful,
however, we must show that the targeted asymptotic precision of the oracle pooled estimator is in fact attainable in a batched experiment where the propensity scores used are adaptive (data-dependent)
and the nuisance parameters need to be estimated.
We address these challenges in Section~\ref{sec:batch_clt}
by extending the framework of double machine learning formalized by~\citet{chernozhukov2018double},
hereafter DML,
to what we call a ``convergent split batch adaptive experiment," or CSBAE.
In a CSBAE,
observations in each batch are split into $K$ folds,
and treatment is assigned in each fold according to an adaptive propensity score that only depends on observations from previous batches within the same fold.
Within a given batch,
the $K$ adaptive propensity scores are further required to converge to a common limit
at rate $O_p(N^{-1/4})$ in root mean square (RMS).
Our DML extension then shows that by plugging in estimated nuisance functions,
we can construct a feasible estimator in a CSBAE
that is asymptotically equivalent to
the oracle pooled estimator computed on the limiting non-adaptive batch experiment.
The nuisance function estimates only need to converge at the rate $o_p(N^{-1/4})$.

Section~\ref{sec:batch_learning} details a finite dimensional concave maximization procedure (Algorithm~\ref{alg:csbae})
that provably constructs a CSBAE
for which the limiting propensities in each batch are sequentially optimal within a function class satisfying standard complexity conditions.
Hence, we can effectively design for our pooled estimator in a batch adaptive experiment
with theoretical guarantees that our final estimator will indeed attain the targeted optimal asymptotic precision,
even with a fairly flexible propensity score learning method and nonparametric machine learning estimates for nuisance functions.

To the best of our knowledge,
existing work either designs for a less efficient alternative to our pooled estimator,
or discretizes the covariate space at the design and estimation stages to construct a feasible variant of the pooled estimator on a batch adaptive experiment.
Our simulations in Section~\ref{sec:simulations} suggest that the latter approach in particular
can lead to substantial precision losses that swamp any gains from design,
even with a moderate number of continuous covariates.
Thus, for the practitioner,
we provide an end-to-end design and estimation procedure to efficiently handle continuous covariates in a batched experiment.

\subsection{Related work}
\label{sec:lit_review}
There has been substantial research interest in adaptive experiment designs in recent years.
In many applications,
treatment assignments are updated in an attempt to maximize the (expected) 
response values
of either those in the experiment,
as in adaptive bandit algorithms~\citep{russo2018tutorial,hao2020adaptive},
or those in the superpopulation from which the experimental subjects are assumed to arrive~\citep{xu2018fully,kasy2021adaptive}.
Inference on data collected from these algorithms can be challenging since the treatment assignment rules often do not converge~\citep{hadad2021confidence,zhang2020inference,zhang2021statistical}.
By contrast,
in our setting where the goal is purely statistical (maximizing asymptotic precision of the treatment estimate),
the design objective is a static propensity score to be learned consistently.
An interesting direction for further study would be to design for a mixture of both statistical and non-statistical objectives.
For example, one might expand the literature on tie-breaker designs~\citep{owen2020optimizing,morrison2022optimality,li2023general,kluger2023kernel}) to the setting of batched experiments.

The present work can be viewed as an extension of~\citet{hahn2011adaptive} in several directions.
Those authors considered a two batch experiment
to estimate the average treatment effect (ATE) as precisely as possible.
Using data from the first batch to estimate variance functions,
they estimate the asymptotic variance of a pooled version of the semiparametric efficient ATE estimator of~\citet{hirano2003efficient}
for a coarsely discretized covariate.
Then,
they learn a propensity score for the second batch that approximately minimizes this variance.
The covariate discretization ensures
nuisance functions and optimal propensities can be estimated at parametric $O_p(N^{-1/2})$ rates without parametric assumptions.
Consequently,
a feasible version of the pooled estimator
indeed attains the targeted asymptotic variance  on the batch adaptive experiment.
We generalize this pooling construction beyond the setting of ATE estimation and relax these rate requirements
to those described in the previous section.
This permits more efficient handling of continuous covariates through nonparametric nuisance function estimates
and more flexible adaptive propensity scores.

Other approaches to extend the work of~\citet{hahn2011adaptive} include~\citet{kato2020efficient},
who consider an online setting where subjects from a stationary superpopulation enter one at a time,
without batches.
Similarly, the literature on covariate-adjusted response-adaptive (CARA) designs has focused on different but related objectives, both statistical and ethical~\citep{zhang2007asymptotic,zhu2023covariate}.
In the batched setting,
~\citet{tabord-meehan2022stratification}
proposes a method to learn a variance-minimizing stratification of the covariate space of fixed size,
avoiding the need to discretize the space prior to observing the data as in~\citet{hahn2011adaptive}.~\citet{cytrynbaum2021designing} showed that by performing a form of highly stratified treatment assignment called local randomization,
consistent variance function estimates from the first batch
make it possible to attain the semiparametric lower bound for ATE estimation in the second batch with optimal propensity score
without having to estimate the conditional mean functions.
Neither of these approaches, however,
maintains the efficiency advantages of pooling.
They also do not immediately extend beyond ATE estimation.

\section{Setup and notation}
\label{sec:setup}
Let $T \geq 2$ be the number of batches in the experiment.
Each subject $i=1,\ldots,N_t$ in batch $t=1,\ldots,T$
has observed covariates $X_{ti}\in\real^d$ and potential
outcomes $Y_{ti}(0),Y_{ti}(1)\in\real$.
We place these in the vectors $S_{ti}=(X_{ti}^\top, Y_{ti}(0),Y_{ti}(1))^\top$ which
are exogenous in our model.
Let $Z_{ti} \in \{0,1\}$ be the binary treatment indicator for this subject,
which is controlled by the investigator.
Under the usual stable unit value treatment assumption (SUTVA),
the observed outcome is
\begin{equation}
\label{eq:sutva}
Y_{ti} = Z_{ti}Y_{ti}(1) + (1-Z_{ti})Y_{ti}(0).
\end{equation}
Then the available data for the subject is $W_{ti}=(X_{ti}^\top,Z_{ti},Y_{ti})^\top \in \cw$.
We assume that the vectors $S_{ti}\simiid P^S$ for $t=1,\ldots,T$ and $i=1,\ldots,N_t$,
for some distribution $P^S$.
Appendix~\ref{app:nonstationary} relaxes this assumption to 
permit certain forms of non-stationarity across batches,
such as covariate shifts.
It will be convenient to define the functions 
\begin{equation}
\label{eq:cond_moments}
m_0(z,x)=\e[Y(z) \giv X=x] 
\quad \text{and}\quad   
v_0(z,x)=\Var(Y(z) \giv X=x),
\end{equation}
for $z\in\{0,1\}$ and $x\in\cx$.
These expectations are taken under $P^S$.

Let $N = N_1 + \ldots +N_T$. 
Then as in~\citet{hahn2011adaptive},~\citet{che2023adaptive} and others,
we consider a proportional asymptotic regime
\begin{equation}
\label{eq:prop_asymp_limit}
\lim_{N \rightarrow \infty} \frac{N_t}{N} = \kappa_t \in (0,1), \quad t=1,\ldots,T
\end{equation}
as $N \rightarrow \infty$.
In settings where the batch sizes are fully controlled by the experimenter,
it may be theoretically preferable to make initial batch sample sizes a vanishing fraction of the total sample size~\citep{zhao2023adaptive}.
However, 
in many settings
the batch sizes are exogenously constrained to satisfy~\eqref{eq:prop_asymp_limit}
unless observations are discarded.

For various $q \geq 1$ and probability measures $P$ on some space $\Omega$,
it will be useful to consider function norms of the form
\[
\|f\|_{q,P} = \left(\int |f(w)|^q \mrd P(w)\right)^{1/q}
\]
for $f\in L^q(P)$.
We will  
use propensity scores
denoted by $e(\cdot)$ with various subscripts.
A propensity score $e(\cdot)$ specifies $e(x)=\Pr(Z=1 \mid X=x)$,
the probability of treatment conditional on covariates.
We will typically require propensity scores to lie in $\cf_{\gamma}$,
the set of all measurable functions on $\cx$ taking on values in the interval $[\gamma,1-\gamma]$ for some $\gamma \in [0,1/2)$.
We use $\|A\|$ to denote the square root of the sum of the squared entries of any vector, matrix, or tensor $A$.
For any integer $p \geq 1$, 
$\ess_+^p$ will denote the set of symmetric positive semidefinite $p \times p$ real matrices,
and $\ess_{++}^p$ will be the set of symmetric positive definite $p \times p$ real matrices.
Finally, for any real vector $v$, we write $v^{\otimes 2}=vv^\top$.

We summarize the preceding requirements for the data generating process in Assumption~\ref{assump:DGP}.
Assumption~\ref{assump:DGP} does not impose any restrictions on the treatment assignment process,
which will be discussed at length in subsequent sections.

\begin{assumption}[Data generating process]
\label{assump:DGP}
For some fixed number of batches $T \geq 2$,
the vectors
\[
S_{ti}=(X_{ti},Y_{ti}(0),Y_{ti}(1)), \quad 1 \leq t \leq T,\quad 1 \leq i \leq N_t
\]
are independent and identically distributed (i.i.d.) from a distribution $P^S$.
Furthermore, the sample sizes $N_t$ satisfy~\eqref{eq:prop_asymp_limit},
and the vector $W_{ti}=(X_{ti},Z_{ti},Y_{ti})$ is observed
where the outcomes $Y_{ti}$ satisfy the SUTVA assumption~\eqref{eq:sutva}.
\end{assumption} 

\subsection{Estimands and score equations}
Consider the setting where $T=1$ (so we can drop the batch subscript $t$),
and the observations $W_1,\ldots,W_N$
are i.i.d.
Suppose additionally that~\eqref{eq:sutva} holds
along with the unconfoundedness assumption
\[
(Y_i(0),Y_i(1)) \indep Z_i \mid X_i, \quad i=1,\ldots,N.
\]
Then many popular causal estimands $\theta_0 \in \Theta \subseteq \real^p$ are identified by a score equation 
\[
\e[s(W;\theta_0,\nu_0,e_0)]=0.
\]
In this score equation,
$\nu_0$ is a vector of possibly infinite-dimensional nuisance parameters lying in a nuisance set $\cn$,
and $e_0=e_0(\cdot):\cx \rightarrow [0,1]$ is the propensity score.
Following Section 3.1 of~\citet{chernozhukov2018double},
we will assume for simplicity that the score $s(\cdot)$ is \emph{linear} in the sense that
\begin{equation}
\label{eq:linear_score}
s(w;\theta,\nu,e) = s_a(w;\nu,e)\theta + s_b(w;\nu,e), \quad \forall w \in \cw,\ \theta \in \Theta,\ (\nu,e) \in \cn \times \cf_{\gamma}
\end{equation}
for some $\gamma \in [0,1/2)$,
$s_a(\cdot,\nu,e):\cw \rightarrow \real^{p \times p}$,
and $s_b(\cdot,\nu,e): \cw \rightarrow \real^p$.

When $T>1$,
propensity scores may vary across batches by design or external constraints.
For any propensity $e=e(\cdot)$ and integrable function $f:\cw \rightarrow \real$,
we use the subscripted notation $\e_e[f(W)]=\int f(w) \mrd P_e(w)$ where $P_e=P_e^W$ is the distribution of 
$W=(X,Z,Y)=(X,Z,ZY(1)+(1-Z)Y(0))$
induced by $S=(X,Y(0),Y(1)) \sim P^S$
and $Z \giv X \sim \Bern(e(X))$ under the SUTVA assumption~\eqref{eq:sutva}.
Further let $P^X$ be the marginal distribution of $X$ under $S \sim P^S$.
Then we will require the following score equations to hold for some $\gamma \in [0,1/2)$
to identify the our causal estimand $\theta_0$:
\begin{equation}
\label{eq:score}
\e_e[s(W;\theta_0,\nu_0,e')]=0, \quad \forall e,e' \in \cf_{\gamma}.
\end{equation}

Note that~\eqref{eq:score}
requires the score $s(\cdot)$ to have mean 0 when \emph{any} propensity score $e'(\cdot) \in \cf_{\gamma}$ is plugged in.
This plug-in propensity $e'(\cdot)$ may differ from the propensity $e(\cdot) \in \cf_{\gamma}$ used for treatment assignment
in the experiment that generated the observations $W$.
Such a robustness property is satisfied by definition so long as the score $s(\cdot)$ is \emph{doubly robust}.
It is required to ensure the validity of the pooled estimator that we propose in Section~\ref{sec:pooled_estimation}.
That estimator requires plugging in a mixture propensity score that averages propensities across all batches $t=1,\ldots,T$.
While identification of $\theta_0$ within each batch is possible by only requiring~\eqref{eq:score} to hold when $e(\cdot)=e'(\cdot)$,
this will not be sufficient to ensure validity of our pooled estimator. 
We formally restate our requirements on identification of the estimand $\theta_0$ in Assumption~\ref{assump:identification}.
\begin{assumption}[Estimand identification]
\label{assump:identification}
The estimand $\theta_0 \in \real^p$ of interest satisfies~\eqref{eq:score} for some $\gamma \in [0,1/2)$,
some nuisance parameters $\nu_0$ lying in a known convex set $\cn$,
and some score $s(\cdot)$ satisfying~\eqref{eq:linear_score}.
\end{assumption}

The first estimand we are motivated by is the ATE, given by
$\theta_{0,\ate}=\e[Y(1)-Y(0)]\in\real$
in our notation.
For an investigator interested in modeling how the treatment effect varies with $X$,
they may instead wish to estimate the regression parameter $\theta_{0,\pl} \in \real^p$ under a linear treatment effect assumption
\begin{equation}
\label{eq:pl_assumption}
\e[Y(1)-Y(0) \mid X] = \psi(X)^{\top}\theta_{0,\pl}.
\end{equation}
See~\citet{robinson1988root} for background on semiparametric estimation of $\theta_{0,\pl}$ under~\eqref{eq:pl_assumption},
which characterizes the well-known ``partially linear model."
We show next that both $\theta_{0,\ate}$ and $\theta_{0,\pl}$ are identified by score functions
that are linear in the sense of~\eqref{eq:linear_score}
and robust in the sense of~\eqref{eq:score},
and hence identifiable according to Assumption~\ref{assump:identification}.

\begin{example}[ATE estimation]
\label{ex:aipw_score}
Let $\theta_0=\theta_{0,\ate}$ be the estimand of interest.
Now consider the augmented inverse propensity weighting (AIPW) score function
\begin{equation}
\label{eq:aipw_score}
s_{\aipw}(W;\theta,\nu,e) = m(1,X)-m(0,X)+\frac{Z(Y-m(1,X))}{e(X)}-\frac{(1-Z)(Y-m(0,X))}{1-e(X)}-\theta
\end{equation}
for nuisance parameter $\nu=(m(0,\cdot),m(1,\cdot))$.
For each $\gamma > 0$,
it is well known that $\e_{e}[s_{\aipw}(W;\theta_0,\nu_0,e')]=0$ for any $e(\cdot),e'(\cdot)$ in $\cf_{\gamma}$
when $\nu_0=\nu_{0,\aipw}=(m_0(0,\cdot), m_0(1,\cdot))$ lies in
the nuisance set $\cn=\cn_{\aipw}=L^1(P^X) \times L^1(P^X)$.
Hence,
$s_{\aipw}(\cdot)$ satisfies the score equation~\eqref{eq:score}.
This score is also linear
because $s_{\aipw}=s_{\aipw,a}\theta+s_{\aipw,b}$ for
\begin{align*}
s_{\aipw,a}(W;\nu,e) & = -1,\quad\text{and} \\
s_{\aipw,b}(W;\nu,e) & = m(1,X)-m(0,X)+\frac{Z(Y-m(1,X))}{e(X)}-\frac{(1-Z)(Y-m(0,X))}{1-e(X)}.
\end{align*}
which completes the task of showing that $\theta_{0,\ate}$ satisfies Assumption~\ref{assump:identification}.
\end{example}

\begin{example}[Partially linear model]
\label{ex:epl_score}
Suppose the linear treatment effect assumption~\eqref{eq:pl_assumption} holds and $\theta_0=\theta_{0,\pl}$ is the estimand of interest. Now consider the weighted least squares score
\begin{equation}
\label{eq:epl_score}
s_{\epl}(W;\theta,\nu,e) = w(X;\nu,e)(Z-e(X))(Y-m(0,X)-Z\psi(X)^{\top}\theta)\psi(X)
\end{equation}
for nonnegative weights 
$$w(X,\nu,e)= (v(0,x)e(x)+v(1,x)(1-e(x)))^{-1}$$ 
with nuisance parameter $\nu=(m(0,\cdot),v(0,\cdot),v(1,\cdot)^\top$.
Let $\cf(\cx;I)$ be the set of all measurable functions $f:\cx \rightarrow I$.
Then if $\nu_0=(m_0(0,\cdot),v_0(0,\cdot),v_0(1,\cdot))$ lies in the nuisance set $\cn=\cn_{\epl}=L^2(P^X) \times \cf(\cx;[c,\infty)) \times \cf(\cx;[c,\infty))$
for some $c>0$,
we have $\e_e[s_{\epl}(W;\theta_0,\nu_0,e')]=0$ for any $e(\cdot),e'(\cdot)$ in $\cf_0$.
Furthermore,
$s_{\epl}(\cdot)$ is linear, 
because $s_{\epl}=s_{\epl,a}\theta+s_{\epl,b}$ for
\begin{align*}
s_{\epl,a}(W;\nu,e) & = -w(X;\nu,e)Z(Z-e(X))\psi(X)\psi(X)^{\top},\quad\text{and} \\
s_{\epl,b}(W;\nu,e) &= w(X;\nu,e)(Z-e(X))(Y-m(0,X))\psi(X).
\end{align*}
Thus, $\theta_{0,\epl}$ satisfies Assumption~\ref{assump:identification} with $\gamma=0$ and the score $s_{\epl}(\cdot)$.
Note that the score equations~\eqref{eq:score} hold for any nonnegative weight functions $w(\cdot,\nu,e) \in L^1(P^X)$,
though the specific choice in $s_{\epl}(\cdot)$ is semiparametrically efficient~\citep{chamberlain1992efficiency,ma2006efficient}.
\end{example}

Some of our later results pertain to general estimands identified by Assumption~\ref{assump:identification}.
Others will be specialized to the settings of Examples~\ref{ex:aipw_score} and~\ref{ex:epl_score}.

\section{Oracle pooled and aggregate estimation (non-adaptive)}
\label{sec:pooled_estimation}

Here we propose and analyze an oracle estimator $\hat{\theta}^*$ of a generic estimand $\theta_0$ satisfying Assumption~\ref{assump:identification} with score $s(\cdot)$.
It is an oracle in the sense that it uses the unknown
true value of the nuisance parameter $\nu_0$ from the score equations~\eqref{eq:score}.
We discuss feasible estimation of $\theta_0$, including estimation of $\nu_0$, 
in Section~\ref{sec:batch_clt}.
The estimator $\hat{\theta}^*$ pools observations across all batches $t=1,\ldots,T$.
We then prove a central limit theorem (CLT) for it
in the setting of a \emph{non-adaptive batch experiment}
where treatment in each batch $t=1,\ldots,T$ is assigned according to a fixed (non-random) propensity score $e_t(\cdot)$:

\begin{definition}[Non-adaptive batch experiment]
\label{def:non_adaptive_batch_experiment}
A \textbf{non-adaptive batch experiment} has data generating process satisfying Assumption~\ref{assump:DGP} and treatment assignments satisfying
\[
Z_{ti} = \indic(U_{ti} \leq e_t(X_{ti})), \quad t=1,\ldots,T,\ i=1,\ldots,N_t
\]
for some \emph{nonrandom} (i.e., non-adaptive) batch propensity scores $e_1(\cdot),\ldots,e_T(\cdot)$ and uniformly distributed random variables $\{U_{ti} \mid t=1,\ldots,T,i=1,\ldots,N_t\}$ that are i.i.d.\ and independent of the vectors $\{S_{ti} \mid t=1,\ldots,T,i=1,\ldots,N_t\}$.
\end{definition}

Next, we compare the pooled estimator $\hat{\theta}^*$  to an alternative oracle
that makes an optimal linear aggregation of per-batch
estimates.
It also satisfies a CLT,
but we show that our pooling
strategy dominates aggregation in terms of efficiency in the setting of Examples~\ref{ex:aipw_score} and~\ref{ex:epl_score}.
We remark that some authors
(e.g.~\citet{tabord-meehan2022stratification}) refer to this aggregation approach as pooling,
but we reserve that term for pooling data, not estimators. 

\subsection{Pooled oracle estimator}
The main idea behind the construction of our oracle estimator $\hat{\theta}^*$
is as follows.
After collecting the observations from all batches $t=1,\ldots,T$ in a non-adaptive batch experiment,
we ignore the batch structure and pool together the observations across batches.
Now consider a random draw $W=(X,Z,Y)$ from these pooled observations $\{W_{ti} \mid 1 \leq t \leq T, 1 \leq i \leq N_t\}$.
In the notation of Section~\ref{sec:setup},
it is straightforward to show that the distribution of $W$ is $P_{e_{0,N}}=\sum_{t=1}^T (N_t/N)P_{e_t}$,
where $e_{0,N}(\cdot)$ is the mixture propensity score
\begin{equation}
\label{eq:e_0_N}
e_{0,N}(x) = \Pr(Z=1 \mid X) = \sum_{t=1}^T \frac{N_t}{N}e_t(x).
\end{equation}
It will also be helpful to define the limiting mixture propensity score $e_0(\cdot)$ under the proportional asymptotics~\eqref{eq:prop_asymp_limit}:
\begin{equation}
\label{eq:e_0}
e_0(x) = \sum_{t=1}^T \kappa_te_t(x).
\end{equation}
When a particular set of nonrandom batch propensities $e_1(\cdot),\ldots,e_T(\cdot)$ is relevant,
we omit an additional subscript by letting
$\e_{0,N}[f(W)]=\e_{e_{0,N}}[f(W)]$
and $\e_0[f(W)]=\e_{e_0}[f(W)]$.
Using this notation, the oracle pooled estimator $\hat{\theta}^*$ is derived by solving the sample analogue of the mixture score equations $\e_{0,N}[s(W;\theta_0,\nu_0,e_{0,N})]=0$ for $\theta$:
\begin{equation}
\label{eq:theta_star}
\hat{\theta}^* = -\biggl(\frac{1}{N}\sum_{t=1}^T\sum_{i=1}^{N_t} s_a(W_{ti};\nu_0,e_{0,N})\biggr)^{-1} \biggl(\frac{1}{N}\sum_{t=1}^T\sum_{i=1}^{N_t} s_b(W_{ti};\nu_0,e_{0,N})\biggr).
\end{equation}
This $\hat\theta^*$ is only defined when the matrix inverse 
in~\eqref{eq:theta_star} exists.
That will be the case with probability tending to 1 so long as $\e_0[s_a(W;\nu_0,e_0)]$ is invertible,
which we will require in our theoretical results.
One such result is a CLT for $\hat{\theta}^*$.
The proofs of all our technical results are provided in Appendix~\ref{app:proofs}.
\begin{proposition}[Oracle CLT]
\label{prop:oracle_clt}
Let $\theta_0 \in \real^p$ be an estimand satisfying Assumption~\ref{assump:identification} for some $\gamma \in [0,1/2)$,
some score $s(\cdot)$,
and some nuisance parameters $\nu_0$.
Suppose observations $\{W_{ti} \mid t=1,\ldots,T,i=1,\ldots,N_t\}$ are collected from a non-adaptive batch experiment with batch propensities $e_1(\cdot),\ldots,e_T(\cdot) \in \cf_{\gamma}$,
and define $e_{0,N}=e_{0,N}(\cdot)$ and $e_0=e_0(\cdot)$ as in~\eqref{eq:e_0_N} and~\eqref{eq:e_0}, respectively.
Further assume the following conditions hold:
\begin{enumerate}
    \item \label{cond:score_continuity} For some sequence $\delta_N \downarrow 0$, we have 
    \begin{align*}
    \bigl(\e_0[\|s_a(W;\nu_0,e_{0,N})-s_a(W;\nu_0,e_0)\|^2]\bigr)^{1/2} & \leq \delta_N,\quad\text{and} \\
    \bigl(\e_0[\|s(W;\theta_0,\nu_0,e_{0,N})-s(W;\theta_0,\nu_0,e_0)\|^2]\bigr)^{1/2} & \leq \delta_N.
    \end{align*}
    \item \label{cond:s_a_invertibility} $\e_0[s_a(W;\nu_0,e_0)]$ is invertible and $\e_0[\|s(W;\theta_0,\nu_0,e_0)\|^2] < \infty$.
    \item \label{cond:moment_boundedness} For some $q>2$ and $C < \infty$ we have $\e_0[\|s(W;\theta_0,\nu_0,e_{0,N})\|^q]\leq C$ for all sufficiently large $N$.
\end{enumerate}
Then with $\hat{\theta}^*$ as defined in~\eqref{eq:theta_star}, we have
\begin{align*}
\sqrt{N}(\hat{\theta}^*-\theta_0) & \tod \mathcal{N}(0, V_0)
\end{align*}
where
\begin{align*}
V_0 & = \bigl(\e_0[s_a(W;\nu_0,e_0)]\bigr)^{-1}\bigl(\e_0[s(W;\theta_0,\nu_0,e_0)^{\otimes 2}]\bigr)\bigl(\e_0[s_a(W;\nu_0,e_0)]\bigr)^{-1}.
\end{align*}
\end{proposition}
\begin{proof}
See Appendix~\ref{proof:prop:oracle_clt}.
\end{proof}

For the estimands
$\theta_{0,\ate}$ and $\theta_{0,\pl}$,
following~\eqref{eq:theta_star} 
we use the scores $s_{\aipw}(\cdot)$ and $s_{\epl}(\cdot)$,
respectively to derive the oracle estimates
\begin{align}
\hat{\theta}^*_{\aipw} & =  \frac{1}{N} \sum_{t=1}^T\sum_{i=1}^{N_t} s_{\aipw,b}(W_{ti};\nu_{0,\aipw},e_{0,N}) \quad \text{ (recall $s_{\aipw, a}=-1$)},\quad \text{and} \label{eq:theta_star_aipw} \\
\hat{\theta}^*_{\epl} & = -\biggl(\frac{1}{N} \sum_{t=1}^T\sum_{i=1}^{N_t} s_{\epl,a}(W_{ti};\nu_{0,\epl},e_{0,N})\biggr)^{-1}\biggl(\frac{1}{N}\sum_{t=1}^T\sum_{i=1}^{N_t} s_{\epl,b}(W_{ti};\nu_{0,\epl},e_{0,N})\biggr), \label{eq:theta_star_epl} 
\end{align}
We now specialize the generic oracle CLT of Proposition~\ref{prop:oracle_clt} to these two estimators under some regularity conditions.
\begin{assumption}
\label{assump:ate_regularity}[Regularity for estimating $\theta_{0,\ate}$]
For some $C< \infty$ and $q>2$,
we have $(\e[|Y(z)|^q])^{1/q} \leq C$ and $\e[Y(z)^2 \mid X=x] \leq C$ for all $z=0,1$ and $x \in \cx$.
\end{assumption}
\begin{assumption}
\label{assump:pl_regularity}[Regularity for estimating $\theta_{0,\pl}$]
For some $C< \infty$ and $q>2$,
Assumption~\ref{assump:ate_regularity} holds.
Additionally, $\|\psi(x)\| \leq C$ for all $x \in \cx$,
and there exists $c>0$ such that $v_0(z,x) \geq c$ for all $z=0,1$ and $x \in \cx$.
Finally, the linear treatment effect assumption~\eqref{eq:pl_assumption} holds.
\end{assumption}
\begin{corollary}[Oracle CLT for $\hat{\theta}_{\aipw}^*$]
\label{cor:ate_oracle_clt}
Suppose Assumption~\ref{assump:ate_regularity} holds,
and let $\{W_{ti} \mid t=1,\ldots,T,i=1,\ldots,N_t\}$ be observations from a non-adaptive batch experiment with batch propensities $e_1(\cdot),\ldots,e_T(\cdot) \in \cf_{\gamma}$ for some $\gamma>0$.
Then $\sqrt{N}(\hat{\theta}^*_{\aipw}-\theta_{0,\ate}) \tod \mathcal{N}(0,V_{0,\aipw})$
where
\begin{align}
V_{0,\aipw} & = \e\left[\frac{v_0(1,X)}{e_0(X)} + \frac{v_0(0,X)}{1-e_0(X)} + (m_0(1,X)-m_0(0,X)-\theta_{0,\ate})^2\right]. \label{eq:v_0_aipw}
\end{align}
\end{corollary}
\begin{proof}
See Appendix \ref{proof:cor:ate_oracle_clt}.
\end{proof}

\begin{corollary}[Oracle CLT for $\hat{\theta}_{\epl}^*$]
\label{cor:pl_oracle_clt}
Suppose Assumption~\ref{assump:pl_regularity} holds,
and let $\{W_{ti} \mid t=1,\ldots,T,i=1,\ldots,N_t\}$ be observations from a non-adaptive batch experiment with batch propensities $e_1(\cdot),\ldots,e_T(\cdot) \in \cf_0$,
where $\e[e_0^2(X)(1-e_0(X))^2\psi(X)\psi(X)^{\top}] \in \ess^p_{++}$.
Then $\sqrt{N}(\hat{\theta}^*_{\epl}-\theta_{0,\pl}) \tod \mathcal{N}(0,V_{0,\epl})$,
where
\begin{align}
V_{0,\epl} & = \left(\e\left[\frac{e_0(X)(1-e_0(X))}{v_0(0,X)e_0(X) + v_0(1,X)(1-e_0(X))}\psi(X)\psi(X)^{\top}\right]\right)^{-1}. \label{eq:v_0_pl}
\end{align}
\end{corollary}
\begin{proof}
See Appendix \ref{proof:cor:pl_oracle_clt}.
\end{proof}

\subsection{The aggregated oracle estimator}

When the covariate space $\cx$ is finite,
the oracle pooled estimator $\hat{\theta}^*_{\aipw}$ is equivalent to an oracle variant of the estimator proposed by~\citet{hahn2011adaptive}.
The complexity of constructing a \emph{feasible} pooled estimator 
for an \emph{adaptive} experiment in more general settings
has led other authors to instead consider single batch estimators that can lose considerable efficiency.
For instance,~\citet{cytrynbaum2021designing} proposes simply discarding the first batch in a two-batch experiment when computing the final estimate,
which is clearly inadmissible in our proportional asymptotic regime~\eqref{eq:prop_asymp_conv}.
Section 3.2 of~\citet{tabord-meehan2022stratification} suggests instead taking a linear aggregation of estimates computed separately on each batch,
as described above.
We will now show that even the best linearly aggregated oracle estimator is asymptotically dominated by our pooled estimators $\hat{\theta}^*_{\aipw}$ and $\hat{\theta}^*_{\epl}$.
While~\citet{tabord-meehan2022stratification} hypothesized in their Appendix C.2 that this may be true for ATE estimation,
they do not pursue this further as are unable to construct a feasible pooled estimator 
attaining the targeted oracle variance
when batch propensities are chosen adaptively using their stratification trees.
By contrast, our design approach
will allow us to construct such an estimator using our extension of double machine learning in Section~\ref{sec:batch_clt}.

Fix a non-adaptive batch experiment with batch propensities $e_1(\cdot),\ldots,e_T(\cdot)$.
For each batch $t=1,\ldots,T$,
consider an (oracle) estimator $\hat{\theta}_{t,\aipw}^*$ for $\theta_{0,\ate}$ computed by solving the empirical analogue of the score equations $\e_{e_t}[s_{\aipw}(W;\theta_{0,\ate},\nu_{0,\aipw},e_t)]=0$
that averages only those observations in batch $t$:
\[
\hat{\theta}_{t,\aipw}^* = -\biggl(\frac{1}{N_t}\sum_{i=1}^{N_t} s_{\aipw,a}(W_{ti};\nu_{0,\aipw},e_t)\biggr)^{-1} \biggl(\frac{1}{N_t}\sum_{i=1}^{N_t} s_{\aipw,b}(W_{ti};\nu_{0,\aipw},e_t)\biggr).
\]
By applying Proposition~\ref{prop:oracle_clt} with a single batch,
for each $t=1,\ldots,T$ we obtain the CLT
\begin{align*}
\sqrt{N_t}(\hat{\theta}_{t,\aipw}^*-\theta_{0,\ate}) & \tod \mathcal{N}(0,V_{t,\aipw}), 
\end{align*}
where
\begin{align*}
V_{t,\aipw} &=A_{t,\aipw}^{-1}B_{t,\aipw}A_{t,\aipw}^{-1},\quad\text{for} \\
A_{t,\aipw} & = \e_{e_t}[s_{\aipw,a}(W;\nu_{0,\aipw},e_t)],\quad\text{and} \\
B_{t,\aipw} & = \e_{e_t}[s(W;\theta_{0,\ate},\nu_{0,\aipw},e_t)^{\otimes 2}].
\end{align*}
Now, as stated in~\citet{slud2018combining},
the asymptotically unbiased linear combination of $\hat{\theta}_{1,\aipw}^*,\ldots,\hat{\theta}_{T,\aipw}^*$ with the smallest asymptotic covariance matrix with respect to the semidefinite ordering is the inverse covariance weighted estimator
\[
\hat{\theta}_{\aipw}^{*,(\la)} = \biggl(\,\sum_{t=1}^T \kappa_t V_{t,\aipw}^{-1}\biggr)^{-1} \sum_{t=1}^T \kappa_t V_{t,\aipw}^{-1}\hat{\theta}_{t,\aipw}^*.
\]
This optimal linearly aggregated estimator $\hat{\theta}_{\aipw}^{*,(\la)}$ satisfies the CLT
\[
\sqrt{N}\bigl(\hat{\theta}_{\aipw}^{*,(\la)}-\theta_{0,\ate}\bigr) \tod \mathcal{N}(0,V_{\aipw}^{(\la)}), \quad V_{\aipw}^{(\la)} = \biggl(\sum_{t=1}^T \kappa_t V_{t,\aipw}^{-1}\biggr)^{-1}.
\]

We can similarly define the linearly aggregated oracle estimator $\hat{\theta}_{\epl}^{*,(\la)}$ for $\theta_{0,\pl}$ based on combining per-batch estimates from the score $s_{\epl}(\cdot)$,
which satisfies a CLT $\sqrt{N}(\hat{\theta}^{*,(\la)}_{\epl}-\theta_{0,\pl}) \tod \mathcal{N}(0,V_{\epl}^{(\la)})$.
Here $V_{\epl}^{(\la)}$ is given by replacing 
$s_{\aipw}(\cdot)$ with $s_{\epl}(\cdot)$, $\nu_{0,\aipw}$ with $\nu_{0,\epl}$,
and $\theta_{0,\ate}$ with $\theta_{0,\pl}$ in the definition of $V_{\aipw}^{(\la)}$.
Our main result is then that regardless of the batch propensities,
$\hat{\theta}_{\aipw}^{*,(\la)}$ and $\hat{\theta}_{\epl}^{*,(\la)}$ are asymptotically dominated by our pooled estimators $\hat{\theta}_{0,\aipw}$ and $\hat{\theta}_{0,\epl}$, respectively.
This motivates our work in Section~\ref{sec:batch_learning} that designs for these estimators.
\begin{theorem}[Pooling dominates linear aggregation]
\label{thm:pooled_covariance}
Under the conditions of Corollary~\ref{cor:ate_oracle_clt},
$V_{0,\aipw} \leq V^{(\la)}_{\aipw}$.
Under the conditions of Corollary~\ref{cor:pl_oracle_clt},
$V_{0,\epl} \preceq V^{(\la)}_{\epl}$.
\end{theorem}
\begin{proof}
See Appendix~\ref{proof:thm:pooled_covariance}.
\end{proof}

\section{Feasible pooled estimation in batch adaptive experiments}
\label{sec:batch_clt}
The oracle estimator $\hat\theta^*$ of~\eqref{eq:theta_star} depends on nuisance parameters $\nu_0$ that are unknown in practice.
Additionally, recall that our CLT for $\hat\theta^*$ (Proposition~\ref{prop:oracle_clt}) holds for a non-adaptive batch experiment.
Our goal is to choose propensities adaptively in each batch to improve precision.
Therefore, we 
would like to 
develop a feasible estimator $\hat\theta$ that attains the targeted asymptotic variance for experiments where treatment is assigned \emph{adaptively},
even when the nuisance parameters $\nu_0$ must be estimated.

As mentioned above,
our construction of such a feasible estimator $\hat{\theta}$ is based on extending the double machine learning (DML) framework of~\citet{chernozhukov2018double}.
The main requirements for $\hat{\theta}$ to have the same asymptotic variance as the corresponding oracle are convergence rate guarantees for both nuisance parameter estimates and the adaptive propensities.
Our DML extension ensures that these rate requirements can be made sub-parametric,
enabling the use of somewhat flexible machine learning methods.

The typical DML setting assumes access to a single sample $W_1,\ldots,W_N$ of i.i.d.\ observations.
An example of this setting is a non-adaptive batch experiment with $T=1$ and propensity $e_0(\cdot)$.
Then a standard DML estimator is based on two ingredients: a Neyman orthogonal score and cross-fitting.
Neyman orthogonality of the score $s(\cdot)$ at $(\nu_0,e_0)$ means a local insensitivity of the score equations to perturbations in $(\nu_0,e_0)$ in any direction:
\[
\frac{\partial}{\partial \lambda} \e_{e_0}\bigl[s(W_i;\theta_0,\nu_0+\lambda(\nu-\nu_0),e_0+\lambda(e-e_0))\bigr]=0, \quad \forall (\nu,e) \in \cn \times \cf_{\gamma}.
\]
It is well known that the scores $s_{\aipw}(\cdot)$ and $s_{\epl}(\cdot)$ are Neyman orthogonal~\citep{chernozhukov2018double}.
Given a Neyman orthogonal score $s(\cdot)$,
DML proceeds by constructing an estimator $\hat{\theta}$ by cross-fitting.
In cross-fitting,
the indices $1,\ldots,N$ are partitioned into $K$ (roughly) equally sized folds $\ci_1,\ldots,\ci_K$.
Then $\hat{\theta}$ is computed as the solution to the empirical score equations
$N^{-1} \sum_{k=1}^K \sum_{i \in \ci_k} s(W_i;\hat{\theta},\hat{\nu}^{(-k)},\hat{e}^{(-k)}) = 0$,
where for each $k=1,\ldots,K$,
$\hat{\nu}^{(-k)}$ and $\hat{e}^{(-k)}(\cdot)$ are estimates of $\nu_0$ and $e_0(\cdot)$,
respectively.
Each pair $(\hat{\nu}^{(-k)},\hat{e}^{(-k)}(\cdot))$ depends only on the observations $\{W_i \mid i \notin \ci_k\}$ outside fold $k$.
The sample splitting ensures that for each $k=1,\ldots,K$,
the estimates $(\hat{\nu}^{(-k)},\hat{e}^{(-k)})$ are independent of the observations in fold $k$.
By the arguments of~\citet{chernozhukov2018double},
such independence is key to guarantee that the feasible estimator $\hat{\theta}$ is equivalent (up to first order asymptotics) to the oracle $\hat{\theta}^*$ solving
$N^{-1} \sum_{k=1}^K \sum_{i \in \ci_k} s(W_i;\hat{\theta},\nu_0,e_0) = 0$ (cf.~\eqref{eq:theta_star}),
even when the estimates $(\hat{\nu}^{(-k)},\hat{e}^{(-k)})$ converge at sub-parametric rates.

To maintain this independence in an \emph{adaptive} batched experiment,
we require sample splitting at the \emph{design} stage,
as illustrated in Figure~\ref{fig:csbae},
along with convergence of the adaptive propensity scores.
Our notion of a convergent split batch adaptive experiment (CSBAE) in Definition~\ref{def:CSBAE} formalizes this. 
The main idea is to split the observations in \emph{every} batch $t=1,\ldots,T$ into $K$ folds.
Re-using the notation above from the standard DML setting with $T=1$,
we let $\ci_k$ denote the set of batch and observation indices $(t,i)$ assigned to fold $k=1,\ldots,K$.
Then the adaptive propensity used to assign treatment to a subject in batches $t=2,\ldots,T$ 
is allowed to only depend on observations in previous batches from the same fold as this subject.
To ensure that the adaptivity does not introduce any additional variability
(up to first-order asymptotics)
into the final estimator,
a CSBAE requires these adaptive propensities to converge to nonrandom limits $e_1(\cdot),\ldots,e_T(\cdot)$ at RMS rate $O_p(N^{-1/4})$.
While this convergence requirement may appear restrictive,
in Section~\ref{sec:batch_learning} we show how it can be ensured by design by solving an appropriate finite-dimensional concave maximization procedure.
Moreover, the limiting propensity scores from this procedure will be provably optimal,
in a sense we make more precise in Section~\ref{sec:batch_learning}.

\begin{figure}[!ht]
\includegraphics[width=\linewidth]{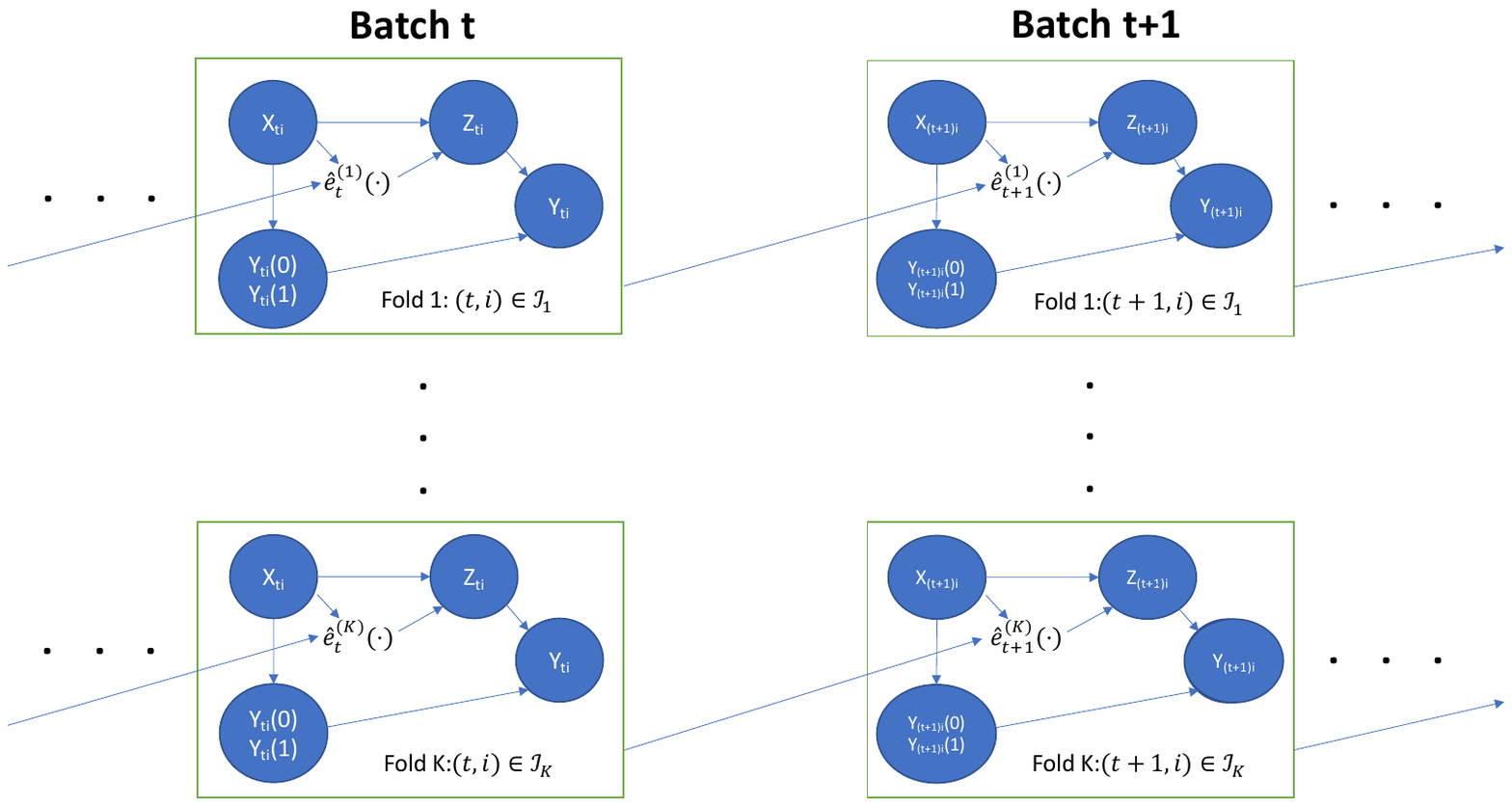}
\caption{A graphical representation of the dependencies among the observations in two batches of a CSBAE (Definition~\ref{def:CSBAE}).
Note the lack of vertical arrows,
indicating independence between the observations in fold $j$ and fold $k$ for any distinct $j$, $k$ in $\{1,\ldots,K\}$.}
\label{fig:csbae}
\end{figure}
\begin{definition}[Convergent split batch adaptive experiment]
\label{def:CSBAE}
A \textbf{convergent split batch adaptive experiment (CSBAE)}
is an experiment with data generating process satisfying Assumption~\ref{assump:DGP}
where each observation index $(t,i)$ is assigned
to one of $K$ folds $\ci_1,\ldots,\ci_K$.
The fold assignments are such that $n_{t,k}= |\{(t,i) \in \ci_k \mid i=1,\ldots,N_t\}|$,
the number of observations in batch $t$ assigned to fold $k$,
satisfies $|n_{t,k}-N_t/K| \leq 1$ for all $t=1,\ldots,T$, $k=1,\ldots,K$.
Now let $P_{N,t}^{X,(k)}$ be the empirical distribution on $\{X_{ti} \mid (t,i) \in \ci_k, 1 \leq i \leq N_t\}$,
the covariates in batch $t$ and fold $k$,
and define $\cs_t^{X,(k)}$
to be the $\sigma$-algebra generated by the covariates $\{X_{ti} \mid (t,i) \in \ci_k, 1 \leq i \leq N_t\}$ in batch $t$ and fold $k$ along with the observations $\{W_{ui} \mid (u,i) \in \ci_k, u=1,\ldots,t-1,i=1,\ldots,N_u\}$ in fold $k$ and any of the previous batches $1,\ldots,t-1$.
We further require the following for each batch $t=1,\ldots,T$ and fold $k=1,\ldots,K$:
\begin{enumerate}
\item Treatment is assigned according to an adaptive propensity $\hat{e}_t^{(k)}(\cdot)$ that is measurable with respect to $\cs_t^{X,(k)}$.
    That is, the treatment indicators can be represented as
    \begin{equation}
    \label{eq:CSBAE_treatment}
    Z_{ti}=\indic(U_{ti} \leq \hat{e}_t^{(k)}(X_{it})), \quad (t,i) \in \ci_k
    \end{equation}
    where $\{U_{ti}:1 \leq t \leq T, 1 \leq i \leq N_t\}$ is a collection of i.i.d.\ uniformly distributed random variables independent of the vectors $\{S_{ti} \mid t=1,\ldots,T,i=1,\ldots,N_t\}$.
    \item For some nonrandom propensity $e_t(\cdot)$,
    the adaptive propensity $\hat{e}_t^{(k)}(\cdot)$ satisfies
    \begin{equation}
    \label{eq:csbae_limit}
    \|\hat{e}_t^{(k)}-e_t\|_{2,P_{N,t}^{X,(k)}} = O_p(N^{-1/4}), \quad t=1,\ldots,T, \ k=1,\ldots,K. 
    \end{equation}
\end{enumerate}
\begin{remark}
The left-hand side of equation~\eqref{eq:csbae_limit} uses an $L^2$ norm on the empirical distribution $P_{N,t}^{X,(k)}$ of the covariates of the subjects that will be assigned treatment according to the learned propensity $\hat{e}_t^{(k)}(\cdot)$.
These covariates will also be used to learn $\hat{e}_t^{(k)}(\cdot)$ itself in our propensity learning procedure of Section~\ref{sec:batch_learning}.
Thus, we can interpret~\eqref{eq:csbae_limit} as a rate requirement on the ``in-sample" convergence of $\hat e_t^{(k)}(\cdot)$.
\end{remark}
\end{definition}
Given a CSBAE,
our
feasible estimator is
\begin{equation}
\label{eq:theta_hat}
\hat{\theta} = -\biggl(\frac{1}{N} \sum_{k=1}^K \sum_{(t,i) \in \ci_k} s_a(W_{ti};\hat{\nu}^{(-k)},\hat{e}^{(-k)})\biggr)^{-1}\biggl(\frac{1}{N} \sum_{k=1}^K \sum_{(t,i) \in \ci_k} s_b(W_{ti};\hat{\nu}^{(-k)},\hat{e}^{(-k)})\biggr).
\end{equation}
As in the standard (single batch) DML setting,
for each $k=1,\ldots,K$,
$\hat{\nu}^{(-k)}$ and $\hat{e}^{(-k)}(\cdot)$ are estimates of the nuisance parameters $\nu_0$ and the mixture propensity $e_{0,N}(\cdot)$ defined in~\eqref{eq:e_0_N}, respectively,
that depend only on the observations $\{W_{ti} \mid (t,i) \notin \ci_k\}$ outside fold $k$.
These observations are fully independent of the observations in fold $k$ (across \emph{all} batches $t=1,\ldots,T$) 
by the construction of a CSBAE.
As in the single batch case,
given $o_p(N^{-1/4})$ convergence of the estimators $\hat{\nu}^{(-k)}$ to $\nu_0$,
this independence along with~\eqref{eq:csbae_limit} enable a DML-style argument that $\hat{\theta}$ is asymptotically equivalent to the oracle $\hat{\theta}^*$ under~\eqref{eq:csbae_limit} computed on a counterfactual
non-adaptive batch experiment with propensities $e_1(\cdot),\ldots,e_T(\cdot)$.
This argument proceeds by coupling the treatment indicators $Z_{ti}$ in the
CSBAE with counterfactual treatment indicators $\tilde{Z}_{ti}=\bm{1}(U_{ti} \leq e_t(X_{ti}))$.
\begin{assumption}[Score properties and convergence rates for estimating nuisance parameters and the mixture propensity in a CSBAE]
\label{assump:dml}
Observations $\{W_{ti}:1 \leq t \leq T, 1 \leq i \leq N_t\}$ are collected from a CSBAE with limiting batch propensities $e_1(\cdot),\ldots,e_T(\cdot)$.
Additionally,
the estimand $\theta_0$ of interest is identified as in Assumption~\ref{assump:identification} by some score $s(\cdot)$,
nuisance parameters $\nu_0 \in \cn$,
and $\gamma \in [0,1/2)$,
such that the propensity collection $\cf_{\gamma}$ contains $e_1(\cdot),\ldots,e_T(\cdot)$.
Defining $W_{ti}(z)=(Y_{ti}(z),X_{ti},z)$ for $z=0,1$,
the score $s(\cdot)$ has the following properties:
\begin{enumerate}[label=(\alph*)]
\item \label{cond:regularity} The matrix $\e_0[s_a(W;\nu_0,e_0)]$ is invertible and $\e_0\bigl[\|s(W;\theta_0,\nu_0,e_0)\|^2\bigr] < \infty$.
\item \label{cond:score_differentiability} The mapping $\lambda \mapsto \e_{0,N}[s(W;\theta_0,\nu_0+\lambda(\nu-\nu_0),e_{0,N}+\lambda(e-e_{0,N})]$ is twice continuously differentiable on $[0,1]$ for each $(\nu,e) \in \ct=\cn \times \cf_{\gamma}$.
\item \label{cond:potential_score_equations} 
All propensities $e(\cdot) \in \cf_{\gamma}$ satisfy
\begin{equation}
\label{eq:potential_score_diff}
\e\bigl[s(W_{ti}(1);\theta_0,\nu_0,e)-s(W_{ti}(0);\theta_0,\nu_0,e) \giv X_{ti}\bigr] = 0,
\end{equation}
for all $t=1,\dots,T$ and $i=1,\dots,N_t$.
\end{enumerate}
Also, there exist estimators $\hat{\nu}^{(-k)}$ and $\hat{e}^{(-k)}(\cdot)$ of $\nu_0$ and the mixture propensity $e_{0,N}(\cdot)$ defined in~\eqref{eq:e_0_N}, 
respectively,
that depend only on the observations outside fold $k$ of the CSBAE.
Next,
there are nonrandom subsets $\ct_N \subseteq \ct$ containing $(\nu_0,e_{0,N}(\cdot))$,
such that for all $k=1,\ldots,K$,
$\Pr((\hat{\nu}^{(-k)},\hat{e}^{(-k)}(\cdot)) \in \ct_N) \rightarrow 1$ as $N \rightarrow \infty$.
The sets $\ct_N$ shrink quickly enough for the following to hold for all $(\nu,e)\in\ct_N$,
all $\lambda\in(0,1)$
and all $z\in\{0,1\}$
when $N$ is sufficiently large:
\begin{align}
\bigg\|\frac{\partial}{\partial \lambda} \e_{0,N}[s(W;\theta_0,\nu_0+\lambda(\nu-\nu_0),e_{0,N}+\lambda(e-e_{0,N}))] \Big|_{\lambda=0} \bigg\| & \leq N^{-1/2}\delta_N \label{eq:neyman_orthogonality} \\
\bigg\|\frac{\partial^2}{\partial \lambda^2} \e_{0,N}[s(W;\theta_0,\nu_0+\lambda(\nu-\nu_0),e_{0,N}+\lambda(e-e_{0,N}))] \bigg\| & \leq N^{-1/2}\delta_N
\label{eq:vanishing_second_derivatives} \\
\left(\e_{0}\bigl[\|s_a(W;\nu,e)-s_a(W;\nu_0,e_0)\|^2\bigr]\right)^{1/2} & \leq \delta_N  \label{eq:s_a_consistency} \\
\left(\e_{0}\bigl[\|s(W;\theta_0,\nu,e)-s(W;\theta_0,\nu_0,e_0)\|^2\bigr]\right)^{1/2} & \leq \delta_N
\label{eq:s_consistency} \\
\bigl(\e_{0}\bigl[\|s_a(W(z);\nu,e)\|^q\bigr]\bigr)^{1/q} & \leq C
\label{eq:s_a_moment_condition} \\
\bigl(\e_{0}\bigr[\|s(W(z);\theta_0,\nu,e\|^q\bigr]\bigr)^{1/q} & \leq C.
\label{eq:s_moment_condition}
\end{align}
Finally, 
letting $\cs^{(-k)}$ be the $\sigma$-algebra generated by the observations $\{W_{ti}:(t,i) \notin \ci_k\}$ outside fold $k$ across all batches 1 through $T$,
we require
\begin{equation}
\label{eq:S_t_neg_k}
S_t^{(-k)}(z) = o_p(N^{-1/4}), \quad z \in \{0,1\}, \quad k=1,\ldots,K
\end{equation}
where
\[
S_t^{(-k)}(z) = \sqrt{
\frac{1}{n_{t,k}} \sum_{i:(t,i) \in \ci_k} \Bigl\|\e\bigl[s(W_{ti}(z);\hat{\nu}^{(-k)},\hat{e}^{(-k)})-s(W_{ti}(z);\nu_0,e_{0,N}) \bigm| \cs^{(-k)},X_{ti}\bigr]\Bigr\|^2}.
\]
\end{assumption}

\begin{theorem}[Feasible CLT for a CSBAE]
\label{thm:batch_clt}
Suppose Assumption~\ref{assump:dml} holds. Then for $\hat{\theta}$ defined in~\eqref{eq:theta_hat}
there exists a non-adaptive batch experiment with propensities $e_1(\cdot),\ldots,e_T(\cdot)$
for which
\[
\hat{\theta} = \hat{\theta}^* + o_p(N^{-1/2}).
\]
Here $\hat{\theta}^*$ is the oracle~\eqref{eq:theta_star} computed on this non-adaptive batch experiment.
Then $\sqrt{N}(\hat{\theta}-\theta_0) \tod \mathcal{N}(0,V_0)$ where $V_0$ is the limiting covariance derived in Proposition~\ref{prop:oracle_clt}.
\end{theorem}
\begin{proof} See Appendix~\ref{proof:thm:batch_clt}.
\end{proof}

In Assumption~\ref{assump:dml},
the conditions~\ref{cond:regularity} and~\ref{cond:score_differentiability}
along with the inequalities~\eqref{eq:neyman_orthogonality} through~\eqref{eq:s_moment_condition} are direct extensions of Assumptions 3.1 and 3.2 in~\citet{chernozhukov2018double}
for ordinary DML ($T=1$).
The equations~\eqref{eq:potential_score_diff} and~\eqref{eq:S_t_neg_k} are additional requirements that enable the dependence across batches in a CSBAE to be sufficiently weak
so that $\hat{\theta}$ computed on the CSBAE is asymptotically equivalent to a version of it computed on the limiting non-adaptive batch experiment.

We can show that Assumption~\ref{assump:dml} is satisfied for estimating $\theta_{0,\ate}$ and $\theta_{0,\pl}$ with $s_{\aipw}(\cdot)$ and $s_{\epl}(\cdot)$, respectively
under simple rate conditions on nuisance parameter and propensity estimation rates that mirror those in Section 5 of~\citet{chernozhukov2018double} for single-batch DML estimators.
Then we apply Theorem~\ref{thm:batch_clt} to construct feasible pooled estimators $\hat{\theta}_{\aipw}$ and $\hat{\theta}_{\epl}$ as special cases of equation~\eqref{eq:theta_hat},
which attain the oracle asymptotic variances $V_{0,\aipw}$ and $V_{0,\epl}$ defined in~\eqref{eq:v_0_aipw} and~\eqref{eq:v_0_pl}, respectively. 

\begin{corollary}[Feasible pooled estimation of $\theta_{0,\ate}$ in a CSBAE]
\label{cor:AIPW}
Suppose observations are collected from a CSBAE
for which the regularity conditions of Assumption~\ref{assump:ate_regularity} hold
for some $q>2$ and $C<\infty$
and the limiting batch propensities $e_1(\cdot),\ldots,e_T(\cdot)$ are in $\cf_{\gamma}$
for some $\gamma > 0$.
Additionally, for each $k=1,\ldots,K$, 
suppose we have estimates $\hat{m}^{(-k)}(\cdot)$ and $\hat{e}^{(-k)}(\cdot)$ of the mean function $m_0(\cdot)$ and mixture propensity $e_{0,N}(\cdot)$,
respectively, both depending only on the observations outside fold $k$,
such that the following are true:
\begin{enumerate}
    \item $\|\hat{m}^{(-k)}(z,\cdot)-m_0(z,\cdot)\|_{2,P^X}=o_p(N^{-1/4}), \quad z=0,1,$
    \item $\|\hat{e}^{(-k)}-e_{0,N}\|_{2,P^X} = O_p(N^{-1/4}),$
    \item $\|\hat{m}^{(-k)}(z,\cdot)-m_0(z,\cdot)\|_{q,P^X} \leq C, \quad z=0,1$, 
    \item $\hat{e}^{(-k)}(\cdot) \in \cf_{\gamma}$ with probability tending to 1.
\end{enumerate}
Then $N^{1/2}(\hat{\theta}_{\aipw}-\theta_{0,\ate}) \tod \mathcal{N}(0,V_{0,\aipw})$.
\end{corollary}
\begin{proof}
See Appendix \ref{proof:cor:AIPW}.
\end{proof}

\begin{corollary}[Feasible estimation of $\theta_{0,\pl}$ in a CSBAE]
\label{cor:pl}
Suppose observations are collected from a CSBAE
for which the regularity conditions of Assumption~\ref{assump:pl_regularity} hold
for some $q>2$ and $0<c<C<\infty$,
and for which the limiting batch propensities $e_1(\cdot),\ldots,e_T(\cdot)$ are in $\cf_0$
with $\e[e_0^2(X)(1-e_0(X))^2\psi(X)\psi(X)^{\top}] \in \ess^p_{++}$.
For each $k=1,\ldots,K$, 
assume we have estimates $\hat{m}^{(-k)}(0,\cdot)$, $\hat{v}^{(-k)}(\cdot,\cdot)$, and $\hat{e}^{(-k)}(\cdot)$ of the mean function $m_0(0,\cdot)$,
the variance function $v_0(\cdot, \cdot)$,
and the mixture propensity $e_{0,N}(\cdot)$,
respectively,
all depending only on the observations outside fold $k$,
such that the following are true:
\begin{enumerate}
    \item $\|\hat{m}^{(-k)}(0,\cdot)-m_0(0,\cdot)\|_{2,P^X}=o_p(N^{-1/4})$,
    \item $\|\hat{v}^{(-k)}(z,\cdot)-v_0(z,\cdot)\|_{2,P^X}=o_p(1), \quad z=0,1$,
    \item $\|\hat{e}^{(-k)}-e_{0,N}\|_{2,P^X} = O_p(N^{-1/4})$,
    \item $\|\hat{m}^{(-k)}(0,\cdot)-m_0(0,\cdot)\|_{q,P^X} \leq C,$
    \item $\inf_{x \in \cx} \hat{v}^{(-k)}(z,x) \geq c, \quad z=0,1$.
\end{enumerate}
Then $N^{1/2}(\hat{\theta}_{\epl}-\theta_{0,\pl}) \tod \mathcal{N}(0,V_{0,\epl})$.
\end{corollary}
\begin{proof}
See Appendix \ref{proof:cor:pl}.
\end{proof}
We now compare the rate requirements in Corollaries~\ref{cor:AIPW} and~\ref{cor:pl}
with those needed to prove a feasible CLT for the linearly aggregated estimators discussed in Section~\ref{sec:pooled_estimation}.
Consider a batch adaptive experiment \emph{without} sample splitting,
so that we assign treatment using the propensities $\hat{e}_1(\cdot),\ldots,\hat{e}_T(\cdot)$
where each $\hat{e}_t(\cdot)$ is possibly random but can only depend on the observations in batches $1,\ldots,t-1$,
and converges to some nonrandom $e_t(\cdot)$
at \emph{any} rate;
in particular, this rate may be slower than the $O_p(N^{-1/4})$ rate of~\eqref{eq:csbae_limit}.
Then the standard DML results in Section 5 of~\citet{chernozhukov2018double} show that we can construct feasible estimators $\hat{\theta}_{t,\aipw}$ and $\hat{\theta}_{t,\epl}$ that are asymptotically equivalent to the oracle single-batch estimators $\hat{\theta}_{t,\aipw}^*$ and $\hat{\theta}_{t,\epl}^*$ in Section~\ref{sec:pooled_estimation},
so long as we plug in the true propensity score $\hat{e}_t(\cdot)$ used to assign treatment
and use cross-fitting at the \emph{estimation} stage,
even if all rates in Corollaries~\ref{cor:AIPW} and~\ref{cor:pl} are weakened to $o_p(1)$.
Unfortunately,
we cannot extend this construction to our pooled estimator $\hat{\theta}$ when $T \geq 2$,
since $\hat{\theta}$ plugs in estimates of the mixture propensity $e_{0,N}(\cdot)$,
which does not correspond (in general) to a propensity actually used for treatment in any batch.
Our numerical studies in Section~\ref{sec:simulations} suggest,
however, that the weaker rate requirements for a feasible CLT for the linearly aggregated estimators 
(compared to our pooled estimators) make little difference in practice.
Indeed,
there we also see finite sample advantages to pooling,
beyond those predicted by the asymptotics.

\section{Batch adaptive learning of the optimal propensity score}
\label{sec:batch_learning}
We now discuss how to learn adaptive propensity scores $\hat{e}_t^{(k)}(\cdot)$ that satisfy~\eqref{eq:csbae_limit}
with limiting propensity scores $e_t(\cdot)$ that maximize asymptotic precision of the final estimators $\hat{\theta}_{\aipw}$ and $\hat{\theta}_{\epl}$ constructed in the previous section,
as measured by their asymptotic covariance matrices $V_{0,\aipw}$ and $V_{0,\epl}$.
This generates a CSBAE on which,
by Theorem~\ref{thm:batch_clt},
feasible estimators $\hat{\theta}_{0,\aipw}$ and $\hat{\theta}_{0,\epl}$ achieve the targeted asymptotic variances $V_{0,\aipw}$ and $V_{0,\epl}$.

While $V_{0,\aipw}$ is scalar
and so there is no ambiguity in what it means to maximize asymptotic precision of $\hat{\theta}_{\aipw}$,
when $\theta_{0,\pl}$ is multivariate ($p>1$),
$V_{0,\epl}$ is a matrix.
To handle the multivariate setting,
we follow classical literature on experiment design in regression models
by scalarizing an \emph{information matrix} $\ci \in \ess_+^p$ (typically an inverse covariance matrix) using an \textit{information function} $\Psi:\ess_+^p \rightarrow \real \cup \{-\infty\}$.
See textbooks such as~\citet{atkinson2007optimum,pukelsheim2006optimal} for more background.
We generically write $\ci=\ci(e,\eta)$ to emphasize that in our setting,
$\ci$ will be indexed by a propensity score $e=e(\cdot)$ in some function class $\cf_*$
along with some unknown nuisance functions $\eta=\eta(\cdot)$ to be estimated.
Then the design objective is to learn an optimal propensity $e^*(\cdot)$,
in the sense of maximizing $\Psi(\ci)$:
\begin{equation}
\label{eq:e_star}
e^*(\cdot) \in \argmax_{e \in \cf_*} \Psi(\ci(e;\eta)).
\end{equation}
The nuisance functions $\eta$ in the information matrix are distinct from the nuisance functions $\nu$ in the score function.
Examples of $\eta$ for information matrices based on $V_{0,\aipw}$ and $V_{0,\epl}$ are given in Section~\ref{sec:apply_concave_maximization}.

\subsection{Generic convergence rates with concave maximization} 
We start by considering a feasible procedure to learn the propensity $e^*(\cdot)$ of~\eqref{eq:e_star} when the information matrix $\ci$
and the function class $\cf_*$ being optimized over generically satisfy Assumption~\ref{assump:info_matrix} below.
An explanation of how to apply this generic procedure to construct an appropriate CSBAE that designs for the estimators $\hat{\theta}_{\aipw}$ and $\hat{\theta}_{\epl}$ is deferred to
Section~\ref{sec:apply_concave_maximization}.
Here we instead focus on a precise exposition of our propensity learning procedure and the technical assumptions needed to ensure convergence,
without explicitly invoking any notation or setup from previous sections.
We begin with some generic structure on the information matrix $\ci$ and the function class $\cf_*$. 
Roughly, 
we require the information matrix to be strongly concave in $e(\cdot)$
and the function class $\cf_*$ to be not too complex.
We also allow for interval constraints on the per-batch expected proportion of subjects treated.

\begin{assumption}[Generic optimization setup]
\label{assump:info_matrix}
The information matrix $\ci=\ci(e,\eta)$ in~\eqref{eq:e_star} takes the form
\begin{equation}
\label{eq:generic_info}
\ci(e;\eta) = \e_{X \sim P}\left[f(e(X),\eta(X))\right]
\end{equation}
where $\eta:\cx \rightarrow \cw$ is a vector of possibly unknown functions taking values in some compact set $\cw \subseteq \real^r$,
and $P$ is a \textbf{covariate distribution} from which i.i.d.\ observations $X_1,\ldots,X_n \in \cx$ are drawn.
Additionally,
the function $f:[0,1] \times \cw \rightarrow \ess_+^p$ appearing in~\eqref{eq:generic_info} satisfies the following properties:
\begin{itemize}
    \item There exists an extension of $f(\cdot,\cdot)$ with continuous second partial derivatives to an open neighborhood of $\real^{r+1}$ containing $(\delta,1-\delta) \times \cw$ for some $\delta < 0$.
    \item For some $c>0$ we have 
\begin{equation}
\label{eq:f_strong_concavity}
-f''(e,w) \in \ess_+^p, \quad\text{and}\quad \tr(f''(e,w)) \leq -c,\quad  \forall (e,w) \in [0,1] \times \cw
\end{equation}
where $f''(e,w)$ denotes the second partial derivative of $f(\cdot,\cdot)$ with respect to the first argument, 
evaluated at $(e,w)$.
\end{itemize}
Next, the collection of propensity scores $\cf_*$ to be optimized over
takes the form $\cf_* = \cf_*(m_L,m_H;P) = \{e \in \ce \mid m_L \leq \e_P[e(X)] \leq m_H\}$,
for some \textbf{base propensity class} $\ce \subseteq \cf_0$
and some known \textbf{budget constraints} $(m_L,m_H)$ with $0 \leq m_L \leq m_H \leq 1$.
The base propensity class $\ce$ is convex and closed in $L^2(P)$
and additionally satisfies the following properties for all $n \geq 1$:
\begin{enumerate}
    \item \label{cond:finite_entropy} There exists 
    $C<\infty$ for which 
    \begin{equation}
    \label{eq:finite_entropy}
    \int_0^1 \sqrt{\log \cn(\epsilon, \ce, L^2(P_n))} \,\mrd\epsilon \leq C\quad  \forall n\quad \text{w.p.1}
    \end{equation}
    where $\log \cn(\epsilon,\ce,L^2(P_n))$ is the metric entropy of the base propensity class $\ce$ in $L^2(P_n)$,
    as defined in Appendix~\ref{app:asymptotics},
    and $P_n$ is the empirical distribution on the observations $X_1,\ldots,X_n$.
    \item \label{cond:transform_En} Given any $x_1,\ldots,x_n \in \cx$,
    there exists a convex set $E_n \subseteq [0,1]^n$, possibly depending on $x_1,\ldots,x_n$, such that 
    \begin{itemize}
        \item For every $e \in \ce$, we must have $(e(x_1),e(x_2),\ldots,e(x_n)) \in E_n$, and
        \item For every $(e_1,e_2,\ldots,e_n) \in E_n$, there exists $e \in \ce$ with $e(x_i) = e_i$ for all $i=1,\ldots,n$.
    \end{itemize}
    \item \label{cond:positivity} There exists $e_*(\cdot) \in \cf_*$ such that
    $\e_P[f(e_*(X),\eta(X))] \succeq cI$ for some $c>0$.
    \item \label{cond:hilo} There exist $e_L(\cdot),e_H(\cdot) \in \ce$ such that $\e_P[e_L(X)] > m_L$ and $\e_P[e_H(X)] < m_H$.
\end{enumerate}
\end{assumption}
\begin{remark}
If~\eqref{eq:f_strong_concavity} only holds for $(e,w)$ in $[\gamma,1-\gamma] \times \cw$ for some $\gamma \in (0,1/2)$ (instead of on all of $[0,1] \times \cw$) and the base propensity class $\ce$ is chosen to lie within $\cf_{\gamma}$,
then one can rewrite~\eqref{eq:generic_info} as
\[
\ci(e;\eta) = \e_{X \sim P}[\tilde{f}(e(X),\eta(X))], \quad \tilde{f}(e,w) = f(L_{\gamma}(e), w)
\]
where $L_{\gamma}(e)=\gamma+(1-2\gamma)e$ is an invertible linear mapping.
Then the remainder of Assumption~\ref{assump:info_matrix} holds as stated with $f$ replaced by $\tilde{f}$ and the base propensity class $\ce$ replaced by $\tilde{\ce} = \{L_{\gamma}^{-1}(e(\cdot)) \mid e(\cdot) \in \ce\}$,
and so all results below depending on Assumption~\ref{assump:info_matrix} hold by considering optimization over $\tilde{\ce}$ instead of $\ce$.
\end{remark}
Similar to the idea of empirical risk minimization in supervised learning (ERM, e.g.~\citet{vapnik1991principles} and~\citet{montanari2022universality}),
when Assumption~\ref{assump:info_matrix} is satisfied,
our learning procedure replaces the unknown population expectation appearing in the objective~\eqref{eq:e_star} by a sample average over observations $X_1,\ldots,X_n \simiid P$,
where the generic sample size $n$ diverges.
In our designed CSBAE,
these observations will be the covariates of those subjects in a given batch $t=1,\ldots,T$ of a CSBAE within a given fold $k=1,\ldots,K$,
so $n$ will be identified with the quantity $n_{t,k}$ of Definition~\ref{def:CSBAE}.
The main computational procedure is then a finite dimensional optimization over the values of the propensity score at the $n$ points $X_1,\ldots,X_n$:
\begin{equation}
\label{eq:e_hat}
(\hat{e}_1,\ldots,\hat{e}_n) \in \argmax_{(e_1,\ldots,e_n) \in F_n} \Psi\biggl(n^{-1} \sum_{i=1}^n f(e_i,\hat{\eta}(X_i))\biggr).
\end{equation}
Here $\hat{\eta}(\cdot)$ is an estimate of the nuisance function $\eta(\cdot)$ in~\eqref{eq:generic_info},
and the optimization set $F_n=F_n(m_L,m_H) \subseteq \real^n$ is defined by
\begin{equation}
\label{eq:F_n}
F_n(m_L,m_H) = \biggl\{(e_1,\ldots,e_n) \in E_n \mid m_L \leq \frac1n \sum_{i=1}^n e_i \leq m_H\biggr\} \subseteq \real^n
\end{equation}
where $(m_L,m_H)$ are budget constraints as in Assumption~\ref{assump:info_matrix}
and $E_n$ is as in the numbered condition~\ref{cond:transform_En} of that assumption.

We convert the vector $(\hat{e}_1,\ldots,\hat{e}_n)$ in~\eqref{eq:e_hat} to a propensity score $\hat{e}(\cdot)$
by taking $\hat{e}(\cdot)$ to be any member of the base propensity class $\ce$ of Assumption~\ref{assump:info_matrix}
with $\hat{e}(X_i)=\hat{e}_i$ for each $i=1,\ldots,n$.
The existence of such a propensity $\hat{e}(\cdot)$ is guaranteed by the numbered condition~\ref{cond:transform_En} in Assumption~\ref{assump:info_matrix}.
Then as in the ERM literature,
we use empirical process arguments to guarantee the learned propensity $\hat{e}(\cdot)$ converges to $e^*(\cdot)$ at rate $O_p(n^{-1/4})$
under appropriate restrictions on the complexity of the base propensity class $\ce$.
Our main such restriction is the finite entropy integral requirement in~\eqref{eq:finite_entropy}.
Examples of function classes satisfying this condition
can be found in the literature on empirical processes.
A sufficient condition is that $\log \cn(\epsilon,\ce,L^2(P_n)) \leq K\epsilon^{-2+\delta}$ for some $\delta >0$, some 
$K<\infty$,
and all $\epsilon > 0$.
Examples~\ref{example:monotone} through~\ref{example:VC_hull} below show that this requirement is loose enough to admit some relatively rich function classes.
\begin{example}[Monotone]
\label{example:monotone}
Suppose $d=1$ and let $\ce$ be the set of nondecreasing functions in $\cf_0$.
By Lemma 9.11 of~\citet{kosorok2008introduction},
we know that $\log \cn(\epsilon,\ce,L^2(P_n)) \leq K/\epsilon$
for each $\epsilon>0$ and some positive universal constant $K<\infty$.
\end{example}
\begin{example}[Lipschitz]
\label{example:lipschitz}
Again suppose $d=1$ and let $\ce$ be the set of $L$-Lipschitz functions in $\cf_0$
for some fixed $L>0$.
If $\cx$ is a bounded closed interval,
then the discussion preceding Example 5.11 of~\citet{wainwright2019high} shows that $\log \cn(\epsilon,\ce,L^2(P_n)) \leq K/\epsilon$
for each $\epsilon>0$ and some positive universal constant
$K<\infty$ (which may depend on $L$).
\end{example}
\begin{example}(VC-subgraph class)
\label{example:VC}
Let $\ce$ be any subset of $\cf_0$ that is closed and convex in $L^2(P)$,
and whose subgraphs are a Vapnik-Chervonenkis (VC) class,
meaning they have a finite VC dimension $V$.
A special case is a fully parametric class like $\{x \mapsto  \theta^{\top}\xi(x) \mid \theta^{\top}\bm{1}_p \leq 1, \theta \succeq 0\}$
where $\xi(x) \in [0,1]^p$ is a known set of basis functions
and $\bm{1}_p := (1,\ldots,1) \in \real^p$.
Then by Theorem 2.6.7 of~\citet{van1996weak},
$\cn(\epsilon,\ce,L^2(P_n)) \leq K(1/\epsilon)^{2V-2}$ for some universal $K>0$ depending on the VC dimension $V$ of the subgraphs.
Note that $K$ may depend on $p$.
\end{example}
\begin{example}(Symmetric convex hull of VC-subgraph class)
\label{example:VC_hull}
Let $\ce_0$ be a VC-subgraph class of functions.
The symmetric convex hull of $\ce_0$ is defined as 
\[
\sconv(\ce_0) = \biggl\{\,\sum_{i=1}^m \omega_i e_i \Bigm| e_i \in \ce_0, \sum_{i=1}^m |\omega_i| \leq 1 \biggr\}.
\]
Now suppose $\ce$ is contained within $\ovl{\sconv(\ce_0)}$, the pointwise closure of $\sconv(\ce_0)$.
Let $V<\infty$ be the VC dimension of the collection of subgraphs of functions in $\ce_0$.
Then by Theorem 2.6.9 of~\citet{van1996weak}
we have $\log \cn(\epsilon,\ce,L^2(P_n)) \leq K(1/\epsilon)^{2(1-1/V)}$ for all $\epsilon>0$.
For example, with $\expit(x)=\exp(x)/(1+\exp(x))$ we can take
\begin{equation}
\label{eq:E_conv_hull}
\ce = \biggl\{\,\sum_{i=1}^m \omega_i \expit(\theta_i^{\top}\varphi(x)) \Bigm| \omega_i \geq 0, \sum_{i=1}^m \omega_i \leq 1\biggr\}
\end{equation}
where $\varphi(\cdot)$ is any vector of $p$ real-valued basis functions,
$m$ can be made arbitrarily large,
and $\theta_1,\ldots,\theta_m \in \real^p$ are arbitrary.
This choice of $\ce$ is evidently a closed and convex subset of $\sconv(\ce_0)$ with $\ce_0=\{\expit(\theta^{\top}\varphi(x)) \mid \theta \in \real^p\}$.
Note the collection $\ce_0$ is indeed a VC-subgraph class by Lemmas 2.6.15 and 2.6.17 of~\citet{van1996weak},
as each function in $\ce_0$ is the composition of the monotone function $\expit(\cdot)$
with the $p$-dimensional vector space of functions $\{x \mapsto \theta^{\top}\varphi(x) \mid \theta \in \real^p\}$.
\end{example}

Next, we construct sets $E_n$ that satisfy condition~\ref{cond:transform_En} of Assumption~\ref{assump:info_matrix}
for the base propensity classes in Examples~\ref{example:monotone} to~\ref{example:VC_hull}.
For the set of monotone functions in one dimension (Example~\ref{example:monotone}) we can take
\[
E_n = \{(e_1,\ldots,e_n) \mid  0 \leq e_{\pi(1)} \leq e_{\pi(2)} \leq \cdots \leq e_{\pi(n)} \leq 1\}
\]
where $\pi(\cdot)$ is the inverse of the function that maps each $i \in \{1,\ldots,n\}$ to the rank of $x_i$ among $x_1,\ldots,x_n$ (with any ties broken in some deterministic way).
For the set of $L$-Lipschitz functions (Example~\ref{example:lipschitz}) we can take
\[
E_n = \bigl\{(e_1,\ldots,e_n) \bigm| |e_{\pi(i)}-e_{\pi(i-1)}| \leq L(x_{\pi(i)}-x_{\pi(i-1)}),\, i=2,\dots,n\bigr\}.
\]
For the parametric class in Example~\ref{example:VC} we can take
\[
E_n = \bigl\{(\theta^{\top}\xi(x_1),\ldots,\theta^{\top}\xi(x_n)) \mid \theta^{\top}\bm{1}_p \leq 1, \theta \succeq 0\bigr\}.
\]
Finally, for the class~\eqref{eq:E_conv_hull} in Example~\ref{example:VC_hull} we take
\[
E_n = \biggl\{\,\sum_{i=1}^m \omega_i\bigl(\expit(\theta_i^{\top}\varphi(x_1)),\ldots,\expit(\theta_i^{\top}\varphi(x_n))\bigr) \Bigm| \omega_i \geq 0, \sum_{i=1}^m \omega_i \leq 1 \biggr\}.
\]

The convergence rates of $\hat{e}(\cdot)$ to $e^*(\cdot)$ will be proven using strong concavity of the design objective~\eqref{eq:e_star} on the space $\cf^*$.
This is ensured by Assumption~\ref{assump:info_matrix} along with the following conditions on the information function $\Psi(\cdot)$:

\begin{assumption}[Information function regularity]
\label{assump:Psi}
The information function $\Psi:\ess_+^p \rightarrow \real \cup \{-\infty\}$ is concave, continuous, and nondecreasing with respect to the semidefinite ordering $\succeq$ on $\real^{p \times p}$ and satisfies the following conditions:
\begin{itemize}
\item[(a)] For every $k>0$, $\inf_{B \succeq kI} \Psi(B) > \sup_{A \in \ess_+^p \setminus \ess_{++}^p} \Psi(A) =: \Psi_0$.
\item[(b)] $\Psi(\cdot)$ is
twice continuously differentiable on $\ess_{++}^p$,
such that for all $0<k<K$, there exists $C>0$ such that $\|\nabla \Psi(A)-\nabla \Psi(B)\| \leq C\|A-B\|$ whenever $KI \succeq A \succeq kI$ and $KI \succeq B \succeq kI$.
\item[(c)] For every $0 < k < K$,
$kI \preceq A \preceq KI$ implies $\tilde{k}I \preceq \nabla \Psi(A) \preceq \tilde{K}I$ for some $0 < \tilde{k} < \tilde{K}$.
\item[(d)] For every $K<\infty$ and $\tilde{\Psi}_0 > \Psi_0$, there exists $k > 0$ such that for all $0 \preceq A \preceq KI$ with $\Psi(A) \geq \tilde{\Psi}_0$, we have $A \succeq kI$.
\end{itemize}
\end{assumption}

We can show that Assumption~\ref{assump:Psi} is satisfied by two common information functions:
the ``$A$-optimality" function $\Psi_a(\cdot) = -\tr((\cdot)^{-1})$ with $\Psi_a(M) :=-\infty$ whenever $M$ is singular,
and the ``$D$-optimality" function $\Psi_d(\cdot) = \log(\det(\cdot))$.
The $A$-optimality criterion corresponds to minimizing the average (asymptotic) variance of the components of the estimand,
while $D$-optimality corresponds to minimizing the volume of the ellipsoid spanned by the columns of the (asymptotic) covariance matrix.
\begin{lemma}
\label{lemma:Psi_cond}
The information functions $\Psi_d(\cdot)$ and $\Psi_a(\cdot)$ satisfy Assumption~\ref{assump:Psi}.
\end{lemma}
\begin{proof}
See Appendix~\ref{proof:lemma:Psi_cond}.
\end{proof}

There are some common information functions that do not satisfy
Assumption~\ref{assump:Psi}.
For example, the ``$E$-optimality" function $\Psi_e(\cdot) = \lambda_{\min}(\cdot)$,
where $\lambda_{\min}(M)$ refers to the smallest eigenvalue of $M \in \real^{p \times p}$,
is not differentiable.
Similarly the function $\Psi_c(\cdot) = -c^{\top}(\cdot)^{-1}c$ (for some fixed $c \in \real^p$),
corresponding to ``$c$-optimality," does not satisfy condition (c).
We leave open the question of whether the $O_p(n^{-1/4})$
convergence rate of $\hat{e}(\cdot)$ to $e^*(\cdot)$ in Lemma~\ref{lemma:concave_maximization} can be extended to these (and other) information functions using different techniques,
and now prove this rate under Assumptions~\ref{assump:info_matrix} and~\ref{assump:Psi}. 

\begin{lemma}[Convergence of generic concave maximization routine]
\label{lemma:concave_maximization}
Suppose Assumption~\ref{assump:info_matrix} holds for some
information matrix $\ci$ of the form~\eqref{eq:generic_info},
covariate distribution $P$, 
and budget constraints $(m_L,m_H)$.
Further assume that for some sequence $\alpha_n \downarrow 0$,
we have an estimate $\hat{\eta}(\cdot)$ of $\eta(\cdot)$ satisfying
\begin{equation}
\label{eq:eta_n}
\hat{\eta}(x) \in \cw,\ \forall x \in \cx, \quad\text{and}\quad \|\hat{\eta}-\eta\|_{2,P_n} = O_p(\alpha_n)
\end{equation}
where $\eta(\cdot)$ is defined by $\ci$ by~\eqref{eq:generic_info}.
Then for any information function $\Psi(\cdot)$ satisfying Assumption~\ref{assump:Psi},
the following statements are true:
\begin{enumerate}
\item There exists an optimal propensity function $e^*(\cdot)$ satisfying~\eqref{eq:e_star} which is unique $P$-almost everywhere. 
\item \label{cond:F_n} There exist optimal
finite sample treatment probabilities $(\hat{e}_1,\ldots,\hat{e}_n) \in [0,1]^n$ satisfying~\eqref{eq:e_hat},
where $F_n=F_n(m_L,m_H)$ is defined as in~\eqref{eq:F_n}.
\item For any such optimal probabilities $(\hat{e}_1,\ldots,\hat{e}_n)$,
there exists a propensity score $\hat{e}(\cdot) \in \ce$ for which $\hat{e}(X_i) = \hat{e}_i$ for each $i=1,\ldots,n$.
Any such function $\hat{e}(\cdot)$ satisfies both $\|\hat{e}-e^*\|_{2,P}=O_p(n^{-1/4} + \alpha_n)$ and $\|\hat{e}-e^*\|_{2,P_n}=O_p(n^{-1/4}+\alpha_n)$.
\end{enumerate}
\end{lemma}
\begin{proof}
See Appendix \ref{proof:lemma:concave_maximization}.
\end{proof}

\subsection{Convergence of batch adaptive designs }
\label{sec:apply_concave_maximization}
We now leverage Lemma~\ref{lemma:concave_maximization} to develop a procedure
(Algorithm~\ref{alg:csbae})
that can learn adaptive propensities $\hat{e}_t^{(k)}(\cdot)$ with the convergence guarantees~\eqref{eq:csbae_limit} so that when used for treatment assignment,
they lead to a CSBAE
with limiting propensities that optimize objectives of the form~\eqref{eq:e_star}
with information matrices based on $V_{0,\aipw}^{-1}$ and $V_{0,\epl}^{-1}$.
This shows we can effectively design for the estimators $\hat{\theta}_{\aipw}$ and $\hat{\theta}_{\epl}$.

\begin{algorithm}[tb]
\caption{CSBAE for estimating $\theta_{0,\ate}$ resp. $\theta_{0,\pl}$}
\label{alg:csbae}
\begin{algorithmic}[1]
\REQUIRE Base propensity class $\ce$ satisfying conditions in Assumption~\ref{assump:info_matrix}, initial propensity $e_1(\cdot) \in \cf_{\epsilon_1}$ for some $\epsilon_1 > 0$,
information function $\Psi(\cdot)$ satisfying Assumption~\ref{assump:Psi}, number of folds $K \geq 2$
\FOR{batch $t=1,\ldots,T$} 
\STATE Observe subject covariates $X_{t1},\ldots,X_{tN_t}$
\STATE Split subject indices $i=1,\ldots,N_t$ into $K$ folds $\ci_1,\ldots,\ci_K$, of size as equal as possible
\FOR{fold $k=1,\ldots,K$}
\STATE Label the covariates in batch $t$, fold $k$ by $X_{t1}^{(k)},\ldots,X_{tn_{t,k}}^{(k)}$
\IF{$t=1$}
\STATE Set $(\hat{e}_{t1}^{(k)},\ldots,\hat{e}_{tn_{t,k}}^{(k)})=\big(e_1(X_{t1}^{(k)}),\ldots,e_1(X_{tn_{t,k}}^{(k)})\big)$
\ELSE
\STATE Compute $(\hat{e}_{t1}^{(k)},\ldots,\hat{e}_{tn_{t,k}}^{(k)})$ using right-hand side of~\eqref{eq:e_t_hat_aipw} resp.~\eqref{eq:e_t_hat_epl} given batch $t$ budget constraints $m_{L,t} \leq m_{H,t}$ 
\ENDIF 
\STATE Draw $(U_{t1}^{(k)},\ldots,U_{tn_{t,k}}^{(k)}) \simiid \Unif(0,1)$
\STATE Assign treatments according to $Z_{ti}^{(k)} = \indic(U_{ti}^{(k)} \leq \hat{e}_{ti}^{(k)})$, $i=1,\ldots,n_{t,k}$
\STATE Choose and store any $\hat{e}_t^{(k)}(\cdot) \in \ce$ with $\hat{e}_t^{(k)}(X_{ti}^{(k)})=\hat{e}_{ti}^{(k)}, i=1,\ldots,n_{t,k}$
\ENDFOR
\STATE Observe outcomes $\{W_{ti} \mid 1 \leq i \leq N_t\}$ in batch $t$
\STATE Compute and store $N^{1/4}$-consistent estimates $\big(\hat{v}_{1:t}^{(1)}(\cdot,\cdot),\ldots,\hat{v}_{1:t}^{(K)}(\cdot,\cdot)\big)$ of $v_0(\cdot,\cdot)$ where for each fold $k=1,\ldots,K$,
$\hat{v}_{1:t}^{(k)}(\cdot,\cdot)$ depends only on the observations $\{W_{ui}:(u,i) \in \ci_k, 1 \leq u \leq t\}$ in fold $k$ and batches $1,\ldots,t$
\ENDFOR
\STATE Compute final estimator $\hat\theta$ via~\eqref{eq:theta_hat}, using $(s_{a,\aipw},s_{b,\aipw})$, resp. $(s_{a,\pl},s_{b,\pl})$.
\end{algorithmic}
\end{algorithm}

For simplicity,
we assume treatment in the first batch is assigned according to a non-random propensity
$e_1(\cdot) \in \cf_{\epsilon_1}$ for some $\epsilon_1>0$.
We let $\hat{e}_1^{(k)}(\cdot)=e_1(\cdot), k=1,\ldots,K$.
Then for later batches $t=2,\ldots,T$,
the target propensities are taken to be one of the following,
for an information function $\Psi(\cdot)$ satisfying Assumption~\ref{assump:Psi}:
\begin{equation}
\label{eq:e_t_star_aipw}
e_{t,\aipw}^*(\cdot) \in \argmax_{e_t(\cdot) \in \cf_{*,t}} \Psi(V_{0:t,\aipw}^{-1}) 
\quad\text{or}\quad 
e_{t,\epl}^*(\cdot) \in \argmax_{e_t(\cdot) \in \cf_{*,t}} \Psi(V_{0:t,\epl}^{-1}).
\end{equation}
Above,
$V_{0:t,\aipw}$ and $V_{0:t,\epl}$ are the asymptotic variances of the oracle pooled estimators $\hat{\theta}_{\aipw}$ and $\hat{\theta}_{\epl}$ of~\eqref{eq:theta_star_aipw} and~\eqref{eq:theta_star_epl},
respectively,
\emph{when computed using observations in a non-adaptive batch experiment with only $t$ batches and propensities $e_1(\cdot),\ldots,e_t(\cdot)$}.
By~\eqref{eq:v_0_aipw} we can compute
\begin{equation}
\label{eq:V_0_t_aipw}
V_{0:t,\aipw} = V_{0:t,\aipw}(e_t;\eta_{0,\aipw}) = \e_{P^X}\bigg[\frac{v_0(1,X)}{e_{0:t}(X)} + \frac{v_0(0,X)}{1-e_{0:t}(X)} + (\tau_0(X)-\theta_0)^2\bigg]
\end{equation}
where $\eta_{0,\aipw}(x)$ includes the components $(v_0(0,x),v_0(1,x),\tau_0(x),\theta_0)$.
Similarly,
by~\eqref{eq:v_0_pl} we have
\begin{equation}
\label{eq:V_0_t_pl}
V_{0:t,\epl}(e_t;\eta_{0,\epl}) = \Bigg(\e_{P^X}\bigg[\frac{e_{0:t}(X)(1-e_{0:t}(X))}{v_0(0,X)e_{0:t}(X)+v_0(1,X)(1-e_{0:t}(X))}\psi(X)\psi(X)^{\top}\bigg]\Bigg)^{-1}
\end{equation}
where $\eta_{0,\epl}(x)$ includes the components $(v_0(0,x),v_0(1,x))$.
In both of the preceding equations,
the dependence on the batch $t$ propensity score $e_t(\cdot)$ is through the mixture $e_{0:t}(\cdot)$ given by
\begin{equation}
\label{eq:e_0_x}
e_{0:t}(x) := \left(\sum_{u=1}^t \kappa_u\right)^{-1} \left(\sum_{u=1}^{t-1} \kappa_ue_u(x) + \kappa_te_t(x)\right), \quad x \in \cx.
\end{equation}
Finally, the optimization set $\cf_{*,t}$ in~\eqref{eq:e_t_star_aipw} is
\begin{equation}
\label{eq:F_t}
\cf_{*,t}=\cf_*(m_{L,t},m_{H,t};P^X)=\{e(\cdot) \in \ce \mid m_{L,t} \leq \e_{P^X}[e(X)] \leq m_{H,t}\}
\end{equation}
which satisfies all the conditions in Assumption~\ref{assump:info_matrix} with covariate distribution $P^X$,
budget constraints $(m_{L,t},m_{H,t})$,
and base propensity class $\ce$.

We do not target the final covariances $V_{0,\aipw}=V_{0:T,\aipw}$ and $V_{0,\epl}=V_{0:T,\epl}$ in our discussion here
since the sample sizes and budget constraints in future batches may not be known.
If they are known in advance, 
then we can learn propensities for all future batches simultaneously at the time batch 2 covariates are observed,
and indeed target $V_{0:T,\aipw}$ or $V_{0:T,\epl}$ at that stage.

Now suppose we split our observations into $K$ folds as in a CSBAE,
and for notational simplicity we re-index the covariates in each batch $t=1,\ldots,T$, fold $k=1,\ldots,K$ as $X_{t1}^{(k)},\ldots,X_{tn_{t,k}}^{(k)}$.
Then following~\eqref{eq:e_hat},
we can estimate $e_{t,\aipw}^*(\cdot)$ for each batch $t \geq 2$
within
each fold $k=1,\ldots,K$
by computing
\begin{equation}
\label{eq:e_t_hat_aipw}
\Big(\hat{e}_{t1,\aipw}^{(k)},\ldots,\hat{e}_{tn_{t,k},\aipw}^{(k)}\Big) \in \argmax_{(e_1,\ldots,e_{n_{t,k}}) \in F_{n_{t,k}}} \Psi\biggl(\Big(\hat{V}_{0:t,\aipw}^{(k)}\Big)^{-1}\bigg), \quad k=1,\ldots,K.
\end{equation}
Here $F_{n_{t,k}}=F_{n_{t,k}}(m_{L,t},m_{H,t})$ is defined as in~\eqref{eq:F_n},
and the estimate $\hat{V}_{0:t,\aipw}^{(k)}$ of $V_{0:t,\aipw}$ is given by 
\begin{align}
\hat{V}_{0:t,\aipw}^{(k)} & = \hat{V}_{0:t,\aipw}^{(k)}\Bigl(e_1,\ldots,e_{n_{t,k}};\hat{e}_1^{(k)}(\cdot),\ldots,\hat{e}_{t-1}^{(k)}(\cdot),\hat{v}_{1:(t-1)}^{(k)}(\cdot,\cdot)\Bigr) \nonumber \\
& = -\frac{1}{n_{t,k}} \sum_{i=1}^{n_{t,k}} \frac{\hat{v}_{1:(t-1)}^{(k)}\big(1,X_{ti}^{(k)}\big)}{\hat{e}_{(0:t)i}^{(k)}} + \frac{\hat{v}_{1:(t-1)}^{(k)}\big(0,X_{ti}^{(k)}\big)}{1-\hat{e}_{(0:t)i}^{(t)}} \label{eq:V_0_t_hat_ate}
\end{align}
where each estimate $\hat{v}_{1:(t-1)}^{(k)}(z,\cdot)$ of the variance function $v_0(z,\cdot)$
is computed using only the observations from batches $u=1,\ldots,t-1$ within fold $k$.
Note that any plug-in estimate of $\tau_0(\cdot)$ and $\theta_0$ does not affect the optimization~\eqref{eq:e_t_hat_aipw}
and so can be omitted.
The dependence of $\hat{V}_{0:t,\aipw}^{(k)}$ on the optimization variables $e_1,\ldots,e_{n_{t,k}}$ is through the mixture
quantities
\begin{equation}
\label{eq:e_0_i_hat}
\hat{e}_{(0:t)i}^{(k)} :=  \frac1{N_{1:t}} \Bigg(\,\sum_{u=1}^{t-1} N_u\hat{e}_u^{(k)}(X_{ti}^{(k)}) + N_te_i\Bigg), \quad i = 1,\ldots,n_{t,k}
\end{equation}
where for $u=1,\ldots,t-1$,
$\hat{e}_u^{(k)}(\cdot)$ is the (possibly adaptive) propensity used to assign treatment in batch $u$, fold $k$.
By comparing~\eqref{eq:e_0_i_hat} and~\eqref{eq:e_0_x},
we see that for each batch $u=1,\ldots,t$,
$N_u/N$ is being used as a plug-in estimate of $\kappa_u$
and the adaptive propensity score $\hat{e}_u^{(k)}(\cdot)$ is used as a plug-in estimate of its limit $e_u(\cdot)$.
Finally,
as in the conclusion of Lemma~\ref{lemma:concave_maximization},
the learned adaptive propensity $\hat{e}_{t,\aipw}^{(k)}(\cdot)$ to be used for treatment assignment in batch $t$, fold $k$
is taken to be any choice in the predetermined base collection $\ce$ with $\hat{e}_{t,\aipw}^{(k)}\left(X_{ti}^{(k)}\right)=\hat{e}_{ti,\aipw}^{(k)}$ for all $i=1,\ldots,n_{t,k}$.

Learning $e_{t,\epl}^*(\cdot)$ is exactly analogous;
first we compute
\begin{equation}
\label{eq:e_t_hat_epl}
\Big(\hat{e}_{t1,\epl}^{(k)},\ldots,\hat{e}_{tn_{t,k},\epl}^{(k)}\Big) \in \argmax_{(e_1,\ldots,e_{n_{t,k}}) \in F_{n_{t,k}}} \Psi\bigg(\Big(\hat{V}_{0:t,\epl}^{(k)}\Big)^{-1}\bigg)
\end{equation}
where
\begin{align}
\label{eq:V_0_t_hat_pl}
\hat{V}_{0:t,\epl}^{(k)} & =
\hat{V}_{0:t,\epl}^{(k)}\Bigl(e_1,\ldots,e_{n_{t,k}};\hat{e}_1^{(k)}(\cdot),\ldots,\hat{e}_{t-1}^{(k)}(\cdot),\hat{v}_{1:(t-1)}^{(k)}(\cdot,\cdot)\Bigr) \\
& = \frac{1}{n_{t,k}} \sum_{i=1}^{n_{t,k}}\frac{\hat{e}_{(0:t)i}^{(k)}\big(1-\hat{e}_{(0:t)i}^{(k)}\big)}{\hat{v}_{1:(t-1)}^{(k)}\big(0,X_{ti}^{(k)}\big)\hat{e}_{(0:t)i}^{(k)} + \hat{v}_{1:(t-1)}^{(k)}\big(1,X_{ti}^{(k)}\big)\big(1-\hat{e}_{(0:t)i}^{(k)}\big)}\psi\big(X_{ti}^{(k)}\big)\psi\big(X_{ti}^{(k)}\big)^{\top}.
\end{align}
Then we assign treatment with any propensity $\hat{e}_{t,\epl}^{(k)}(\cdot)$ in the base propensity class $\ce$ satisfying
$\hat{e}_{t,\epl}\big(X_{ti}^{(k)}\big)=\hat{e}_{ti,\epl}^{(k)}$ for all $i=1,\ldots,n_{t,k}$.

The main additional regularity condition required to ensure the adaptive propensities $\hat{e}_{t,\aipw}^{(k)}(\cdot)$ and $\hat{e}_{t,\epl}^{(k)}$ above converge at the desired $O_p(N^{-1/4})$ RMS rate to $e_{t,\aipw}^*$ and $e_{t,\epl}^*$ of~\eqref{eq:e_t_star_aipw}
is the same rate of convergence in the estimates $\hat{\eta}^{(k)}_{\aipw}(\cdot)$ and $\hat{\eta}^{(k)}_{\epl}(\cdot)$ of the nuisance parameters $\eta_{0,\aipw}(\cdot)$ and $\eta_{0,\epl}(\cdot)$.
We also strengthen the sample size asymptotics~\eqref{eq:prop_asymp_limit} by requiring
\begin{equation}
\label{eq:prop_asymp_conv}
\frac{N_t}{N} = \kappa_t + O(N^{-1/4}), \quad t=1,\ldots,T.
\end{equation}

\begin{theorem}[Convergence of Algorithm~\ref{alg:csbae}]
\label{thm:concave_maximization}
Suppose $T \geq 2$,
fix a batch $t \in \{2,\ldots,T\}$
and suppose Assumption~\ref{assump:DGP} holds along with~\eqref{eq:prop_asymp_conv}.
Further assume treatment in batches $1,\ldots,t-1$ is assigned according to a CSBAE
where the batch 1 propensities are $\hat{e}_1^{(1)}(\cdot)=\ldots=\hat{e}_1^{(K)}(\cdot)=e_1(\cdot) \in \cf_{\epsilon_1}$ for some $\epsilon_1 >0$,
and that for each fold $k=1,\ldots,K$:
\begin{itemize}
\item There exists an estimator $\hat{v}^{(k)}(\cdot)$ of the variance function $v_0(\cdot)$ depending only on the observations $\{W_{ui}^{(k)} \mid 1 \leq u \leq t-1\}$ in batches $1,\ldots,t-1$ assigned to fold $k$,
such that $\|\hat{v}^{(k)}(z,\cdot)-v_0(z,\cdot)\|_{2,P^X} = O_p(N^{-1/4})$ for $z=0,1$.
\item There are universal constants $0<c<C<\infty$ for which 
\[
c \leq \inf_{(z,x) \in \{0,1\} \times \cx} \min\big(\hat{v}^{(k)}(z,x),v_0(z,x)\big) \leq \sup_{(z,x) \in \{0,1\} \times \cx} \max\big(\hat{v}^{(k)}(z,x),v_0(z,x)\big) \leq C.
\]
\item The information function $\Psi(\cdot)$ satisfies Assumption~\ref{assump:Psi}.
\end{itemize}
Let $\ce \subseteq \cf_0$ be any base propensity class satisfying the conditions of Assumption~\ref{assump:info_matrix},
and define $P_{N,t}^{(k),X}$ to be the empirical distribution on the covariates $X_{t1}^{(k)},\ldots,X_{tn_{t,k}}^{(k)}$ in batch $t$, fold $k$.
Then for any budget constraints $0 \leq m_{L,t} \leq m_{H,t} \leq 1$ and each fold $k=1,\ldots,K$,
the following holds:
\begin{enumerate}
    \item (Design for $\hat{\theta}_{\aipw}$)
    There exists a target propensity $e_{t,\aipw}^*(\cdot) \in \ce$ satisfying~\eqref{eq:e_t_star_aipw}
    that is unique $P^X$-almost surely.
    Additionally, there exists a solution $(\hat{e}_{t1,\aipw}^{(k)},\ldots,\hat{e}_{tn_{t,k},\aipw}^{(k)})$ to~\eqref{eq:e_t_hat_aipw};
    any such solution has the property that any propensity $\hat{e}_t^{(k)}(\cdot) \in \ce$ with $\hat{e}_t^{(k)}(X_{ti}^{(k)})=\hat{e}_{ti,\aipw}^{(k)}$ for $i=1,\ldots,n_{t,k}$ satisfies $$\|\hat{e}_t^{(k)}-e_{t,\aipw}^*\|_{2,P^X} + \|\hat{e}_t^{(k)}-e_{t,\aipw}^*\|_{2,P_{N,t}^X} = O_p(N^{-1/4}).$$ 
\end{enumerate}
If additionally,
the linear treatment effect assumption~\eqref{eq:pl_assumption} holds for some basis function $\psi(X)$ containing an intercept,
then:
\begin{enumerate}[resume]
    \item (Design for $\hat{\theta}_{\epl}$) There exists a target propensity $e_{t,\epl}^*(\cdot) \in \ce$ satisfying~\eqref{eq:e_t_star_aipw}
    that is unique $P^X$-almost surely.
    Furthermore, 
    there exists a solution $(\hat{e}_{t1,\epl}^{(k)},\ldots,\hat{e}_{tn_{t,k},\epl}^{(k)})$ to~\eqref{eq:e_t_hat_epl};
    any such solution has the property that any propensity $\hat{e}_t^{(k)}(\cdot) \in \ce$ with $\hat{e}_t^{(k)}(X_{ti}^{(k)})=\hat{e}_{ti,\epl}^{(k)}$ for $i=1,\ldots,n_{t,k}$ satisfies $$\|\hat{e}_t^{(k)}-e_{t,\epl}^*\|_{2,P^X} + \|\hat{e}_t^{(k)}-e_{t,\epl}^*\|_{2,P_{N,t}^X} = O_p(N^{-1/4}).$$
\end{enumerate}
\end{theorem}
\begin{proof}
See Appendix \ref{proof:thm:concave_maximization}.
\end{proof}

\section{Numerical simulations}
\label{sec:simulations}
We implement Algorithm~\ref{alg:csbae} to construct some synthetic CSBAE's that illustrate the finite sample performance of our proposed methods. 
For simplicity we consider $T=2$ batches throughout. 
Our evaluation metric is
the average mean squared error (AMSE) of the estimators $\hat{\theta}_{\aipw}$ and $\hat{\theta}_{\epl}$
computed at the end of each CSBAE.
As a baseline,
we also compute feasible variants of the linearly aggregated estimators $\hat{\theta}_{\aipw}^{(\la)}$ and $\hat{\theta}_{\epl}^{(\la)}$ computed on a ``simple RCT": 
that is,
a non-adaptive batch experiment with a constant propensity score in each batch.
We additionally consider the approach of~\citet{hahn2011adaptive} for design and estimation of $\theta_{0,\ate}$.
This is equivalent to using Algorithm~\ref{alg:csbae} for design (without sample splitting, i.e.\ $K=1$)
and using the pooled $\hat{\theta}_{\aipw}$ as the final estimator,
but with the covariates $X$ replaced everywhere by a coarse discretization $S=S(X)$.
As a hybrid we also consider using the discretized covariates $S$ for design but computing the final estimates $\hat{\theta}_{\aipw}$
and $\hat{\theta}_{\epl}$
using the full original covariate $X$.
To separately attribute efficiency gains to design and pooling,
we also consider the linearly aggregated estimators $\hat{\theta}_{\aipw}^{(\la)}$ and $\hat{\theta}_{\epl}^{(\la)}$ when a modification of Algorithm~\ref{alg:csbae} that targets $V_{2,\aipw}$ and $V_{2,\epl}$
(the asymptotic variances of the estimators $\hat{\theta}_{2,\aipw}$ and $\hat{\theta}_{2,\epl}$ given in Section~\ref{sec:pooled_estimation},
depending only on observations in batch 2)
is used for design.

We consider four data generating processes (DGPs),
distinguished by whether the covariate dimension $d$ is 1 or 10 and whether the conditional variance functions $v_0(\cdot,\cdot)$ are homoskedastic (with $v_0(0,x)=v_0(1,x)=1$ for all $x \in \cx$) or heteroskedastic (with $v_0(0,x)=v_0(1,x)/2=\exp((\bm{1}_d^{\top}x)/(2\sqrt{d}))$.
The scaling by the covariate dimension $d$ in the heteroskedastic variance functions ensures the variance in $v_0(z,X)$ is independent of the covariate dimension $d$. 
In all of the DGPs the covariates are i.i.d. spherical Gaussian,
i.e.\ $P^X=\mathcal{N}(0,I_d)$.
The outcome mean functions are taken to be $m_0(0,x)=m_0(1,x)=\bm{1}_d^{\top}x$ for $\bm{1}_d=(1,\ldots,1)' \in \real^d$.
For estimating $\theta_{0,\pl}$
we use the basis functions $\psi(x)=(1,x^{\top})^{\top} \in \real^p$ where $p=d+1$.
Note $\theta_{0,\ate}=\theta_{0,\pl}=0$. 
The potential outcomes $Y(0),Y(1)$ are generated as follows:
\begin{align*}
(\epsilon(0),\epsilon(1)) \mid X & \sim \mathcal{N}((0,0)^\top,\diag(v_0(0,X),v_0(1,X))) \\
Y(z) & = m_0(z,X) + \epsilon(z), \quad z=0,1.
\end{align*}

For each DGP we run Algorithm~\ref{alg:csbae} with $K=2$ folds,
batch sample sizes $N_1=N_2=1000$,
information function $\Psi=\Psi_a(\cdot)$ (corresponding to $A$-optimality),
and treatment fraction constraints $m_{L,t}=m_{H,t}=0.2$ for $t=1,2$.
In Appendix~\ref{app:simulations},
we present additional simulation results where $m_{L,1}=m_{H,1}$ remain at 0.2 but $m_{L,2}=m_{H,2}=0.4$,
so that the treatment budget for the second batch has increased.
The initial propensity score $e_1(\cdot)$ is taken to be constant (i.e. $e_1(x)=0.2$ for all $x \in \cx$),
and the base propensity class $\ce$ is the set of all $1$-Lipschitz functions taking values in $[0,1]$ when $d=1$ (cf. Example~\ref{example:lipschitz}).
When $d=10$,
we take $\ce$ as in~\eqref{eq:E_conv_hull},
with $\{\theta_1,\ldots,\theta_m\}$ the collection of vectors $(a_1,\ldots,a_{11})' \in \real^{11}$ with each coordinate $a_i \in \{-2,-1,0,1,2\}$ and no more than two of the $a_i$'s nonzero.

The discretization used to implement the approach of~\citet{hahn2011adaptive} partitions $\real^d$ into four bins based on the quartiles of $\bm{1}_d^{\top}X$.
Note that this partition is along
the single dimension along which the variance functions $v_0(\cdot,\cdot)$ vary in the heteroskedastic DGPs,
so we would expect this to perform better than in practice,
where the structure of the variance functions is not known 
(it could possibly be learned,
as in~\citet{tabord-meehan2022stratification}).
 The choice of four bins is based on the experiments of~\citet{hahn2011adaptive},
which find minimal performance difference between two and six bins for the DGP's they consider.
We denote their ``binned" AIPW estimator by $\hat{\theta}_{\aipw}^{(\text{bin})}$.
Recall,
as indicated above,
that this is equivalent to $\hat{\theta}_{\aipw}$ when the only available covariate is the binned $S(X)$ and no cross-fitting is used.
For $\hat{\theta}_{\aipw}^{(\text{bin})}$ the conditional means $\tilde{m}_0(z,s) = \e(Y(z) \mid S=s), z=0,1$ are estimated nonparametrically by sample outcome means among all units with $Z=z$ and $S=s$ in the appropriate batch(es) and fold(s).
Similarly, the conditional variances $\tilde{v}_0(z,s)=\Var(Y(z) \mid S=s)$ are estimated by sample outcome variances.

For all simulations,
the concave maximizations are performed using the CVXR software~\citep{fu2017cvxr} and the MOSEK solver~\citep{mosek}.
All estimates $\hat{m}(z,\cdot)$ of the mean function $m_0(z,\cdot), z=0,1$ are computed 
by fitting a generalized additive model (GAM) to  the outcomes $Y$ and covariates $X$ from the observations with treatment indicators equal to $z$ in the appropriate fold(s) and batch(es).
The GAMs use a thin-plate regression spline basis~\citep{wood2003thin-plate} 
and the degrees of freedom are chosen using the generalized cross-validation procedure implemented in the \texttt{mgcv} package in R~\citep{wood2004stable}.
Variance function estimates $\hat{v}(z,\cdot)$
are computed by first computing $\hat{m}(z,\cdot)$ as above on the appropriate observations,
then fitting a GAM on these same observations to predict the squared residuals $(Y-\hat{m}(z,X))^2$ from $X$.

\subsection{Average treatment effect}
\begin{table}[!htb]
\centering
\caption{Simulated and asymptotic relative efficiencies (as defined in Section~\ref{sec:simulations} of the various design and estimation approaches for $\theta_{0,\ate}$ that we study numerically under each of the four DGP's described in the text.
}
\label{table:ate_sim}
\begin{tabular}{ccccc}
\toprule
DGP & Estimator & Design &  Sim. rel. eff. (90\% CI) & Asymp. rel. eff. \\
\midrule
\multirow{5}{0.15\linewidth}{\centering $d=1$, Homoskedastic} & $\hat{\theta}_{\aipw}$ & Flexible & 0.989 (0.965, 1.013) & 1.000 \\
& $\hat{\theta}_{\aipw}$ & Binned & 1.011 (0.977, 1.046) & 0.999 \\
& $\hat{\theta}_{\aipw}$ & Simple RCT & 1.016 (1.003, 1.029) & 1.000 \\
& $\hat{\theta}_{\aipw}^{(\la)}$ & Flexible & 0.984 (0.971, 0.998) & 0.999 \\
& $\hat{\theta}_{\aipw}^{(\text{bin})}$ & Binned & 0.919 (0.877, 0.961) & 0.885 \\
\midrule
\multirow{5}{0.15\linewidth}{\centering $d=1$, Heteroskedastic} & $\hat{\theta}_{\aipw}$ & Flexible & 1.016 (0.968, 1.064) & 1.048 \\
& $\hat{\theta}_{\aipw}$ & Binned & 1.046 (0.997, 1.096) & 1.041 \\
& $\hat{\theta}_{\aipw}$ & Simple RCT & 0.997 (0.979, 1.016) & 1.000 \\
& $\hat{\theta}_{\aipw}^{(\la)}$ & Flexible & 1.008 (0.971, 1.045) & 1.024 \\
& $\hat{\theta}_{\aipw}^{(\text{bin})}$ & Binned & 1.001 (0.946, 1.057) & 0.963 \\
\midrule
\multirow{5}{0.15\linewidth}{\centering $d=10$, Homoskedastic} & $\hat{\theta}_{\aipw}$ & Flexible & 1.039 (0.990, 1.089) & 1.000 \\
& $\hat{\theta}_{\aipw}$ & Binned & 1.011 (0.959, 1.065) & 1.000 \\
& $\hat{\theta}_{\aipw}$ & Simple RCT & 1.081 (1.052, 1.110) & 1.000 \\
& $\hat{\theta}_{\aipw}^{(\la)}$ & Flexible & 1.047 (1.018, 1.077) & 1.000 \\
& $\hat{\theta}_{\aipw}^{(\text{bin})}$ & Binned & 0.477 (0.439, 0.519) & 0.417 \\
\midrule
\multirow{5}{0.15\linewidth}{\centering $d=10$, Heteroskedastic} & $\hat{\theta}_{\aipw}$ & Flexible & 1.081 (1.023, 1.143) & 1.036 \\
& $\hat{\theta}_{\aipw}$ & Binned & 1.050 (0.986, 1.117) & 1.000 \\
& $\hat{\theta}_{\aipw}$ & Simple RCT & 1.102 (1.070, 1.135) & 1.000 \\
& $\hat{\theta}_{\aipw}^{(\la)}$ & Flexible & 0.991 (0.952, 1.030) & 1.024 \\
& $\hat{\theta}_{\aipw}^{(\text{bin})}$ & Binned & 0.651 (0.597, 0.708) & 0.595 \\
\bottomrule
\end{tabular}
\end{table}

Table~\ref{table:ate_sim} shows the performance of the various design and estimation procedures for $\theta_{0,\ate}$ that we consider.
Each entry is a ``relative efficiency": 
that is, a ratio of the MSE of the baseline approach
(which computes the linearly aggregated $\hat{\theta}_{\aipw}^{(\la)}$ on a simple RCT)
to the MSE of the relevant approach.
The simulated relative efficiencies in the table estimate these MSE's by averaging the squared error of each estimator over 1,000 simulations.
The 90\% confidence intervals for the true finite sample relative efficiency are computed using 10,000 bootstrap replications of these 1,000 simulations.
Finally, the asymptotic relative efficiencies in Table~\ref{table:ate_sim} are computed by estimating the asymptotic variance of each estimator using the appropriate formula,
i.e. $V_{0,\aipw}$ for $\hat{\theta}_{\aipw}$
and $V_{\aipw}^{(\la)}$ for $\hat{\theta}_{\aipw}^{(\la)}$.
The computation is based on a non-adaptive batch experiment with second batch propensity $e_2(\cdot)$ equal to the average of the learned propensities from the relevant approach across the 1,000 simulations,
as a closed form solution for the limiting $e_2(\cdot)$ is not easily obtained in general.
All expectations over the covariate distribution are computed using Monte Carlo integration.

For both of the homoskedastic DGPs,
it is straightforward to show using Jensen's inequality that $e_2^*(x)=0.2$ for all $x$.
It can be further shown that the unpooled and pooled estimators are asymptotically equivalent.
Thus,
for the homoskedastic DGPs,
there is no asymptotic efficiency gain to be had.
With unequal budget constraints
there is some asymptotic benefit to pooling with homoskedastic variance functions (Appendix~\ref{app:simulations_unequal});
however the optimal design will still be the simple RCT.
Nonetheless,
we do observe some finite sample benefits to pooling in Table~\ref{table:ate_sim}.
For example,
the simulated relative efficiency of $\hat{\theta}_{\aipw}$ on the simple RCT
is significantly larger than 1 for both $d=1$ and $d=10$.
We attribute this to improved nuisance estimates for the pooled estimator,
as discussed further in Appendix~\ref{app:oracle_sim}.
As one might expect,
this finite sample improvement from pooling is apparently offset by variance in both the flexible and binned design procedures.
Still,
in the homoskedastic DGPs,
our adaptive approaches do not show significant finite sample performance decline relative to the baseline,
which here is an oracle.
With unequal treatment constraints,
we further obtain asymptotic efficiency gains from pooling (Appendix~\ref{app:simulations_unequal}).

We notice that using the discretized covariate $S(X)$ in place of $X$ for both design and estimation,
as in~\citet{hahn2011adaptive},
leads to a substantial loss of efficiency.
Indeed, when $d=10$ the (asymptotic and simulated) variance of the estimator $\hat{\theta}_{\aipw}^{(\text{bin})}$ is more than double that of our baseline under the homoskedastic DGP,
both asymptotically and in our finite sample simulations.
This efficiency loss occurs because the discretized $S(X)$ explains much less of the variation in the potential outcomes $Y(z)$ than the original $X$.
We expect greater precision losses from this discretization at the estimation stage when the variance functions $v_0(z,\cdot)$ vary substantially within the strata defined by $S(X)$.

For the heteroskedastic DGPs,
we see modest asymptotic efficiency gains from both pooling and design.
Design using the flexible base propensity class leads to about a 2.4\% asymptotic efficiency gain for both $d=1$ and $d=10$,
while pooling provides an additional 1--2\% gain.
These small asymptotic gains appear to be largely canceled out at our sample sizes by the finite sample variability in learning propensity scores,
limiting the net finite sample gains from design.
Of course, with greater heteroskedasticity and/or differences between $v_0(0,\cdot)$ and $v_0(1,\cdot)$,
we would expect greater efficiency gains from design;
in our simulations we have chosen to keep these differences within common ranges in social science studies as per~\citet{blackwell2022batch}.

\subsection{Partially linear model}
Unlike for estimating $\theta_{0,\ate}$,
for estimating $\theta_{0,\pl}$,
we see clear efficiency gains over the baseline from \emph{design} when $d=1$ (Table~\ref{table:epl_sim}).
For instance,
the linearly aggregated estimator $\hat{\theta}_{\epl}^{(\la)}$ exhibits a 5.6\% asymptotic efficiency gain as a result of the flexible design;
replacing this with the pooled estimator $\hat{\theta}_{\epl}$
then yields a total asymptotic gain of 10.0\% over the baseline,
even in the homoskedastic DGP.
The analogous gains for the heteroskedastic DGP are slightly larger.
We once again observe a substantial finite sample benefit to pooling,
with the simulated relative efficiency of the approaches using the pooled $\hat{\theta}_{\epl}$ tending to be larger than the asymptotic relative efficiency.
We attribute this to both improved use of nuisance estimates by the pooled estimator 
(as in the ATE case)
as well as a more fundamental finite sample efficiency boost due to the fact that the asymptotic variance $V_{0,\epl}$ is not exact for the oracle pooled $\hat{\theta}_{\epl}^*$ in finite samples,
whereas $V_{0,\aipw}$ is exact for $\hat{\theta}_{\aipw}^*$;
see Appendix~\ref{app:oracle_sim}.
When $d=10$,
design introduces some more salient finite sample variance from the errors in the concave maximization procedure.
This is offset in both the homoskedastic and heteroskedastic DGPs by the asymptotic gains from the flexible design,
so that when $\hat{\theta}_{\epl}$ is used,
the flexible design ultimately performs similarly to the simple RCT in finite samples.

Even if we use $\hat{\theta}_{\epl}$ and hence the original covariate(s) $X$ in the final estimation step,
we see that the ``binned" design which uses only $S(X)$ for choosing the second batch propensity score
struggles to learn a substantially better propensity score than the simple RCT in all DGPs.
Indeed, in the heteroskedastic DGP with $d=1$,
we see a 2.1\% asymptotic efficiency \emph{loss} relative to the baseline from using the binned design.
By comparison, there is an 11.2\% asymptotic efficiency gain from using the flexible design.
The efficiency loss can occur with the binned design because the objective in~\eqref{eq:e_star} changes when working in terms of the discretized covariate $S(X)$ instead of $X$.
In other words, even though the simple RCT propensity $e_2(x)=0.2$ is within the class of propensities that can be chosen by the binned design,
it is worse according to the binned objective based on $S(X)$,
but not according to the objective based on the original $X$. \\

\begin{table}[!htb]
\centering
\caption{Same as Table~\ref{table:ate_sim},
but for estimating $\theta_{0,\epl}$.
\\
}
\label{table:epl_sim}
\begin{tabular}{ccccc}
\toprule
DGP & Estimator & Design &  Sim. rel. eff. (90\% CI) & Asymp. rel. eff. \\
\midrule
\multirow{4}{0.15\linewidth}{\centering $d=1$, Homoskedastic} & $\hat{\theta}_{\epl}$ & Flexible & 1.150 (1.103, 1.198) & 1.100 \\
& $\hat{\theta}_{\epl}$ & Binned & 1.057 (1.009, 1.107) & 1.026 \\
& $\hat{\theta}_{\epl}$ & Simple RCT & 1.038 (1.018, 1.058) & 1.000 \\
& $\hat{\theta}_{\epl}^{(\la)}$ & Flexible & 1.049 (1.016, 1.083) & 1.056 \\
\midrule
\multirow{4}{0.15\linewidth}{\centering $d=1$, Heteroskedastic} & $\hat{\theta}_{\epl}$ & Flexible & 1.310 (1.214, 1.412) & 1.112 \\
& $\hat{\theta}_{\epl}$ & Binned & 1.049 (0.972, 1.133) & 0.979 \\
& $\hat{\theta}_{\epl}$ & Simple RCT & 1.084 (1.014, 1.159) & 1.000 \\
& $\hat{\theta}_{\epl}^{(\la)}$ & Flexible & 0.983 (0.880, 1.076) & 1.073 \\
\midrule
\multirow{4}{0.15\linewidth}{\centering $d=10$, Homoskedastic} & $\hat{\theta}_{\epl}$ & Flexible & 1.229 (1.203, 1.256) & 1.023 \\
& $\hat{\theta}_{\epl}$ & Binned & 1.177 (1.147, 1.206) & 1.000 \\
& $\hat{\theta}_{\epl}$ & Simple RCT & 1.228 (1.204, 1.251) & 1.000 \\
& $\hat{\theta}_{\epl}^{(\la)}$ & Flexible & 1.023 (1.011, 1.035) & 1.019 \\
\midrule
\multirow{4}{0.15\linewidth}{\centering $d=10$, Heteroskedastic} & $\hat{\theta}_{\epl}$ & Flexible & 0.997 (0.963, 1.032) & 1.071 \\
& $\hat{\theta}_{\epl}$ & Binned & 0.932 (0.899, 0.965) & 1.002 \\
& $\hat{\theta}_{\epl}$ & Simple RCT & 0.982 (0.952, 1.014) & 1.000 \\
& $\hat{\theta}_{\epl}^{(\la)}$ & Flexible & 0.969 (0.942, 0.996) & 1.060 \\
\bottomrule
\end{tabular}
\end{table}
\section{Discussion}
\label{sec:discussion}
We view our primary technical contribution in this paper to be a careful extension of the double machine learning framework that enables both estimation and design in batched experiments based on pooled treatment effect estimators.
This allows the investigator to take advantage of the efficiency gains from pooling and design without needing to make strong parametric assumptions or to discretize their covariates.
As our numerical study in Section~\ref{sec:simulations} shows,
the latter can more than wipe out any efficiency gains from design.

Related to our work is the extensive literature on combining observational data with a (single batch) randomized experiment.
In that setting a primary concern is mitigating bias from unobserved confounders in the observational data~\citep{rosenman2021designing,gagnon2023precise}.
By contrast,
in our setting unconfoundedness holds by design in each batch of the experiment. 
It would be useful to examine if the ideas from the present work can be extended to the observational setting where confounding bias is a concern.  \\

\noindent
\textbf{Acknowledgments:} The authors thank Stefan Wager, Lihua Lei, and Kevin Guo for comments that improved the content of this paper. 
{H.L. was partially supported by the Stanford Interdisciplinary Graduate Fellowship (SIGF).
This work was also supported by the NSF under grant DMS-2152780.

\appendix
\section{Nonstationary batches}
\label{app:nonstationary}
Assumption~\ref{assump:DGP} in the main text
supposes the distribution $P^S$ of the covariates and potential outcomes $S_{ti}=(X_{ti},Y_{ti}(0),Y_{ti}(1))$ is stationary across batches $t=1,\ldots,T$.
Here we relax that assumption.
Let $P^S_t$ be
the distribution of the vector $S_{ti}$ in batch $t=1,\ldots,T$,
now allowed to vary across batches
(in the main text it is assumed that $P^S_1=\dots=P^S_T=P^S$).
Then we have the following relaxation of Assumption~\ref{assump:DGP}:
\begin{assumption}[Relaxation of Assumption~\ref{assump:DGP}]
\label{assump:gen_DGP}
For some fixed number of batches $T \geq 2$,
the vectors
\[
S_{ti}=(X_{ti},Y_{ti}(0),Y_{ti}(1)), \quad 1 \leq t \leq T, \quad 1 \leq i \leq N_t
\]
are mutually independent such that for each batch $t=1,\ldots,T$,
we have $S_{ti} \sim P_t^S$.
Furthermore, the sample sizes $N_t$ satisfy~\eqref{eq:prop_asymp_limit},
and the vector $W_{ti}=(X_{ti},Z_{ti},Y_{ti})$ is observed
where the outcomes $Y_{ti}$ satisfy the SUTVA assumption~\eqref{eq:sutva}.
\end{assumption}
Now letting $P_0^S$ be the mixture distribution $\sum_{t=1}^T \kappa_t P_t^S$,
we introduce the notation
$\e_{t,e}[f(W)]$,
which denotes an expectation under the distribution $P_{t,e}=P_{t,e}^W$ on $W=(X,Z,Y)$ induced by $S=(X,Y(0),Y(1)) \sim P_t^S$ and $Z \mid X \sim \Bern(e(X))$
for any propensity $e(\cdot)$ and $t=0,1,\ldots,T$. 
The notation $P_t^X$ refers to the corresponding marginal distribution of the covariates $X$.
Then the score equation~\eqref{eq:score} will be generalized to
\begin{equation}
\label{eq:gen_score}
\e_{t,e}[s(W;\theta_0,\nu_0,e')]=0, \quad \forall e,e' \in \cf_{\gamma},\quad t=1,\ldots,T,
\end{equation}
which we will require to identify $\theta_0$ in each batch:
\begin{assumption}[Relaxation of Assumption~\ref{assump:identification}]
\label{assump:gen_identification}
The estimand $\theta_0 \in \real^p$ of interest satisfies~\eqref{eq:gen_score} for some $\gamma \in [0,1/2)$,
some nuisance parameters $\nu_0$ lying in a known convex set $\cn$,
and some score $s(\cdot)$ satisfying~\eqref{eq:linear_score}.
\end{assumption}
Equation~\eqref{eq:gen_score} encodes a requirement that the same parameters $\theta_0$ and $\nu_0$ satisfy the score equations for all batches $t=1,\ldots,T$;
to ensure this,
we require those parameters to be stationary across batches.
For instance, for ATE estimation with the score $s_{\aipw}(\cdot)$,
we require the conditional mean functions $\e[Y(z) \mid X=x]=m_0(z,x)$ to be stationary across batches.
For estimation under the partially linear model with the score $s_{\epl}$,
we also require the outcome variance functions $\Var_t(Y(z) \mid X=x)=v_0(z,x)$ to remain stationary.
However, in both cases the covariate distribution can otherwise vary arbitrarily across batches,
as can higher moments of the conditional distributions of the potential outcomes $Y(z)$ given the covariates $X$.

Due to the possibility of covariate shift,
the relevant mixture propensity scores are now
\begin{equation}
\label{eq:mixture_propensity_with_shift}
e_{0,N}(x) = \sum_{t=1}^T \frac{N_t}{N}e_t(x)\frac{\mrd P_t^X}{\mrd P_0^X}(x)\quad\text{and}\quad e_0(x) = \sum_{t=1}^T \kappa_t e_t(x)\frac{\mrd P_t^X}{\mrd P_0^X}(x).
\end{equation}
These definitions generalize~\eqref{eq:e_0_N} and~\eqref{eq:e_0}.
Here $P_t^X$ denotes the marginal distribution of $X$ when $S \sim P_t^S$.
When there is no covariate shift,
we have $\mrd P_t^X/\mrd P_0^X(x)=1$ for all $x$ and we recover~\eqref{eq:e_0_N} and~\eqref{eq:e_0} in the main text.
The expressions in~\eqref{eq:mixture_propensity_with_shift} are derived using Bayes' rule
as the conditional probability that $Z=1$ given $X=x$ when $(X,Z,Y)$ is drawn uniformly at random from the pooled collection of observations $\{W_{ti} \mid 1 \leq t \leq T, 1 \leq i \leq N_t\}$ in a non-adaptive batch experiment
with propensities $e_1(\cdot),\ldots,e_T(\cdot)$.

Where indicated, we are able to generalize various results in Sections~\ref{sec:pooled_estimation} and~\ref{sec:batch_clt}.
The generalized results are as stated in the main text,
if we make the following changes to the notation and assumptions:
\begin{enumerate}
    \item Assumptions~\ref{assump:DGP} and~\ref{assump:identification} are replaced by Assumptions~\ref{assump:gen_DGP}. and~\ref{assump:gen_identification}, respectively
    \item Any references to the mixture propensities in $e_{0,N}(\cdot)$ and $e_0(\cdot)$ correspond to the more general definitions in~\eqref{eq:mixture_propensity_with_shift}, rather than~\eqref{eq:e_0_N} and~\eqref{eq:e_0}.
    \item Any expectations of the form $\e[f(W)]$ without subscripts are interpreted as being taken under the distribution $P_{0,e_0}$ on $W$.
    \item Any expectations of the form $\e_t[f(W)], t=0,1,\ldots,T$ are interpreted as being taken under the distribution $P_{t,e_t}$ on $W$,
    and any expectations of the form $\e_{0,N}[f(W)]$ are interpreted as being taken under the distribution $P_{0,e_{0,N}}$ on $W$.
\end{enumerate}

\section{Technical lemmas} 

Here we give some technical lemmas used in
our proofs.

\subsection{Asymptotics}
\label{app:asymptotics}
\begin{definition}
\label{def:oh_pee}
For any sequence of random vectors $\{X_n:n \geq 1\}$ and constants $a_n \downarrow 0$,
we write $X_n = O_p(a_n)$ if $\lim_{M \rightarrow \infty} \limsup_{n \rightarrow \infty} \Pr(\|X_n\| > Ma_n)=0$.
We write $X_n=o_p(a_n)$ if for every $M > 0$, $\limsup_{n \rightarrow \infty} \Pr(\|X_n\| > Ma_n) = 0$.
\end{definition}
\begin{lemma}
\label{lemma:conditional_oh_pee}
Let $X_n$ be a sequence of random vectors and $\{\cf_n, n \geq 1\}$ be a sequence of $\sigma$-algebras such that $\e\bigl[\|X_n\| \mid \cf_n\bigr] = o_p(1)$.
Then $X_n = o_p(1)$.
\end{lemma}
\begin{proof}[Proof of Lemma~\ref{lemma:conditional_oh_pee}]
Fixing $M>0$, we have $M\indic(\|X_n\| > M) \leq \|X_n\|$ for all $n$.
Taking conditional expectations given $\cf_n$ on both sides we have
\begin{equation}
\label{eq:cond_markov}
P\bigl(\|X_n\| > M \mid \cf_n\bigr) \leq M^{-1}\e\bigl[\|X_n\| \mid \cf_n\bigr].
\end{equation}
Thus if $\e\bigl[\|X_n\| \mid \cf_n\bigr] = o_p(1)$ we have $\Pr(\|X_n\|>M \mid \cf_n) = o_p(1)$ as well.
But $\Pr(\|X_n\|>M \mid \cf_n)$ is uniformly bounded so its expectation converges to zero, i.e., $\Pr(\|X_n\|>M)=o(1)$. 
Since $M>0$ was arbitrary we conclude that $X_n=o_p(1)$.
\end{proof}
\begin{lemma}
\label{lemma:oh_pee}
$X_n = O_p(a_n)$ if and only if for every sequence $b_n \uparrow \infty$ we have $\Pr(\|X_n\| > b_na_n) \rightarrow 0$ as $n \rightarrow \infty$.
\end{lemma}
\begin{proof}[Proof of Lemma~\ref{lemma:oh_pee}]
Fix $b_n \uparrow \infty$ and $\epsilon>0$. If $X_n = O_p(a_n)$ then there exists $M < \infty$ such that $\limsup_{n \rightarrow \infty} \Pr(\|X_n\| > Ma_n) < \epsilon$.
Since $b_n > M$ eventually we conclude $\limsup_{n \rightarrow \infty} \Pr(\|X_n\| > b_na_n) < \epsilon$ as well.
With $\epsilon>0$ arbitrary, the result follows. 
Conversely now suppose we do not have $X_n = O_p(a_n)$.
Then there exists $\epsilon>0$ such that $\limsup_{n \rightarrow \infty} \Pr(\|X_n\| > M a_n) \geq \epsilon$ for all $M < \infty$.
Defining $n_0=1$,
this ensures that for each $k=1,2,\ldots$,
there exists $n_k > n_{k-1}$ so that $\Pr(\|X_{n_k}\| > ka_n) \geq \epsilon$.
But then for $b_n = \max\{k \geq 0: n_k \leq n\}$,
we have $b_n \uparrow \infty$ yet $\Pr(\|X_n\| > b_na_n) \geq \epsilon$ for all $n \in n_1,n_2,\ldots$
so $\Pr(\|X_n\| > b_na_n)$ does not converge to 0 as $n \rightarrow \infty$.
\end{proof}
We now state an important result in empirical process theory in our proof of Lemma~\ref{lemma:concave_maximization} above.
For a metric space $(\cm,d)$,
let $B_{\epsilon}(m_0) = \{m \in \cm \mid d(m,m_0) \leq \epsilon\}$ be the $\epsilon$-ball around $m_0 \in \cm$.
For any set $\cs \subseteq \cm$,
the $\epsilon$ covering number $\cn(\epsilon,\cs,d)$ of $\cs$ is then defined as the smallest number of $\epsilon$-balls in $\cm$ whose union contains $\cs$.
For a subset $\cf$ of the space $L^2(\cx;P)$ of $P$-square integrable real-valued functions on $\cx$ with $\sup_{f \in \cf} \|f\|_{\infty} \leq \bar{F} < \infty$,
control over the logarithm of the covering numbers of $\cf$ 
(the \emph{metric entropy})
over a variety of radii $\epsilon$ under the random metric $L^2(P_n)$ given by $L^2(P_n)(f_1,f_2)=\|f_1-f_2\|_{2,P_n} = (\int (f_1(x)-f_2(x))^2 \mrd P_n(x))^{1/2}$
implies control of the empirical process $\sup_{f \in \cf}|(P_n-P)f|$
where $Qf : =\int f(x) \mrd Q(x)$ for any measure $Q$ on $\cx$.
Here $P_n$ is the empirical probability measure on observations $X_1,\ldots,X_n \simiid P$.
This result is due to a ``chaining" argument of~\citet{dudley1967sizes}.
We restate a more direct version of this result below,
which is Lemma A.4 of~\citet{kitagawa2018should}.
\begin{lemma}
\label{lemma:empirical_process}
For a class $\cf$ of $P$-measurable functions $f:\cx \rightarrow \real$ with $\sup_{f \in \cf} \|f\|_{\infty} \leq \bar{F}$,
there exists a universal constant $K<\infty$ such that for $P_n$ the empirical distribution of $X_1,\ldots,X_n \simiid P$,
\[
\sup_{f \in \cf} \big|(P_n-P)f\big| \leq K\bar{F}n^{-1/2} \int_0^1 \sqrt{\log \cn(\epsilon,\cf,L^2(P_n))} \,\mrd\epsilon.
\]
\end{lemma}
\begin{remark}
Often, control of the right-hand side in the previous display is shown by controlling 
\[
\sup_Q \int_0^1 \sqrt{\log \cn(\epsilon,\cf,L^2(Q)} \,\mrd\epsilon
\]
where the supremum is taken over all finitely supported probability measures $Q$.
See Sections 2.5 and 2.6 of~\citet{van1996weak} for further discussion. 
\end{remark}
We also make use of the following elementary results on covering numbers.
\begin{lemma}
\label{lemma:covering_bracketing}
Let $\ce$ be a collection of functions contained in the class $\cf_0$.
Let $\cg = \{x \mapsto g(e(x),\eta(x)) \mid e \in \ce\}$
for some $\eta:\cx \rightarrow \cw$ and $g:[0,1] \times \cw \rightarrow \real$ with $g(\cdot,w)$ continuous on $[0,1]$ and $\sup_{k \in [0,1], w \in \cw} |g'(k,w)| \leq C$
for some $C < \infty$,
where $g'(\cdot,\cdot)$ denotes the partial derivative of $g(\cdot,\cdot)$ with respect to the first argument.
Then for all probability distributions $P$ on $\cx$
and $\epsilon>0$,
we have $\cn(C\epsilon,\cg,L^2(P)) \leq \cn(\epsilon,\ce,L^2(P))$.
\end{lemma}
\begin{proof}[Proof of Lemma~\ref{lemma:covering_bracketing}]
Fix $\epsilon>0$ and a probability distribution $P$ on $\cx$.
For each $e \in \ce$ let $h^{(e)}(x) = g(e(x),\eta(x))$ for each $x \in \cx$,
so that $h^{(e)} \in \cg$.
Suppose $\{e_1,\ldots,e_N\}$ is an $\epsilon$-cover of $\ce$ in the $L^2(P)$ norm.
WLOG we can assume that each
$e_k$ is a member of $\cf_{\gamma}$.
Then for each $k=1,\ldots,N$,
by the uniform bound on $g'$ we have
\[
\|h^{(e)}-h^{(e_k)}\|_{2,P}^2 = \e_P\bigl[|g(e(X),\eta(X))-g(e_k(X),\eta(X))|^2\bigr]  \leq C^2 \|e-e_k\|_{2,P}^2
\]
so that $\{g^{(e_1)},\ldots,g^{(e_N)}\}$ is a $C\epsilon$ cover of $\cg$ in the $L^2(P)$ norm.
\end{proof}
\begin{lemma}
\label{lemma:sq_bracket}
Let $\ce$ be a collection of functions contained in the class $\cf_0$. 
Define $\ce_2^- = \{(f-g)^2: f \in \ce, g \in \ce\}$.
Then for every $\epsilon>0$ and probability measure $P$ on $\cx$ we have
\begin{align}\label{eq:sq_bracket}
\cn(\epsilon; \ce_2^-, L^2(P)) \leq \cn\Bigl(
\frac{\epsilon}4; \ce, L^2(P)\Bigr)^2.
\end{align}
\end{lemma}

\begin{proof}
Fix $\epsilon>0$ and a probability distribution $P$ on $\cx$.
Define the collection $\ce^- = \{f-g:f \in \ce, g \in \ce\}$.
Suppose $\{f_1,\ldots,f_N\}$ is a $\epsilon/4$ cover of $\ce$ in the $L^2(P)$ norm.
Then the collection $D = \{d_{ij}=f_i-f_j: 1 \leq i,j \leq N\}$ is a $\epsilon/2$ cover of $\ce^-$ in the $L^2(P)$ norm, since
for any $f-g \in \ce^-$ there exist $i,j$ such that $\|f-f_i\|_{2,P} \vee \|g-f_j\|_{2,P} \leq \epsilon/4$ and so
\[
\|(f-g)-d_{ij}\|_{2,P} \leq \|f-f_i\|_{2,P} + \|g-f_j\|_{2,P} \leq \frac{\epsilon}{2}
\]
showing that $\cn(\epsilon/2,\ce^-,L^2(P)) \leq (\cn(\epsilon/4,\ce,L^2(P)))^2$.
But by applying Lemma~\ref{lemma:covering_bracketing} with $\ce=\ce^-$ and $g(e,w)=e^2$ (hence we can take $C=2$) we have $\cn(\epsilon,\ce_2^-,L^2(P)) \leq \cn(\epsilon/2,\ce^-,L^2(P))$ for all $\epsilon > 0$.
Chaining together the inequalities preceding two sentences establishes~\eqref{eq:sq_bracket}.
\end{proof}
 
\subsection{Miscellaneous}
Here we have some standalone technical lemmas.
Their proofs
do not depend on any of our other results.
\begin{lemma}
\label{lemma:multivariate_cs}
Suppose $X$ and $Y$ are mean zero random vectors in $\real^p$ with finite second moments where $\Cov(Y)$ has full rank.
Then
\[
\Cov(X) \succeq \Cov(X,Y)(\Cov(Y))^{-1}\Cov(X,Y)^{\fromart\top}.
\]
\end{lemma}
\begin{proof}[Proof of Lemma~\ref{lemma:multivariate_cs}]
For any matrix $A \in \real^{p \times p}$ we have $(X+AY)(X+AY)^\top \succeq 0$,
hence
\[
\e[(X+AY)(X+AY)^\top] = \Cov(X) + \Cov(X,Y)A^\top + A\Cov(X,Y)^\top + A\Cov(Y)A^\top \succeq 0
\]
Taking $A=-\Cov(X,Y)(\Cov(Y))^{-1}$ yields the desired result.
\end{proof}
\begin{lemma}
\label{lemma:mean_norm_sq}
If $Y=\sum_{i=1}^n Y_i$ is a random vector where $Y_1,\ldots,Y_n$ are independent with mean 0
and finite second moments,
then $\e[\|Y\|^2] = \sum_{i=1}^n \e[\|Y_i\|^2]$.
\end{lemma}
\begin{proof}[Proof of Lemma~\ref{lemma:mean_norm_sq}]
By assumption we have $\e[Y_i^{\top}Y_j] = \e[Y_i]^{\top}\e[Y_j]=0$ for $i \neq j$, so
\[
\e[\|Y\|^2] = \e\Biggl[\bigg\|\sum_{i=1}^n Y_i\bigg\|^2\Biggr]  = \e\Biggl[\biggl(\sum_{i=1}^n  Y_i\biggr)^{\top}\biggl(\sum_{i=1}^n Y_i\biggr)\Biggr] = \sum_{i=1}^n \e[Y_i^{\top}Y_i] = \sum_{i=1}^n \e[\|Y_i\|^2].
\]
\end{proof}

\section{Proofs}
\label{app:proofs}
Here we collect proofs of the formal results stated in the main text;
some are generalized to allow for some nonstationarities across batches,
as described in Appendix~\ref{app:nonstationary}.
\subsection{Proof of Proposition~\ref{prop:oracle_clt}: CLT for oracle $\hat{\theta}^*$}
\label{proof:prop:oracle_clt}
The CLT of Proposition~\ref{prop:oracle_clt} for our oracle pooled estimator $\hat{\theta}^*$
holds under the numbered generalizations in Appendix~\ref{app:nonstationary} without further restrictions;
here we prove this more general proposition.
It is helpful to begin by noting that $P_{0,e_0} = \sum_{t=1}^t \kappa_t P_{t,e_t}$,
which shows, for instance, that for any $P_0$-integrable function $f$
\[
\e_t[|f(W)|]  = \e_0\left[|f(W)|\frac{\mrd P_{t,e_t}}{\mrd P_{0,e_0}}(W)\right] \leq \kappa_t^{-1} \e_0[|f(W)|]
\]

By score linearity~\eqref{eq:linear_score} we can write
\begin{equation}
\label{eq:theta_star_diff}
\sqrt{N}(\hat{\theta}^*-\theta_0) = -\biggl(\frac{1}{N}\sum_{t=1}^T\sum_{i=1}^{N_t} s_a(W_{ti};\nu_0,e_{0,N})\biggr)^{-1} \biggl(\frac{1}{\sqrt{N}}\sum_{t=1}^T\sum_{i=1}^{N_t} s(W_{ti};\theta_0,\nu_0,e_{0,N})\biggr)
\end{equation}
whenever $\hat{\theta}^*$ exists.
Defining
\[
r_N = \sum_{t=1}^T \frac{N_t}{N} \frac{1}{N_t} \sum_{i=1}^{N_t}(s_a(W_{ti};\nu_0,e_{0,N})-s_a(W_{ti};\nu_0,e_0)).
\]
we have from~\eqref{eq:prop_asymp_limit} and the law of large numbers that
\begin{align}
\frac{1}{N}\sum_{t=1}^T\sum_{i=1}^{N_t} s_a(W_{ti};\nu_0,e_{0,N}) & = \sum_{t=1}^T \frac{N_t}{N} \frac{1}{N_t} \sum_{i=1}^{N_t} s_a(W_{ti};\nu_0,e_{0,N}) \nonumber \\
& = r_N + \sum_{t=1}^T \frac{N_t}{N} (\e_t[s_a(W;\nu_0,e_0)] + o_p(1)) \nonumber \\
& = \sum_{t=1}^T \kappa_t \e_t[s_a(W;\nu_0,e_0)] + o_p(1) \nonumber \\
& = \e_0\left[\sum_{t=1}^T \kappa_t \frac{\mrd P_t}{\mrd P_0}(W)s_a(W;\nu_0,e_0)\right] + o_p(1) \nonumber \\
& = \e_0[s_a(W;\nu_0,e_0)] + o_p(1). \label{eq:oracle_s_a_lln}
\end{align}
Then the third equality follows because
\[
\|r_N\| \leq \sum_{t=1}^T \frac{N_t}{N} \frac{1}{N_t} \sum_{i=1}^{N_t}\|s_a(W_{ti};\nu_0,e_{0,N})-s_a(W_{ti};\nu_0,e_0)\| = o_p(1)
\]
by Lemma~\ref{lemma:conditional_oh_pee} and condition~\ref{cond:score_continuity} of the Proposition,
noting that 
\begin{align*}
\e\left[\frac{1}{N_t} \sum_{i=1}^{N_t}\|s_a(W_{ti};\nu_0,e_{0,N})-s_a(W_{ti};\nu_0,e_0)\|\right] & \leq \kappa_t^{-1} \e_0[\|s_a(W;\nu_0,e_{0,N})-s_a(W;\nu_0,e_0)\|] \\
& \leq \kappa_t^{-1} (\e_0[\|s_a(W;\nu_0,e_{0,N})-s_a(W;\nu_0,e_0)\|^2])^{1/2}  \\
& \leq \kappa_t^{-1} \delta_N
\end{align*}
for $t=1,\ldots,T$.
Invertibility of $\e_0[s_a]$ from condition~\ref{cond:s_a_invertibility} ensures that $\hat{\theta}^*$ is well-defined with probability tending to 1 by~\eqref{eq:oracle_s_a_lln}.

Next, we fix $c \in \real^p$ with $\|c\|=1$ and a batch $t \in \{1,\ldots,T\}$.
Define
\[
U_{N,t,i} = \frac{c^{\top}s(W_{ti};\theta_0,\nu_0,e_{0,N})}{(N_t \cdot c^{\top}V_{t,N}c)^{1/2}}, \quad i=1,\ldots,N_t
\]
where $V_{t,N} = \e_t[s(W;\theta_0,\nu_0,e_{0,N})^{\otimes 2}]$.
Evidently the random variables $U_{N,1},\ldots,U_{N,N_t}$ are independent,
with $\e[U_{N,t,i}]=0$ for all $i$ by~\eqref{eq:gen_score},
and $\sum_{i=1}^{N_t} \e[U_{N_{t,i}}^2] = 1$.
Furthermore we have
\[
\lim_{N_t \rightarrow \infty} \sum_{i=1}^{N_t} \e[|U_{N,t,i}|^q] \leq N_t \cdot \frac{\e_t[\|s(W;\theta_0,\nu_0,e_{0,N})\|^q]}{(N_t \cdot c^{\top}V_{t,N}c)^{q/2}} = O(N_t^{1-q/2}) = o(1)
\]
for all sufficiently large $N$ and some $q>2$ by condition~\ref{cond:moment_boundedness}.
Then by the Lyapunov CLT we have 
\[
\sum_{i=1}^{N_t} U_{N,t,i} = c^{\top}\biggl((c^{\top}V_{t,N}c)^{-1/2}\frac{1}{\sqrt{N_t}} \sum_{i=1}^{N_t} s(W_{ti};\theta_0,\nu_0,e_{0,N}) \biggr) \tod \mathcal{N}(0,1).
\]
\[
V_{t,N} = \e_t[s(W;\theta_0,\nu_0,e_{0,N})^{\otimes 2}] \rightarrow \e_t[s(W;\theta_0,\nu_0,e_0)^{\otimes 2}] \equiv V_t
\]
as $N \rightarrow \infty$.
Therefore
\[
c^{\top}\biggl(\frac{1}{\sqrt{N_t}} \sum_{i=1}^{N_t} s(W_{ti};\theta_0,\nu_0,e_{0,N}) \biggr) \tod \mathcal{N}(0, c^{\top}V_tc)
\]
and since $c$ was arbitrary,
\[
\frac{1}{\sqrt{N_t}} \sum_{i=1}^{N_t} s(W_{ti};\theta_0,\nu_0,e_{0,N}) \tod \mathcal{N}(0,V_t), \quad t=1,\ldots,T.
\]
With the left-hand side of the preceding display independent across batches $t=1,\ldots,T$, we have
\begin{align}
\label{eq:oracle_s_clt}
\frac{1}{\sqrt{N}}\sum_{t=1}^T \sum_{i=1}^{N_t} s(W_{ti};\theta_0,\nu_0,e_{0,N}) 
&= \sum_{t=1}^T \sqrt{\frac{N_t}{N}}\frac{1}{\sqrt{N_t}}\sum_{i=1}^{N_t} s(W_{ti};\theta_0,\nu_0,e_{0,N}) \notag\\
&\tod \mathcal{N}\biggl(0,\sum_{t=1}^T \kappa_t V_t\biggr).
\end{align}
With 
\begin{align*}
\sum_{t=1}^T \kappa_t V_t = \sum_{t=1}^T \kappa_t \e_t[s(W;\theta_0,\nu_0,e_0)^{\otimes 2}] & = \e_0\left[\sum_{t=1}^T \kappa_t \frac{\mrd P_t}{\mrd P_0}(W) s(W;\theta_0,\nu_0,e_0)^{\otimes 2}\right] \\
& = \e_0[s(W;\theta_0,\nu_0,e_0)^{\otimes 2}]
\end{align*}
the result of the Proposition follows by~\eqref{eq:theta_star_diff} and~\eqref{eq:oracle_s_a_lln}. 
\subsection{Proof of Corollary~\ref{cor:ate_oracle_clt}: CLT for $\hat\theta^*_{\aipw}$}
\label{proof:cor:ate_oracle_clt}

Corollary~\ref{cor:ate_oracle_clt},
which applies Proposition~\ref{prop:oracle_clt}
to prove a CLT for the oracle estimator $\hat{\theta}_{\aipw}^*$ of $\theta_{0,\ate}$,
holds under the numbered generalizations of Appendix~\ref{app:nonstationary}
under one additional condition:
that the mean functions are stationary,
meaning $\e_t[Y(z) \mid X=x] = m_0(z,x)$ for all $t=1,\ldots,T$, $z=0,1$, and $x \in \cx$.
This condition is needed to ensure Assumption~\ref{assump:gen_identification} holds,
as discussed in Appendix~\ref{app:nonstationary}.

Our proof proceeds by showing that the conditions of the Corollary
imply the conditions of generalized Proposition~\ref{prop:oracle_clt} proven in the previous section with $\theta_0=\theta_{0,\ate}$ and $s=s_{\aipw}(\cdot)$.
That is, first we show Assumption~\ref{assump:gen_identification} is satisfied with $\theta_0=\theta_{0,\ate}$ and $s=s_{\aipw}(\cdot)$.
Then we show the three numbered conditions in Proposition~\ref{prop:oracle_clt}.

First, for brevity let $\nu_0=\nu_{0,\aipw}(\cdot)=(m_0(0,\cdot),m_0(1,\cdot))$,
and note that for any $e$, $e' \in \cf_{\gamma}$ we have
\begin{align*}
\e_{t,e}[s_{\aipw}(W;\theta_{0,\ate},\nu_0,e')]  & = \e_{t,e}[m_0(1,X)-m_0(0,X)-\theta_{0,\ate}]\\
&\phe + \e_{t,e}\left[\frac{Z(Y(1)-m_0(1,X))}{e'(X)} - \frac{(1-Z)(Y(0)-m_0(0,X))}{1-e'(X)} \right] \\
& = \e_{t,e}\left[\frac{e(X)}{e'(X)}\e_{t,e}[Y(1)-m_0(1,X) \mid X]\right] \\
&\phe - \e_{t,e}\left[\frac{1-e(X)}{1-e'(X)}\e_{t,e}[Y(0)-m_0(0,X) \mid X]\right] \\
& = 0
\end{align*}
using unconfoundedness and stationarity of the mean function.
All the necessary expectations exist by our assumption that $\gamma>0$.
Hence Assumption~\ref{assump:gen_identification} is satisfied.
Next we show the conditions of Proposition~\ref{prop:oracle_clt}:
\begin{enumerate}
\item Trivially we have
\[
\e_0[|s_{\aipw,a}(W;\nu_0,e_{0,N})-s_{\aipw,a}(W;\nu_0,e_0)|^2] = \e_0[|-1-(-1)|^2] = 0.
\]
Next we compute the following for each $(\nu,e) \in \cn \times \cf_{\gamma}$:
\begin{align}
\label{eq:s_aipw_diff} 
&\phe\, s_{\aipw}(W;\theta_0,\nu,e)-s_{\aipw}(W;\theta_0,\nu_0,e_0) \nonumber \\
& = \left(1-\frac{Z}{e_0(X)}\right)(m(1,X)-m_0(1,X)) \nonumber \\
&\phe +\left(1-\frac{1-Z}{1-e_0(X)}\right)(m_0(0,X)-m(0,X)) \nonumber \\
&\phe + Z(Y-m(1,X))(e(X)^{-1}-e_0(X)^{-1}) \nonumber \\
&\phe - (1-Z)(Y-m(0,X))((1-e(X))^{-1}-(1-e_0(X))^{-1}). \nonumber
\end{align}
Plugging in $(\nu,e)=(\nu_0,e_{0,N})$ gives,
by Minkowski's inequality, that
\begin{align*}
(\e_0[|s_{\aipw}(W;\theta_0,\nu_0,e_{0,N})-s_{\aipw}(W;\theta_0,\nu_0,e_0)|^2])^{1/2} \leq A_0+B_0
\end{align*}
where
\begin{align*}
A_0 & = (\e_0[Z^2(Y(1)-m_0(1,X))^2(e_{0,N}(X)^{-1}-e_0(X)^{-1})^2])^{1/2} \\
& \leq \gamma^{-2} (\e_0[(Y(1)-m_0(1,X))^2(e_{0,N}(X)-e_0(X))^2])^{1/2} \\
& = \gamma^{-2}(\e_0[(e_{0,N}(X)-e_0(X))^2v_0(1,X)])^{1/2} \\
& \leq C\gamma^{-2}\|e_{0,N}-e_0\|_{2,P_0^X}.
\end{align*}
By an analogous computation
\begin{align*}
B_0 & = (\e_0[(1-Z)^2(Y(0)-m_0(0,X))^2((1-e_{0,N}(X))^{-1}-(1-e_0(X))^{-1})^2])^{1/2} \\
& \leq C\gamma^{-2}\|e_{0,N}-e_0\|_{2,P_0^X}.
\end{align*}
The result now follows because
\begin{align*}
|e_{0,N}(x)-e_0(x)| &= \Bigg|\sum_{t=1}^T \left(\frac{N_t}{N}-\kappa_t\right)e_t(x)\frac{\mrd P_t^X}{\mrd P_0^X}(x) \Bigg|  \\
& \leq \left(\sup_{1\le t\le T} \kappa_t^{-1} \Big|\frac{N_t}{N}-\kappa_t\Big|\right) \sum_{1\le t\le T} e_t(x) \nonumber \\
& \leq T\left(\sup_{1\le t\le T} \kappa_t^{-1} \Big|\frac{N_t}{N}-\kappa_t\Big|\right) = o(1)
\end{align*}
and hence
\[
\|e_{0,N}-e_0\|_{2,P_0^X}^2 \leq \sup_{x \in \cx} |e_{0,N}(x)-e_0(x)|^2 = o(1).
\]
This bound only uses the fact that propensities are bounded between 0 and 1, so
\begin{equation}
\label{eq:e_0_N_vs_e_0}
\sup_{x \in \cx, e_1(\cdot),\ldots,e_T(\cdot) \in \cf_0} |e_{0,N}(x)-e_0(x)| = o(1).
\end{equation}

\item Evidently $\e_0[s_{\aipw,a}(W;\nu_0,e_0)] = -1$ is invertible.
Additionally for $z=0,1$ we have $\e[Y(z)^2] \le C$ by the moment conditions in Assumption~\ref{assump:ate_regularity} so
\[
\e_0[m_0(z,X)^2] = \e_0[(\e_0(Y(z) \mid X))^2 ]\leq \e_0[Y(z)^2] < \infty.
\]
Now $s_{\aipw}(W;\theta_{0,\ate},\nu_{0,\aipw},e_0)$ is the sum of the following terms:
\begin{align*}
&m_0(1,X)\left(1-\frac{Z}{e_0(X)}\right),\\
&m_0(0,X)\left(\frac{1-Z}{1-e_0(X)}-1\right)-\theta_{0,\ate},\quad\text{and}\\
&Y\left(\frac{Z}{e_0(X)}-\frac{1-Z}{1-e(X)}\right).
\end{align*}
These are all square integrable, because $Z$, $(e_0(X))^{-1}$, and $(1-e_0(X))^{-1}$ are all uniformly bounded.
\item 
From $(\e_0[|Y(z)|^q])^{1/q} \leq C$ for $z=0,1$ by Assumption~\ref{assump:ate_regularity},
we have
\[
\e_0[|m_0(z,X)|^q] = \e[|\e_{0}[Y(z) \mid X]|^q] 
\leq \e_0[|Y(z)|^q] \leq C^q.
\]
Then each of
\begin{align*}
&\e_0\left[|m_0(1,X)|^q\Big|1-\frac{Z}{e_0(X)}\Big|^q\right], \\
&\e_0\left[|m_0(0,X)|^q\Big|\frac{1-Z}{1-e_0(X)}-1\Big|^q\right],\quad\text{and}\\
&\e_0\left[|Y|^q\Big|\frac{Z}{e_0(X)}-\frac{1-Z}{1-e(X)}\Big|^q\right]
\end{align*}
is at most $[C(1+\gamma^{-1})]^q$
and the desired condition holds by Minkowski's inequality.
\end{enumerate}
Now we can apply Proposition~\ref{prop:oracle_clt},
to conclude $\sqrt{N}(\hat{\theta}^*-\theta_{0,\ate}) \tod \mathcal{N}(0,V_0)$ where
\begin{align*}
V_0 & = \e_0[s_{\aipw}(W;\theta_{0,\ate},\nu_{0,\aipw},e_0)^2] \\
& = \e_0\left[\left((\tau_0(X)-\theta_{0,\ate})+\frac{Z(Y(1)-m_0(1,X))}{e_0(X)} - \frac{(1-Z)(Y(0)-m_0(0,X))}{1-e_0(X)}\right)^2\right] \\
& = \e_0[(\tau_0(X)-\theta_{0,\ate})^2] + \e_0\left[\frac{Z(Y(1)-m_0(1,X))^2}{(e_0(X))^2}\right] + \e_0\left[\frac{(1-Z)(Y(0)-m_0(0,X))^2}{(1-e_0(X))^2}\right] \\
& = \e_0[(\tau_0(X)-\theta_{0,\ate})^2] + \e_0\left[\frac{1}{e_0(X)} \cdot \e[(Y(1)-m_0(1,X))^2 \mid X]\right]  \\
& + \e_0\left[\frac{1}{1-e_0(X)} \cdot \e[(Y(0)-m_0(0,X))^2 \mid X]\right] \\
& = V_{0,\aipw}
\end{align*}
The third equality in the preceding display follows by noting that the three cross terms in the expansion of the square have mean zero. The first two vanish by conditioning on $X$ and the third because $Z(1-Z)=0$.

\subsection{Proof of Corollary~\ref{cor:pl_oracle_clt}:
CLT for $\hat\theta^*_{\epl}$}
\label{proof:cor:pl_oracle_clt}

Here we prove Corollary~\ref{cor:pl_oracle_clt},
the CLT for the partial linear estimator
$\hat\theta^*_{\epl}$ of the regression parameter $\theta_{0,\pl}$ under the linear treatment effect assumption~\eqref{eq:pl_assumption}.
This Corollary holds under the numbered generalizations of Appendix~\ref{app:nonstationary}
with the additional condition
that the mean and variance functions are stationary.
That is,
we have $\e_t[Y(z) \mid X=x] = m_0(z,x)$ and $\Var_t(Y(z) \mid X=x) = v_0(z,x)$ for all $t=1,\ldots,T$, $z=0,1$, and $x \in \cx$.
This condition is needed to ensure Assumption~\ref{assump:gen_identification} holds.

As in the proof of Corollary~\ref{cor:ate_oracle_clt},
we first show that Assumption~\ref{assump:gen_identification} holds with $\gamma=0$,
estimand $\theta_0=\theta_{0,\pl}$,
score $s(\cdot)=s_{\epl}(\cdot)$,
and nuisance functions $\nu_0=\nu_{0,\epl}(\cdot)=(m_0(0,\cdot),v_0(0,\cdot),v_0(1,\cdot))$
lying in the nuisance set $\cn=\cn_{\epl}$.
Then we show that the three numbered conditions in Proposition~\ref{prop:oracle_clt} hold.

Fix $e(\cdot),e'(\cdot) \in \cf_0$.
For each $t=0,1,\ldots,T$, 
\begin{align*}
\e_{t,e}[Y \mid X,Z=0] & = \e_{t,e}[Y(0) \mid X] = m_0(0,X),\quad\text{and} \\
\e_{t,e}[Y \mid X,Z=1] & = \e_{t,e}[Y(1) \mid X] = m_0(1,X)
\end{align*}
hold by the unconfoundedness and SUTVA assumptions.
Hence by~\eqref{eq:pl_assumption}, 
\begin{equation}
\label{eq:pl_cond_means}
\e_{t,e}[Y \mid X,Z] = m_0(0,X) + Z\psi(X)^{\top}\theta_0.
\end{equation}
Thus for any $e'(\cdot) \in \cf_{\gamma}$,
\begin{align*}
\e_{t,e}[s(W;\theta_{0,\pl},\nu_0,e')] & = \e_{t,e}[w(X;\nu_0,e')(Z-e'(X))(Y-m_0(0,X)-Z\psi(X)^{\top}\theta_{0,\pl})\psi(X)] \\
& = 0
\end{align*} 
after conditioning on $(X,Z)$ and applying~\eqref{eq:pl_cond_means}.
Integrability is not a concern because $w(X;\nu_0,e') \leq c^{-1}$ for $\nu_0 \in \cn_{\epl}$.
Thus Assumption~\ref{assump:gen_identification} is satisfied.

Now we consider the numbered conditions of Proposition~\ref{prop:oracle_clt} in turn.
\begin{enumerate}
\item 
Because the predictor variables $\psi(X)$ satisfy $\|\psi(X)\| \leq C$
\begin{align*}
(\e_0[\|s_a(W;\nu_0,e_{0,N})-s_a(W;\nu_0,e_0)\|^2])^{1/2} & = (\e_0[\|Z\Delta(X,Z)\psi(X)\psi(X)^{\top}\|^2])^{1/2} \\
& \leq C^2(\e_0[\Delta(X,Z)^2])^{1/2}
\end{align*}
where
\begin{align*}
\Delta(X,Z) & = w(X;\nu_0,e_0)(Z-e_0(X)) -w(X;\nu_0,e_{0,N})(Z-e_{0,N}(X)) \\
& = (w(X;\nu_0,e_0)-w(X;\nu_0,e_{0,N}))(Z-e_0(X)) + w(X;\nu_0,e_{0,N})(e_{0,N}(X)-e_0(X)).
\end{align*}
This $\Delta$ satisfies
\begin{align}\label{eq:whatdeltasatisfies}
\begin{split}
(\e_0[\Delta(X,Z)^2])^{1/2} & \leq (\e_0[(w(X;\nu_0,e_0)-w(X;\nu_0,e_{0,N}))^2(Z-e_0(X))^2])^{1/2} \\
&\phe + (\e_0[w^2(X;\nu_0,e_{0,N})(e_{0,N}(X)-e_0(X))^2])^{1/2}.
\end{split}
\end{align}
Now $c \leq v_0(z,x) \leq \e[Y(z)^2 \mid X=x] \leq C$ for all $z=0,1$ and $x \in \cx$,
so for any propensity $e(\cdot)$ we have
\[
C^{-1} \leq w(X;\nu_0,e) = (v_0(0,X)e(X)+v_0(1,X)(1-e(X)))^{-1} \leq c^{-1}
\]
and
\begin{align*}
\sup_{e \in [0,1]} \Big|\frac{\partial}{\partial e} \frac{1}{v_0(0,X)e+v_0(1,X)(1-e)}\Big| &= 
\sup_{e \in [0,1]} \Big|\frac{v_0(0,X)-v_0(1,X)}{(v_0(0,X)e+v_0(1,X)(1-e))^2} \Big| \\
&\leq \frac{2C}{c^2}.
\end{align*}
Applying these latter two facts to equation~\eqref{eq:whatdeltasatisfies} yields
\begin{align*}
(\e_0[\Delta(X,Z)^2])^{1/2}& \leq \left(\frac{2C}{c^2}+c^{-1}\right)\|e_{0,N}-e_0\|_{2,P_0^X} 
 = o(1) 
\end{align*}
by~\eqref{eq:e_0_N_vs_e_0}.
Similarly,
\begin{align*}
&\phe\,\, (\e_0 [\|s(W;\theta_{0,\pl},\nu_0,e_{0,N})-s(W;\theta_{0,\pl},\nu_0,e_0)\|^2])^{1/2} \\
& = (\e_0[\|\Delta(X,Z)(Y-m_0(0,X)-Z\psi(X)^{\top}\theta_0)\psi(X)\|^2])^{1/2} \\
& \leq C(\e_0[\Delta(X,Z)^2v_0(Z,X)])^{1/2} \\
& \leq C^{3/2}(\e_0[\Delta(X,Z)^2])^{1/2}\\ 
&= o(1).
\end{align*}
\item For invertibility of $\e_0[s_a]$, we compute
\begin{align*}
\e_0[s_a(W;\nu_0,e_0)] & = \e_0[-w(X;\nu_0,e_0)(Z-e_0(X))Z\psi(X)\psi(X)^{\top}] \\
& = \e_0[-w(X;\nu_0,e_0)e_0(X)(1-e_0(X))\psi(X)\psi(X)^{\top}] \\
& \preceq -C^{-1}\e_0[e_0^2(X)(1-e_0^2(X))\psi(X)\psi(X)^{\top}].
\end{align*}
With $\e_0[e_0(X)(1-e_0(X))\psi(X)\psi(X)^{\top}]$ positive definite by assumption,
$\e_0[s_a(W;\nu_0,e_0)]$ is strictly negative definite and hence invertible.

For boundedness of $\e_0[\Vert s\Vert^2]$ with $e=e_0$, 
we compute
\begin{align*}
\e_0[\|s(W;\theta_0,\nu_0,e_0)\|^2] & = \e_0[w(X;\nu_0,e_0)^2(Z-e_0(X))^2(Y-m_0(0,X)-Z\psi(X)^{\top}\theta_0)^2\|\psi(X)\|^2] \\
& \leq \frac{C^2}{c^2}\e_0[(Y-m_0(0,X)-Z\psi(X)^{\top}\theta_0)^2] \\
& = \frac{C^2}{c^2}\e_0[v_0(Z,X)] \leq \frac{C^3}{c^2}
\end{align*}
\item For boundedness of $\e_0[\Vert s\Vert^q]$ with $e=e_{0,N}$,
\begin{align*}
(\e_0&[\|s(W;\theta_0,\nu_0,e_{0,N})\|^q])^{1/q} \\
& = \Big(\e_0[|w(X;\nu_0,e_{0,N})|^q|Z-e_{0,N}(X)|^q|Y-m_0(0,X)-Z\psi(X)^{\top}\theta_0|^q\|\psi(X)\|^q]\Big)^{1/q} \\
& \leq \frac{C}{c}(\e_0[|Y-m_0(0,X)-Z\psi(X)^{\top}\theta_0|^q])^{1/q} \\
& \leq \frac{4C^2}{c}
\end{align*}
where the final inequality follows from Minkowski's inequality and Jensen's inequality as below:
\begin{align*}
(\e_0[|Y-\e_0(Y \mid X,Z)|^q])^{1/q} & \leq (\e_0[|Y|^q])^{1/q} + (\e_0[|\e_0[Y \mid X,Z]|^q])^{1/q}  \\
& \leq (\e_0[|Y|^q])^{1/q} + (\e_0[\e_0[|Y|^q \mid X,Z]])^{1/q} \\
& = 2(\e_0[|Y|^q])^{1/q} \\
& \leq 2[(\e_0[|Y(0)|^q])^{1/q} + (\e_0[|Y(1)|^q])^{1/q}] \\
& \leq 4C \text{ by Assumption~\ref{assump:pl_regularity}}
\end{align*}
\end{enumerate}
We now compute
\begin{align*}
\e_0[s_a(W;\nu_0,e_0)] & = \e_0[-w(X;\nu_0,e_0)Z(Z-e_0(X))\psi(X)\psi(X)^{\top}] \\
& = \e_0[-w(X;\nu_0,e_0)(1-e_0(X))e_0(X)\psi(X)\psi(X)^{\top}] \\
& = -\e_0\left[\frac{e_0(X)(1-e_0(X))}{v_0(0,X)e_0(X)+v_0(1,X)(1-e_0(X))}\psi(X)\psi(X)^{\top}\right]
\end{align*}
and
\begin{align*}
\e_0[s(W;\theta_0,\nu_0,e_0)^{\otimes 2}] & = \e_0[w^2(X;\nu_0,e_0)(Z-e_0(X))^2(Y-m_0(0,X)-Z\psi(X)^{\top}\theta_0)^2\psi(X)\psi(X)^{\top}] \\
& = \e_0[w^2(X;\nu_0,e_0)(Z-e_0(X))^2v_0(Z,X)\psi(X)\psi(X)^{\top}] \\
& = \e_0[w^2(X;\nu_0,e_0)e_0^2(X)v_0(0,X)(1-e_0(X))\psi(X)\psi(X)^{\top}] \\
&\phe\ + \e_0[w^2(X;\nu_0,e_0)(1-e_0(X))^2v_0(1,X)e_0(X)\psi(X)\psi(X)^{\top}] \\
& = \e_0[e_0(X)(1-e_0(X))w^2(X;\nu_0,e_0)(v_0(0,X)e_0(X)+v_0(1,X)(1-e_0(X)))\psi(X)\psi(X)^{\top}] \\
& = -\e_0[s_a(W;\nu_0,e_0)].
\end{align*}
Finally, we apply Proposition~\ref{prop:oracle_clt} to conclude $\sqrt{N}(\hat{\theta}^*-\theta_{0,\pl}) \rightarrow \mathcal{N}(0,V_0)$ where
\begin{align*}
V_0 & = (\e_0[s_a(W;\nu_0,e_0)])^{-1}(\e_0[s(W;\theta_0,\nu_0,e_0)^{\otimes 2}])(\e_0[s_a(W;\nu_0,e_0)])^{-1}  \\
& = -(\e_0[s_a(W;\nu_0,e_0)])^{-1} \\
& = V_{0,\epl}.
\end{align*}

\subsection{Proof of Theorem~\ref{thm:pooled_covariance}:
Pooling dominates linear aggregation}
\label{proof:thm:pooled_covariance}

Here we prove that the oracle pooled estimators $\hat{\theta}_{\aipw}^*$ and $\hat{\theta}_{\epl}^*$ dominate the best
linearly aggregated estimators $\hat{\theta}_{\aipw}^{*,(\la)}$ and $\hat{\theta}_{\epl}^{*,(\la)}$ for estimating $\theta_{0,\ate}$ and $\theta_{0,\pl}$, respectively.
For $\theta_{0.\pl}$, Theorem~\ref{thm:pooled_covariance} generalizes to nonstationary batches as described in Appendix~\ref{app:nonstationary}. 
But for $\theta_{0,\ate}$, the Theorem does not necessarily hold under the nonstationarities of Appendix~\ref{app:nonstationary}.
For example,
suppose that $T=2$, $P_1^X$ is uniform on $(0,1)$, and $P_2^X$ is the distribution with density $2x$ on $(0,1)$ with respect to Lebesgue measure.
Let $v_0(z,x)=x$
and suppose $e_1(x)=e_2(x)=\kappa_1=1/2$ for each $x \in \cx$.
Then $V_{0,\aipw}=7/3 > V^{(\la)}_{\aipw} = 16/7$.
We can however generalize the ATE result to the case of a multivariate outcome variable,
i.e., $Y(1),Y(0) \in \real^q$,
where the conditional mean and variance functions $m_0(\cdot)$ and $v_0(\cdot)$ take values in $\real^q$ and $\ess_+^q$, respectively.

We begin by showing the result for ATE estimation.
We have $V^{(\la)}_{\aipw} = \left(\sum_{t=1}^T \kappa_t V_{t,\aipw}^{-1}\right)^{-1}$
where we compute
\[
V_{t,\aipw} = \e\left[\frac{v_0(1,X)}{e_t(X)} + \frac{v_0(0,X)}{1-e_t(X)} + (\tau_0(X)-\theta_0)^{\otimes 2}\right], \quad t=1,\ldots,T
\]
by applying Corollary~\ref{cor:ate_oracle_clt} to the observations in each single batch $t \in \{1,\ldots,T\}$.
Letting
\begin{align*}
h(e) & = \left(\e\left[\frac{v_0(1,X)}{e(X)} + \frac{v_0(0,X)}{1-e(X)} + (\tau_0(X)-\theta_0)^{\otimes 2}\right]\right)^{-1} \\
& = \left(\e\left[\frac{v_0(1,X)(1-e(X)) + v_0(0,X)e(X)}{e(X)(1-e(X))} + (\tau_0(X)-\theta_0)^{\otimes 2}\right]\right)^{-1}
\end{align*}
for each $e \in \cf_{\gamma}$,
we have
\begin{align*}
\Bigl(V^{(\la)}_{\aipw}\Bigr)^{-1} & = \sum_{t=1}^T \kappa_t \left(\e\left[\frac{v_0(1,X)}{e_t(X)} + \frac{v_0(0,X)}{1-e_t(X)} + (\tau_0(X)-\theta_0)^{\otimes 2}\right]\right)^{-1} \\
&= \sum_{t=1}^T \kappa_t h(e_t),\quad\text{versus}\\
V_{0,\aipw}^{-1} & = \left(\e\left[\frac{v_0(1,X)}{e_0(X)} + \frac{v_0(0,X)}{1-e_0(X)} + (\tau_0(X)-\theta_0)^{\otimes 2}\right]\right)^{-1}\\ 
&= h(e_0).
\end{align*}
Thus, to prove the desired result,
it suffices to show $h(\cdot)$ is a concave (matrix-valued) function on $\cf_{\gamma}$.
To that end, fix $e_1,e_2 \in \cf_{\gamma}$ and $\lambda \in (0,1)$.
It suffices to show that $g(\lambda) = h(e_1+\lambda(e_2-e_1))$ is concave on $[0,1]$,
i.e., that $g(\lambda)\succeq \lambda g(e_2)+(1-\lambda)g(e_1)$ for each $\lambda \in [0,1]$.
Because
$g(\cdot)$ is 
continuous on $[0,1]$,
we need only show that $g''(\lambda) \preceq 0$ for each $\lambda \in (0,1)$.
Letting $e_{\lambda}=e_1+\lambda(e_2-e_1)$ and $V_{\lambda}=h(e_{\lambda})^{-1}$, we obtain 
\begin{align}\label{eq:gpp}
g''(\lambda) = 2h(e_{\lambda})[B_{\lambda}h(e_{\lambda})B_{\lambda}-C_{\lambda}]h(e_{\lambda})
\end{align}
for
\begin{align*}
B_{\lambda} & = \e\left[(e_2(X)-e_1(X))\left(\frac{v_0(0,X)}{(1-e_{\lambda}(X))^2} - \frac{v_0(1,X)}{(e_{\lambda}(X))^2}\right)\right],\quad\text{and} \\
C_{\lambda} & = \e\left[(e_2(X)-e_1(X))^2\left(\frac{v_0(1,X)}{(e_{\lambda}(X))^3} + \frac{v_0(0,X)}{(1-e_{\lambda}(X))^3}\right)\right].
\end{align*}
To get this, we reversed the order of differentiation and expectation, 
which is justified by the regularity conditions in Assumption~\ref{assump:ate_regularity}.

Since $h(e_{\lambda})$ is symmetric and appears on the
left and right in~\eqref{eq:gpp},
it suffices to show that $B_{\lambda}h(e_{\lambda})B_{\lambda} \preceq C_{\lambda}$.
For this we note that for $\tilde{Z}$ conditionally independent of $(Y(0),Y(1))$ given $X$
with $\tilde{Z} \mid X \sim \Bern(e_{\lambda}(X))$,
we have $h(e_{\lambda}) = (\Cov(b_{\lambda}))^{-1}$,
$B_{\lambda}=\Cov(c_{\lambda},b_{\lambda})=B_{\lambda}'$, and $C_{\lambda}=\Cov(c_{\lambda})$ where
\begin{align*}
b_{\lambda} & = \frac{(1-\tilde{Z})(Y(0)-m_0(0,X))}{1-e_{\lambda}(X)} - \frac{\tilde{Z}(Y(1)-m_0(1,X))}{e_{\lambda}(X)} - (\tau_0(X)-\theta_0) \\
c_{\lambda} & = (e_2(X)-e_1(X))\biggl(\frac{\tilde{Z}(Y(1)-m_0(1,X))}{(e_{\lambda}(X))^2} + \frac{(1-\tilde{Z})(Y(0)-m_0(0,X))}{(1-e_{\lambda}(X))^2}\biggr)
\end{align*}
are mean zero random vectors in $\real^d$.
Applying Lemma~\ref{lemma:multivariate_cs} 
completes the proof for ATE.

Now we consider the partially linear model,
allowing for the generalizations in Appendix~\ref{app:nonstationary},
as discussed above.
We have $V_{\epl}^{(\la)} = \bigl(\sum_{t=1}^T \kappa_t V_{t,\epl}^{-1}\bigr)^{-1}$ where
\begin{align*}
V_{t,\epl} = \biggl(\e_t\left[\frac{e_t(X)(1-e_t(X))}{v_0(0,X)e_t(X)+v_0(1,X)(1-e_t(X))}\psi(X)\psi(X)^{\top}\right]\biggr)^{-1}
\end{align*}
by applying Corollary~\ref{cor:pl_oracle_clt} to the observations in each single batch $t \in \{1,\ldots,T\}$.
Then
\begin{align*}
\bigl(V^{(\la)}_{\epl} \bigr)^{-1} & = \sum_{t=1}^T \kappa_t V_{t,\epl}^{-1} \\
&= \sum_{t=1}^T \e_t\left[\kappa_t\frac{e_t(X)(1-e_t(X))}{v_0(0,X)e_t(X)+v_0(1,X)(1-e_t(X))}\psi(X)\psi(X)^{\top}\right] \\
& = \e_0\left[\sum_{t=1}^T \kappa_t \frac{\mrd P_t^X}{\mrd P_0^X}(X) \frac{e_t(X)(1-e_t(X))}{v_0(0,X)e_t(X)+v_0(1,X)(1-e_t(X))}\psi(X)\psi(X)^{\top}\right] \\
& \preceq \e_0\left[\frac{e_0(X)(1-e_0(X))}{v_0(0,X)e_0(X)+v_0(1,X)(1-e_0(X))}\psi(X)\psi(X)^{\top}\right]\\
&= (V_{0,\epl})^{-1}
\end{align*}
where the inequality
follows from the fact that for each $x \in \cx$,
we have $\sum_{t=1}^T \kappa_t {\mrd P_t^X}/{\mrd P_0^X}(x) =1$ and
the map 
$$e \mapsto \frac{e(1-e)}{v(0,x)e+v(1,x)(1-e)}\psi(x)\psi(x)^{\top}$$
is concave on $e \in [0,1]$.

\subsection{Proof of Theorem~\ref{thm:batch_clt}: CLT for feasible $\hat{\theta}$ in a CSBAE}\label{proof:thm:batch_clt}
Here we prove the CLT for the feasible estimator $\hat{\theta}$ in a CSBAE.
It holds under the nonstationarity conditions
described in Appendix~\ref{app:nonstationary}.

We begin by constructing the non-adaptive batch experiment in the statement of the theorem.
Define the counterfactual treatment indicators via
\[
\tilde{Z}_{ti} = \indic(U_{ti} \leq e_t(X_{ti})), \quad t=1,\ldots,T,\quad i=1,\ldots,N_t
\]
for $U_{ti}\simiid \dunif(0,1)$.
Then the observations in the counterfactual non-adaptive batch experiment are the vectors $\tilde{W}_{ti}=(\tilde{Y}_{ti},X_{ti},\tilde{Z}_{ti})$,
where
\[
\tilde{Y}_{ti} = Y_{ti}(\tilde{Z}_{it})=\tilde{Z}_{it}Y_{ti}(1) + (1-\tilde{Z}_{it})Y_{ti}(0).
\]
The corresponding oracle estimator $\hat{\theta}^*$ from~\eqref{eq:theta_star} is then
\begin{align*}
\hat{\theta}^*
& = -\biggl(\frac{1}{N} \sum_{k=1}^K \sum_{(t,i) \in \ci_k} s_a(\tilde{W}_{ti};\nu_0,e_{0,N}) \biggr)^{-1} \biggl(\frac{1}{N}\sum_{k=1}^K \sum_{(t,i) \in \ci_k}s_b(\tilde{W}_{ti};\nu_0,e_{0,N})\biggr).
\end{align*}
Because $(\nu_0,e_{0,N}) \in \ct_N$,
conditions~\ref{cond:score_continuity},~\ref{cond:s_a_invertibility}, and~\ref{cond:moment_boundedness} of Proposition~\ref{prop:oracle_clt}
are satisfied for this counterfactual non-adaptive batch experiment
by equations~\eqref{eq:s_a_consistency},~\eqref{eq:s_consistency}, and~\eqref{eq:s_moment_condition}
along with condition~\ref{cond:regularity} of Assumption~\ref{assump:dml}.
Hence
the oracle CLT $\sqrt{n}(\hat{\theta}^*-\theta_0) \tod \mathcal{N}(0,V_0)$ holds
with $V_0$ as in the conclusion of Proposition~\ref{prop:oracle_clt}.

It remains to show that 
\[
\hat{\theta} = \hat{\theta}^* + o_p(N^{-1/2})
\]
where
\begin{align*}
\hat{\theta} = -\biggl(\frac{1}{N} \sum_{k=1}^K \sum_{(t,i) \in \ci_k} s_a(W_{ti};\hat{\nu}^{(-k)},\hat{e}^{(-k)})\biggr)^{-1}\biggl(\frac{1}{N} \sum_{k=1}^K \sum_{(t,i) \in \ci_k} s_b(W_{ti};\hat{\nu}^{(-k)},\hat{e}^{(-k)})\biggr)
\end{align*}
as in~\eqref{eq:theta_hat}.
To do so, 
we show that the following intermediate quantity
\[
\tilde{\theta} = -\biggl(\frac{1}{N} \sum_{k=1}^K \sum_{(t,i) \in \ci_k} s_a(\tilde{W}_{ti};\hat{\nu}^{(-k)},\hat{e}^{(-k)})\biggr)^{-1}\biggl(\frac{1}{N} \sum_{k=1}^K \sum_{(t,i) \in \ci_k} s_b(\tilde{W}_{ti};\hat{\nu}^{(-k)},\hat{e}^{(-k)})\biggr)
\]
satisfies both $\tilde{\theta}=\hat{\theta}^*+o_p(N^{-1/2})$ and $\hat{\theta}=\tilde{\theta}+o_p(N^{-1/2})$.

\subsubsection{Showing that $\tilde\theta =\hat\theta^*+o_p(N^{-1/2})$}

We first show $\tilde{\theta}=\hat{\theta}^*+o_p(N^{-1/2})$.
By score linearity~\eqref{eq:linear_score}, we can write
\[
N^{1/2}(\tilde{\theta}-\hat{\theta}^*) = N^{1/2}(\tilde{\theta}-\theta_0) - N^{1/2}(\hat{\theta}^*-\theta_0) = A_1B_1 + A_2B_2
\]
where
\begin{align*}
A_1 & = \biggl(\frac{1}{N} \sum_{k=1}^K \sum_{(t,i) \in \ci_k} s_a(\tilde{W}_{ti};\nu_0,e_{0,N}) \biggr)^{-1}
-\biggl(\frac{1}{N} \sum_{k=1}^K \sum_{(t,i) \in \ci_k} s_a(\tilde{W}_{ti};\hat{\nu}^{(-k)},\hat{e}^{(-k)})\biggr)^{-1}, \\
B_1 & = \frac{1}{\sqrt{N}}\sum_{k=1}^K \sum_{(t,i) \in \ci_k}s(\tilde{W}_{ti};\theta_0,\nu_0,e_{0,N}), \\
A_2 & = \biggl(\frac{1}{N} \sum_{k=1}^K \sum_{(t,i) \in \ci_k} s_a(\tilde{W}_{ti};\hat{\nu}^{(-k)},\hat{e}^{(-k)})\biggr)^{-1}\quad\text{and} \\
B_2 & = \frac{1}{\sqrt{N}} \sum_{k=1}^K \sum_{(t,i) \in \ci_k} [s(\tilde{W}_{ti};\theta_0,\nu_0,e_{0,N})-s(\tilde{W}_{ti};\theta_0,\hat{\nu}^{(-k)},\hat{e}^{(-k)})].
\end{align*}
We will prove that $A_1=o_p(1)$, $B_1=O_p(1)$, $A_2=O_p(1)$, and $B_2=o_p(1)$.

To show $A_1=o_p(1)$,
for each fold $k \in \{1,\ldots,K\}$
let $N_k = \sum_{t=1}^T n_{t,k} = N/K + O(1)$ be the total number of observations in fold $k$ across all batches $t=1,\ldots,T$.
Also define the quantity
\[
\tilde{A}_1^{(k)} = \frac{1}{N_k} \sum_{(t,i) \in \ci_k} [s_a(\tilde{W}_{ti};\nu_0,e_{0,N})-s_a(\tilde{W}_{ti};\hat{\nu}^{(-k)},\hat{e}^{(-k)})] = \sum_{t=1}^T \frac{n_{t,k}}{N_k} \tilde{A}_{1,t}^{(k)}
\]
where
\[
\tilde{A}_{1,t}^{(k)} = \frac{1}{n_{t,k}} \sum_{i:(t,i) \in \ci_k} s_a(\tilde{W}_{ti};\nu_0,e_{0,N})-s_a(\tilde{W}_{ti};\hat{\nu}^{(-k)},\hat{e}^{(-k)}).
\]
For each $k=1,\ldots,K$,
let $\ce_{N,k}$ be the event that $(\hat{\nu}^{(-k)},\hat{e}^{(-k)}) \in \ct_N$.
This $\ce_{N,k}$ is $\cs^{(-k)}$ measurable and $\Pr(\ce_{N,k}) \rightarrow 1$ as $N \rightarrow \infty$ by assumption.
Then for all sufficiently large $N$, 
\begin{align*}
&\phe\,\e\bigl[\|\tilde{A}_{1,t}^{(k)}\indic(\ce_{N,k})\| \mid \cs^{(-k)}\bigr] \\
& \leq \indic(\ce_{N,k}) \biggl(\frac{1}{n_{t,k}} \sum_{i:(t,i) \in I_j} \e\bigl[\|s_a(\tilde{W}_{ti};\hat{\nu}^{(-k)},\hat{e}^{(-k)}) - s_a(\tilde{W}_{ti};\nu_0,e_{0,N})\| \mid \cs^{(-k)}\bigr] \biggr) \\
& \leq \sup_{(\nu,e) \in \ct_N} \e_t\bigl[\|s_a(W;\nu,e)-s_a(W;\nu_0,e_{0,N})\|\bigr] \\
& \leq \kappa_t^{-1} \sup_{(\nu,e) \in \ct_N} \e_0\bigl[\|s_a(W;\nu,e)-s_a(W;\nu_0,e_0)\|\bigr] 
 + \kappa_t^{-1} \e_0[\|s_a(W;\nu_0,e_{0,N})-s_a(W;\nu_0,e_0)\|]  \\
& \leq 2\kappa_t^{-1} \delta_N 
\end{align*}
by equation~\eqref{eq:s_a_consistency}.
The second inequality above uses the fact that $\hat{\nu}^{(-k)}$, $\hat{e}^{(-k)}$ are nonrandom given $\cs^{(-k)}$,
but the vectors $\tilde{W}_{ti}$ in fold $k$ are independent of $\cs^{(-k)}$ and i.i.d.\ from $P_t$;
the third inequality follows because
\begin{equation}
\label{eq:dP_t_dP_0}
\frac{\mrd P_t}{\mrd P_0}(w) \leq \kappa_t^{-1}
\end{equation}
for all $w$ by the definitions of $P_t$ and $P_0$.
Then $\tilde{A}_{1,t}^{(k)} = o_p(1)$ by Lemma~\ref{lemma:conditional_oh_pee}.
This immediately shows $\tilde{A}_1^{(k)}=o_p(1)$ for all folds $k=1,\ldots,K$ and so
\begin{equation}
\label{eq:A_1_tilde}
\tilde{A}_1 = \sum_{k=1}^K \frac{N_k}{N} \tilde{A}_1^{(k)} 
= o_p(1),
\end{equation}
too. Now recall the identity
\begin{equation}
\label{eq:matrix_inverse_submultiplicative}
\|A^{-1}-B^{-1}\| = \|A^{-1}(A-B)B^{-1}\| \leq \|A^{-1}\|\|A-B\|\|B^{-1}\|
\end{equation}
for any invertible square matrices $A,B$ of the same size.
By~\eqref{eq:oracle_s_a_lln} we know that 
\begin{equation}
\label{eq:oracle_s_a_lln_folds}
\frac{1}{N} \sum_{k=1}^K \sum_{(t,i) \in \ci_k} s_a(\tilde{W}_{ti};\nu_0,e_{0,N}) = \e_0[s_a(W;\nu_0,e_0)] + o_p(1)
\end{equation}
and hence
\begin{align*}
\biggl(\frac{1}{N} \sum_{k=1}^K \sum_{(t,i) \in \ci_k} s_a(\tilde{W}_{ti};\nu_0,e_{0,N})\biggr)^{-1}\quad\text{and}\quad \biggl(\frac{1}{N} \sum_{k=1}^K \sum_{(t,i) \in \ci_k} s_a(\tilde{W}_{ti};\hat{\nu}^{(-k)},\hat{e}^{(-k)})\biggr)^{-1}
\end{align*}
are both $O_p(1)$.
The preceding display along with~\eqref{eq:A_1_tilde} and~\eqref{eq:matrix_inverse_submultiplicative} show $A_1 = o_p(1)$.

Equations~\eqref{eq:A_1_tilde} and~\eqref{eq:oracle_s_a_lln_folds} immediately imply that $A_2=O_p(1)$,
while~\eqref{eq:oracle_s_clt} implies that $B_1=O_p(1)$.
It remains to show $B_2=o_p(1)$.
To that end we consider the quantity
\begin{align*}
B_2^{(k)} &= \frac{1}{\sqrt{N_k}} \sum_{(t,i) \in \ci_k} s(\tilde{W}_{ti};\theta_0,\hat{\nu}^{(-k)},\hat{e}^{(-k)}) - s(\tilde{W}_{ti};\theta_0,\nu_0,e_{0,N}) \\
&= \bar{B}_2^{(k)} + \sum_{t=1}^T \sqrt{\frac{n_{t,k}}{N_k}}\tilde{B}_{2,t}^{(k)}
\end{align*}
where for each $t=1,\ldots,T$
\begin{align*}
\bar{B}_2^{(k)} & = \sum_{t=1}^T \sqrt{\frac{n_{t,k}}{N_k}} \frac{1}{\sqrt{n_{t,k}}} \sum_{i:(t,i) \in \ci_k} \e[s(\tilde{W}_{ti};\theta_0,\hat{\nu}^{(-k)},\hat{e}^{(-k)})-s(\tilde{W}_{ti};\theta_0,\nu_0,e_{0,N}) \mid \cs^{(-k)}],\quad\text{and} \\
\tilde{B}_{2,t}^{(k)} & = \frac{1}{\sqrt{n_{t,k}}} \sum_{i:(t,i) \in \ci_k} s(\tilde{W}_{ti};\theta_0,\hat{\nu}^{(-k)},\hat{e}^{(-k)}) - s(\tilde{W}_{ti};\theta_0,\nu_0,e_{0,N}) \\
&\phe - \e[s(\tilde{W}_{ti};\theta_0,\hat{\nu}^{(-k)},\hat{e}^{(-k)})-s(\tilde{W}_{ti};\theta_0,\nu_0,e_{0,N}) \mid \cs^{(-k)}].
\end{align*}
To see that $\bar{B}_2^{(k)} = o_p(1)$,
we write
\begin{align*}
\bar{B}_2^{(k)} & = N_k^{1/2} \sum_{t=1}^T \frac{n_{t,k}}{N_k} \int s(w;\theta_0,\hat{\nu}^{(-k)},\hat{e}^{(-k)})-s(w;\theta_0,\nu_0,e_{0,N}) \mrd P_t(w) \\
& = N_k^{1/2} \int s(w;\theta_0,\hat{\nu}^{(-k)},\hat{e}^{(-k)})-s(w;\theta_0,\nu_0,e_{0,N}) \mrd P_{0,N}(w) \\
&\phe + N_k^{1/2} \sum_{t=1}^T \left(\frac{n_{t,k}}{N_k}-\frac{N_t}{N}\right)\int s(w;\theta_0,\hat{\nu}^{(-k)},\hat{e}^{(-k)})-s(w;\theta_0,\nu_0,e_{0,N}) \mrd P_t(w).
\end{align*}
Now 
\[
f_N^{(k)}(\lambda) = \int s(w;\theta_0,\nu_0+\lambda(\hat{\nu}^{(-k)}-\nu_0),e_{0,N}+\lambda(\hat{e}^{(-k)}-e_{0,N}))-s(w;\theta_0,\nu_0,e_{0,N}) \mrd P_{0,N}(w)
\]
is twice continuously differentiable on $[0,1]$ by regularity.
Hence by Taylor's theorem
\begin{equation}
\label{eq:f_N_k_lambda}
\|f_N^{(k)}(1)-f_N^{(k)}(0)-f_N^{(k)\prime}(0)\| \leq \frac{1}{2} \sup_{\lambda \in (0,1)} \|f_N^{(k)\prime\prime}(\lambda)\|.
\end{equation}
By~\eqref{eq:vanishing_second_derivatives} we know that
\begin{align*}
\sup_{\lambda \in (0,1)} \|f_N^{(k)\prime\prime}(\lambda)\| \indic(\ce_{N,k}) & \leq \sup_{\lambda \in (0,1)} \sup_{(\nu,e) \in \ct_N} \bigg\|\frac{\partial^2}{\partial \lambda^2} \e_{0,N}[s(W;\theta_0,\nu_0+\lambda(\nu-\nu_0),e_{0,N}+\lambda(e-e_{0,N}))]\bigg\| \\
& \leq N^{-1/2}\delta_N.
\end{align*}
With $f_N^{(k)}(0)=0$,
by~\eqref{eq:neyman_orthogonality} and~\eqref{eq:f_N_k_lambda} we conclude that that
\[
\|f_N^{(k)}(1)\|\bm{1}(\ce_{N,K}) \leq \frac{3}{2}N^{-1/2} \delta_N.
\]
Recalling that $\Pr(\mathcal{E}_{N,k}) \rightarrow 1$,
we have $\|f_N^{(k)}(1)\| = o_p(N^{-1/2})$ and
\[
\bigg\|N_k^{1/2} \int s(w;\theta_0,\hat{\nu}^{(-k)},\hat{e}^{(-k)})-s(w;\theta_0,\nu_0,e_{0,N}) \,\mrd P_{0,N}(w)\bigg\| = N_k^{1/2} \|f_N^{(k)}(1)\| = o_p(1).
\]
Using $N_t/N\to\kappa_t$ from~\eqref{eq:prop_asymp_limit} we have
\begin{align*}
&\phe\,\bigg\|N_k^{1/2} \sum_{t=1}^T \left(\frac{n_{t,k}}{N_k}-\frac{N_t}{N}\right)\int s(w;\theta_0,\hat{\nu}^{(-k)},\hat{e}^{(-k)})-s(w;\theta_0,\nu_0,e_{0,N}) \,\mrd P_t(w)\bigg\| \\
&\leq  N_k^{1/2} \sum_{t=1}^T \bigg|\frac{n_{t,k}}{N_k}-\frac{N_t}{N}\bigg| \cdot \frac{N}{N_t} \|f_N^{(k)}(1)\| \\
& = o_p(1)
\end{align*}
as
\begin{equation}
\label{eq:dP_t_dP_0N}
\frac{\mrd P_t}{\mrd P_{0,N}}(w) \leq \frac{N}{N_t}, \quad t=1,\ldots,T, \quad w \in \cw.
\end{equation}
Thus we have shown $\bar{B}_2^{(k)} = o_p(1)$.
Finally, for each batch $t=1,\ldots,T$,
the quantity $\tilde{B}_{2,t}^{(k)}$ is a sum of $n_{t,k}$ random variables that are i.i.d.\ and mean 0 conditional on $\cs^{(-k)}$.
Thus by Lemma~\ref{lemma:mean_norm_sq},
we have
\begin{align*}
\e\bigl[\|\tilde{B}_{2,t}^{(k)}\|^2 \mid \cs^{(-k)}\bigr] & = \sum_{i:(t,i) \in \ci_k} \frac{1}{n_{t,k}} \bar{\e}\bigl[\|r_{ti}^{(k)}\|^2\bigr] \\
& \leq \sum_{i:(t,i) \in \ci_k} \frac{1}{n_{t,k}} \e\bigl[\|s(W_{ti};\theta_0,\hat{\nu}^{(-k)},\hat{e}^{(-k)})-s(W_{ti};\theta_0,\nu_0,e_{0,N})\|^2 \mid \cs^{(-k)}\bigr]
\end{align*}
where for each $i$ such that $(t,i) \in \ci_k$ we've defined
\begin{align}
\begin{split}
r_{ti}^{(k)} &= s(W_{ti};\theta_0,\hat{\nu}^{(-k)},\hat{e}^{(-k)})-s(\tilde{W}_{ti};\theta_0,\nu_0,e_{0,N}) \\&\phe\,- \e[ s(\tilde{W}_{ti};\theta_0,\hat{\nu}^{(-k)},\hat{e}^{(-k)})-s(\tilde{W}_{ti};\theta_0,\nu_0,e_{0,N}) \mid \cs^{(-k)}]
\end{split}\end{align}
and used the basic variance inequality
\begin{equation}
\label{eq:basic_var_inequality}
\e\bigl[\|X-\e[X \mid \cs]\|^2\bigr] \leq \e\bigl[\|X\|^2 \mid \cs\bigr]
\end{equation}
for any random vector $X$ and $\sigma$-algebra $\cs$.
Hence by~\eqref{eq:dP_t_dP_0},
Minkowski's inequality,
and then~\eqref{eq:s_consistency},
we have
\begin{align*}
\bigl(\e\bigl[\|\tilde{B}_{2,t}^{(k)}\|^2 \mid \cs^{(-k)}\bigr]\bigr)^{1/2}\indic(\ce_{N,k}) & \leq \sup_{(\nu,e) \in \ct_N} \bigl(\e_t\bigl[\|s(W;\theta_0,\nu,e)-s(W;\theta_0,\nu_0,e_{0,N})\|^2\bigr]\bigr)^{1/2} \\
& \leq  \kappa_t^{-1/2} \sup_{(\nu,e) \in \ct_N} \bigl(\e_{0,N}\bigl[\|s(W;\theta_0,\nu,e)-s(W;\theta_0,\nu_0,e_0)\|^2\bigr]\bigr)^{1/2} \\
&\phe\, + \kappa_t^{-1/2} \bigl(\e_{0,N}\bigl[\|s(W;\theta_0,\nu_0,e_{0,N})-s(W;\theta_0,\nu_0,e_0)\|^2\bigr]\bigr)^{1/2} \\
& \leq 2\kappa_t^{-1/2}\delta_N.
\end{align*}
Thus $\tilde{B}_{2,t}^{(k)}=o_p(1)$ by Lemma~\ref{lemma:conditional_oh_pee}.
With
\[
-B_2 = \sum_{k=1}^K \sqrt{\frac{N_k}{N}} B_2^{(k)} = \sum_{k=1}^K \sqrt{\frac{N_k}{N}}\biggl(\bar{B}_2^{(k)}+\sum_{t=1}^T \sqrt{\frac{n_{t,k}}{N_k}} \tilde{B}_{2,t}^{(k)}\biggr)
\]
we conclude that $B_2=o_p(1)$, as desired.
This establishes that $\hat\theta^*=\tilde\theta+o_p(N^{-1/2})$.

\subsubsection{Showing that $\hat\theta =\hat\theta^*+o_p(N^{-1/2})$}

To establish that $\hat{\theta}=\tilde{\theta}+o_p(N^{-1/2})$,
similar to above we write
\[
N^{1/2}(\hat{\theta}-\tilde{\theta}) = N^{1/2}(\hat{\theta}-\theta_0)-N^{1/2}(\tilde{\theta}-\theta_0) = A_3B_3+A_4B_4
\]
where
\begin{align*}
A_3 & = \biggl(\frac{1}{N}\sum_{k=1}^K \sum_{(t,i) \in \ci_k} s_a(\tilde{W}_{ti};\hat{\nu}^{(-k)},\hat{e}^{(-k)})\biggr)^{-1} - \biggl(\frac{1}{N}\sum_{k=1}^K \sum_{(t,i) \in \ci_k} s_a(W_{ti};\hat{\nu}^{(-k)},\hat{e}^{(-k)})\biggr)^{-1}, \\
B_3 & = \frac{1}{\sqrt{N}} \sum_{k=1}^K \sum_{(t,i) \in \ci_k} s(\tilde{W}_{ti};\theta_0,\hat{\nu}^{(-k)},\hat{e}^{(-k)}), \\
A_4 & = \biggl(\frac{1}{N}\sum_{k=1}^K \sum_{(t,i) \in \ci_k} s_a(W_{ti};\hat{\nu}^{(-k)},\hat{e}^{(-k)})\biggr)^{-1},\quad\text{and} \\
B_4 & = \frac{1}{\sqrt{N}} \sum_{k=1}^K \sum_{(t,i) \in \ci_k}[s(\tilde{W}_{ti};\theta_0,\hat{\nu}^{(-k)},\hat{e}^{(-k)})-s(W_{ti};\theta_0,\hat{\nu}^{(-k)},\hat{e}^{(-k)})].
\end{align*}
We will show that $A_3=o_p(1)$, $B_3=O_p(1)$, $A_4=O_p(1)$, and $B_4=o_p(1)$.
First we define
\[
\tilde{A}_3 = \sum_{k=1}^K \frac{N_k}{N} \tilde{A}_3^{(k)}
\]
where the term $\tilde{A}_3^{(k)}$ for each fold is decomposed into a sum over batches $t=1,\ldots,T$:
\begin{align*}
\tilde{A}_3^{(k)} & = \frac{1}{N_k} \sum_{(t,i) \in \ci_k}[s_a(\tilde{W}_{ti};\hat{\nu}^{(-k)},\hat{e}^{(-k)})-s_a(W_{ti};\hat{\nu}^{(-k)},\hat{e}^{(-k)})] \\
& = \frac{1}{N_k} \sum_{(t,i) \in \ci_k} (\tilde{Z}_{ti}-Z_{ti})(s_a(W_{ti}(1);\hat{\nu}^{(-k)},\hat{e}^{(-k)})-s_a(W_{ti}(0);\hat{\nu}^{(-k)},\hat{e}^{(-k)})) \\
& = \sum_{t=1}^T \frac{n_{t,k}}{N_k} \tilde{A}_{3,t}^{(k)}.
\end{align*}
Here
\begin{align*}
\tilde{A}_{3,t}^{(k)} & := \frac{1}{n_{t,k}} \sum_{i:(t,i) \in \ci_k} (\tilde{Z}_{ti}-Z_{ti})(s_a(W_{ti}(1);\hat{\nu}^{(-k)},\hat{e}^{(-k)})-s_a(W_{ti}(0);\hat{\nu}^{(-k)},\hat{e}^{(-k)}))
\end{align*}
satisfies
\begin{align*}
\|\tilde{A}_{3,t}^{(k)}\| & \leq \frac{1}{n_{t,k}} \sum_{i:(t,i) \in \ci_k} \bm{1}(Z_{ti} \neq \tilde{Z}_{ti}) \cdot \|s_a(W_{ti}(1);\hat{\nu}^{(-k)},\hat{e}^{(-k)})-s_a(W_{ti}(0);\hat{\nu}^{(-k)},\hat{e}^{(-k)})\| \\
& \leq \sqrt{\frac{1}{n_{t,k}} \sum_{i:(t,i) \in \ci_k}  \|s_a(W_{ti}(1);\hat{\nu}^{(-k)},\hat{e}^{(-k)})-s_a(W_{ti}(0);\hat{\nu}^{(-k)},\hat{e}^{(-k)})\|^2} \\
&\phe\,\cdot  \sqrt{\frac{1}{n_{t,k}} \sum_{i:(t,i) \in \ci_k} \indic(Z_{ti} \neq \tilde{Z}_{ti})}
\end{align*}
by the Cauchy-Schwarz inequality.
Recalling the $\sigma$-algebra $\cs_t^{X,(k)}$ in the definition of a CSBAE (Definition~\ref{def:CSBAE}),
we have
\[
\Pr(Z_{ti} \neq \tilde{Z}_{ti} \mid \cs_t^{X,(k)}) = |\hat{e}_t^{(k)}(X_{ti})-e_t(X_{ti})|
\]
and so by~\eqref{eq:csbae_limit}
\begin{align*}
\e\biggl[\frac{1}{n_{t,k}} \sum_{i:(t,i) \in \ci_k} \indic(Z_{ti} \neq \tilde{Z}_{ti}) \Bigm| \cs_t^{X,(k)} \biggr] & = \frac{1}{n_{t,k}} \sum_{i:(t,i) \in \ci_k} |\hat{e}_t^{(k)}(X_{ti})-e_t(X_{ti})| \\
& \leq \|\hat{e}_t^{(k)}-e_t\|_{2,P_{N,t}^{X,(k)}} \\
&= o_p(1).
\end{align*}
Hence 
\[
\sqrt{\frac{1}{n_{t,k}} \sum_{i:(t,i) \in \ci_k} \indic(Z_{ti} \neq \tilde{Z}_{ti})} = o_p(1)
\]
by Lemma~\ref{lemma:conditional_oh_pee}.
Next, for $z=0,1$ and $(t,i) \in \ci_k$,
by Jensen's inequality and~\eqref{eq:dP_t_dP_0} we have
\begin{align*}
\left(\e\bigl[\|s_a(W_{ti}(z);\hat{\nu}^{(-k)},\hat{e}^{(-k)})\|^2 \bm{1}(\mathcal{E}_{N,k}) \mid \cs^{(-k)}\bigr]\right)^{q/2} & \leq \sup_{(\nu,e) \in \ct_N} (\e_t\bigl[\|s_a(W(z);\nu,e)\|^2\bigr])^{q/2} \\
& \leq \kappa_t^{-1}\sup_{(\nu,e) \in \ct_N} \e_0[\|s_a(W(z);\nu,e)\|^q] \\
& \leq \kappa_t^{-1}C^q 
\end{align*}
by~\eqref{eq:s_moment_condition}.
By Markov's inequality,
for each batch $t=1,\ldots,T$ and fold $k=1,\ldots,K$,
\[
\sqrt{\frac{1}{n_{t,k}} \sum_{i:(t,i) \in \ci_k} \|s_a(W_{ti}(z);\hat{\nu}^{(-k)},\hat{e}^{(-k)})\|^2} = O_p(1)
\]
holds for $z=0,1$.
Applying Minkowski's inequality with the $L^2_{P_{N,t}^X}$ norm shows
\[
\sqrt{\frac{1}{n_{t,k}} \sum_{i:(t,i) \in \ci_k}  \|s_a(W_{ti}(1);\hat{\nu}^{(-k)},\hat{e}^{(-k)})-s_a(W_{ti}(0);\hat{\nu}^{(-k)},\hat{e}^{(-k)})\|^2} = O_p(1).
\]
We conclude $\tilde{A}_{3,t}^{(k)} = o_p(1)$,
hence $\tilde{A}_3^{(k)} = o_p(1)$ 
and $\tilde{A}_3 = o_p(1)$ as well.
Now we recall
\begin{align*}
\frac{1}{N} \sum_{k=1}^K \sum_{(t,i) \in \ci_k} s_a(\tilde{W}_{ti};\hat{\nu}^{(-k)},\hat{e}^{(-k)}) 
& = \frac{1}{N} \sum_{k=1}^K \sum_{(t,i) \in \ci_k} s_a(\tilde{W}_{ti};\nu_0,e_{0,N}) + o_p(1) \text{\quad  (by~\eqref{eq:A_1_tilde})} \\
& = \e_0[s_a(W;\nu_0,e_0)] + o_p(1) \text{\quad  (by~\eqref{eq:oracle_s_a_lln_folds})}.
\end{align*}
Then 
\begin{align*}
\frac{1}{N} \sum_{k=1}^K \sum_{(t,i) \in \ci_k} s_a(W_{ti};\hat{\nu}^{(-k)},\hat{e}^{(-k)}) & = \frac{1}{N} \sum_{k=1}^K \sum_{(t,i) \in \ci_k} s_a(\tilde{W}_{ti};\hat{\nu}^{(-k)},\hat{e}^{(-k)}) - \tilde{A}_3 \\
& =  \e_0[s_a(W;\nu_0,e_0)] + o_p(1)
\end{align*}
as well.
Thus we have both
\begin{align}
\biggl(\frac{1}{N}\sum_{k=1}^K \sum_{(t,i) \in \ci_k} s_a(\tilde{W}_{ti};\hat{\nu}^{(-k)},\hat{e}^{(-k)})\biggr)^{-1} & = (\e_0[s_a(W;\nu_0,e_0)])^{-1} + o_p(1),\quad\text{and} \nonumber \\
\biggl(\frac{1}{N}\sum_{k=1}^K \sum_{(t,i) \in \ci_k} s_a(W_{ti};\hat{\nu}^{(-k)},\hat{e}^{(-k)})\biggr)^{-1} & = (\e_0[s_a(W;\nu_0,e_0)])^{-1} + o_p(1). \label{eq:A_4}
\end{align}
Subtracting these two equations gives $A_3=o_p(1)$.

Next, equation~\eqref{eq:A_4} immediately provides $A_4=O_p(1)$,
while $B_3=B_1-B_2=O_p(1)$ as shown above.
Thus it only remains to show that $B_4=o_p(1)$.
We write
\[
B_4 = \sum_{k=1}^K \sqrt{\frac{N_k}{N}}\Bigl(B_4^{(k),1} +  B_4^{(k),2} +  B_4^{(k),3}\Bigr)
\]
where
\begin{align*}
B_4^{(k),1} & = \sum_{t=1}^T \sqrt{\frac{n_{t,k}}{N_k}} \frac{1}{\sqrt{n_{t,k}}} \sum_{i:(t,i) \in \ci_k} (\tilde{Z}_{ti}-Z_{ti})\bigl(s(W_{ti}(1);\theta_0,\nu_0,e_{0,N})-s(W_{ti}(0);\theta_0,\nu_0,e_{0,N})\bigr), \\
B_4^{(k),2} & = N_k^{1/2} \sum_{t=1}^T \frac{n_{t,k}}{N_k} \frac{1}{n_{t,k}} \sum_{i:(t,i) \in \ci_k} (\tilde{Z}_{ti}-Z_{ti})\e[s(W_{ti}(1);\theta_0,\hat{\nu}^{(-k)},\hat{e}^{(-k)})-s(W_{ti}(1);\theta_0,\nu_0,e_{0,N}) \mid \cs^{(-k)},X_{ti}]\\
&\phe\, -  N_k^{1/2} \sum_{t=1}^T \frac{n_{t,k}}{N_k} \frac{1}{n_{t,k}} \sum_{i:(t,i) \in \ci_k} (\tilde{Z}_{ti}-Z_{ti})\e[s(W_{ti}(0);\theta_0,\hat{\nu}^{(-k)},\hat{e}^{(-k)})-s(W_{ti}(0);\theta_0,\nu_0,e_{0,N}) \mid \cs^{(-k)},X_{ti}] \\
\intertext{and}
B_4^{(k),3} & = \sum_{t=1}^T \sqrt{\frac{n_{t,k}}{N_k}}\frac{1}{\sqrt{n_{t,k}}}\sum_{i:(t,i) \in \ci_k} (\tilde{Z}_{ti}-Z_{ti})\bigl(s(W_{ti}(1);\theta_0,\hat{\nu}^{(-k)},\hat{e}^{(-k)})-s(W_{ti}(0);\theta_0,\hat{\nu}^{(-k)},\hat{e}^{(-k)})\bigr) \\
& - B_4^{(k),1} - B_4^{(k),2}.
\end{align*}
It suffices to prove that for each $k=1,\ldots,K$,
the terms $B_4^{(k),1}$, $B_4^{(k),2}$, and $B_4^{(k),3}$ are all asymptotically negligible.
First, we have
\begin{align*}
\e&\left[(\tilde{Z}_{ti}-Z_{ti})(s(W_{ti}(1);\theta_0,\nu_0,e_{0,N})-s(W_{ti}(0);\theta_0,\nu_0,e_{0,N})) \mid \cs_t^{X,(k)}\right] \\
& = (e_t(X_{ti})-\hat{e}_t^{(k)}(X_{ti}))\e[s(W_{ti}(1);\theta_0,\nu_0,e_{0,N})-s(W_{ti}(0);\theta_0,\nu_0,e_{0,N}) \mid \cs_t^{X,(k)}] \\
& = (e_t(X_{ti})-\hat{e}_t^{(k)}(X_{ti}))\e[s(W_{ti}(1);\theta_0,\nu_0,e_{0,N})-s(W_{ti}(0);\theta_0,\nu_0,e_{0,N}) \mid X_{ti}] \\
& = 0.
\end{align*}
We briefly justify each of the equalities in the preceding display:
\begin{itemize}
    \item The first equality holds because given $\cs_t^{X,(k)}$,
    the only randomness in $(Z_{ti},\tilde{Z}_{ti})$ for any $(t,i) \in \ci_k$ is in $U_{ti}$, which is independent of $(W_{ti}(0),W_{ti}(1))$, $\cs^{(-k)}$, and $\cs_t^{X,(k)}$; hence
\begin{equation}
\label{eq:Z_cond_ind}
\begin{split}
(Z_{ti},\tilde{Z}_{ti}) &\indep (W_{ti}(0),W_{ti}(1)) \mid \cs_t^{X,(k)}\quad\text{ and } \\
(Z_{ti},\tilde{Z}_{ti}) &\indep (W_{ti}(0),W_{ti}(1)) \mid \cs^{(-k)},\cs_t^{X,(k)}.
\end{split}
\end{equation}
    \item The second equality holds because
    the vectors $\{S_{ti}, 1 \leq t \leq T, 1 \leq i \leq N_t\}$ are mutually independent (Assumption~\ref{assump:gen_DGP}).
    \item The third equality follows directly from~\eqref{eq:potential_score_diff}.
\end{itemize}
For batches $t=1,\ldots,T$, 
the vectors $\{(\tilde{Z}_{ti},Z_{ti},W_{ti}(1),W_{ti}(0)), 1 \leq i \leq N_t\}$ are conditionally independent given $\cs_t^{X,(k)}$.
Defining
\[
\ci_{4,t}^{(k)} = \frac{1}{\sqrt{n_{t,k}}} \sum_{i:(t,i) \in \ci_k} (\tilde{Z}_{ti}-Z_{ti})(s(W_{ti}(1);\theta_0,\nu_0,e_{0,N})-s(W_{ti}(0);\theta_0,\nu_0,e_{0,N}))
\]
and letting $p$ be the real solution to $p^{-1}+2q^{-1}=1$,
we apply Lemma~\ref{lemma:mean_norm_sq} and Holder's inequality to get
\begin{align*}
&\phe\,\e[\|\ci_{4,t}^{(k)}\|^2 \mid \cs_t^{X,(k)}] \\
& = \frac{1}{n_{t,k}} \sum_{i:(t,i) \in \ci_k} \e[(\tilde{Z}_{ti}-Z_{ti})^2\|s(W_{ti}(1);\theta_0,\nu_0,e_{0,N})-s(W_{ti}(0);\theta_0,\nu_0,e_{0,N})\|^2 \mid \cs_t^{X,(k)}] \\
& = \frac{1}{n_{t,k}} \sum_{i:(t,i) \in \ci_k} |\hat{e}_t^{(k)}(X_{ti})-e_t(X_{ti})|\e\big[\|s(W_{ti}(1);\theta_0,\nu_0,e_{0,N})-s(W_{ti}(0);\theta_0,\nu_0,e_{0,N})\|^2 \mid \cs_t^{X,(k)}\big] 
\end{align*}
using the conditional independence in~\eqref{eq:Z_cond_ind} once again.
Then
\begin{align*}
\e[\|\ci_{4,t}^{(k)}\|^2 \mid \cs_t^{X,(k)}] & \leq \biggl(\frac{1}{n_{t,k}} \sum_{i:(t,i) \in \ci_k} |\hat{e}_t^{(k)}(X_{ti})-e_t(X_{ti})|^p\biggr)^{1/p}  \\
& \cdot \biggl(\frac{1}{n_{t,k}} \sum_{i:(t,i) \in \ci_k} \left(\e[\|s(W_{ti}(1);\theta_0,\nu_0,e_{0,N})-s(W_{ti}(0);\theta_0,\nu_0,e_{0,N})\|^2 \mid \cs_t^{X,(k)}]\right)^{q/2}\biggr)^{2/q}.
\end{align*}
We can assume WLOG that $q < 4$ so that $p>2$.
Then with $0 \leq |\hat{e}_t^{(j)}(x)-e_t(x)| \leq 1$ for all $x$,
we conclude
\begin{equation}
\label{eq:e_hat_t_j_p_convergence}
\frac{1}{n_{t,k}} \sum_{i:(t,i) \in \ci_k} \big|\hat{e}_t^{(k)}(X_{ti})-e_t(X_{ti})\big|^p \leq \frac{1}{n_{t,k}} \sum_{i:(t,i) \in \ci_k} \bigl(\hat{e}_t^{(k)}(X_{ti})-e_t(X_{ti})\bigr)^2 = o_p(1)
\end{equation}
by~\eqref{eq:csbae_limit}.
Next, from Jensen's inequality
\begin{align*}
&\phe\, \frac{1}{n_{t,k}} \sum_{i:(t,i) \in \ci_k} \left(\e\bigl[\|s(W_{ti}(1);\theta_0,\nu_0,e_{0,N})-s(W_{ti}(0);\theta_0,\nu_0,e_{0,N})\|^2 \mid \cs_t^{X,(k)}\bigr]\right)^{q/2} \\
&\leq  \frac{1}{n_{t,k}} \sum_{i:(t,i) \in \ci_k} \e\left[\|s(W_{ti}(1);\theta_0,\nu_0,e_{0,N})-s(W_{ti}(0);\theta_0,\nu_0,e_{0,N})\|^q \mid \cs_t^{X,(k)}\right] \\
&= \e_t\bigl[\|s(W(1);\theta_0,\nu_0,e_{0,N})-s(W(0);\theta_0,\nu_0,e_{0,N})\|^q \mid X\bigr].
\end{align*}
But by Minkowski's inequality,~\eqref{eq:dP_t_dP_0}, and~\eqref{eq:s_moment_condition}, we see
\begin{align*}
&\phe\ \bigl(\e_t[\|s(W(1);\theta_0,\nu_0,e_{0,N})-s(W(0);\theta_0,\nu_0,e_{0,N})\|^q]\bigr)^{1/q} \\
& \leq \bigl(\e_t[\|s(W(1);\theta_0,\nu_0,e_{0,N})\|^q])^{1/q} + (\e_t[\|s(W(0);\theta_0,\nu_0,e_{0,N})\|^q]\bigr)^{1/q} \\
& \leq \kappa_t^{-1/q}[(\e_0[\|s(W(1);\theta_0,\nu_0,e_{0,N})\|^q])^{1/q} + (\e_0[\|s(W(0);\theta_0,\nu_0,e_{0,N})\|^q])^{1/q} ] \\
& \leq 2C\kappa_t^{-1/q}.
\end{align*}
We conclude by Markov's inequality that
\[
\frac{1}{n_{t,k}} \sum_{i:(t,i) \in \ci_k} \left(\e[\|s(W_{ti}(1);\theta_0,\nu_0,e_{0,N})-s(W_{ti}(0);\theta_0,\nu_0,e_{0,N})\|^2 \mid \cs_t^{X,(k)}]\right)^{q/2} = O_p(1).
\]
Along with~\eqref{eq:e_hat_t_j_p_convergence} we can conclude that $\ci_{4,t}^{(k)} = o_p(1)$ by Lemma~\ref{lemma:conditional_oh_pee}.
Then also
\[
B_4^{(k),1} = \sum_{t=1}^T \sqrt{\frac{n_{t,k}}{N_k}} \ci_{4,t}^{(k)} = o_p(1)
\]

Next, for $z=0,1$ define
\[
B_4^{(k),2}(z) = N_k^{1/2} \sum_{t=1}^T \frac{n_{t,k}}{N_k}\frac{1}{n_{t,k}} \sum_{i:(t,i) \in \ci_k} (\tilde{Z}_{ti}-Z_{ti})\e\bigl[s(W_{ti}(z);\theta_0,\hat{\nu}^{(-k)},\hat{e}^{(-k)})-s(W_{ti}(z);\theta_0,\nu_0,e_{0,N}) \mid \cs^{(-k)},X_{ti}\bigr]
\]
so that $B_4^{(k),2} = B_4^{(k),2}(1)-B_4^{(k),2}(0)$.
Then by the triangle inequality $\|B_4^{(k),2}(z)\|$
is no larger than
\[
 N_k^{1/2} \sum_{t=1}^T \frac{n_{t,k}}{N_k} \frac{1}{n_{t,k}} \sum_{i:(t,i) \in \ci_k}|\tilde{Z}_{ti}-Z_{ti}| \cdot \|\e[s(W_{ti}(z);\theta_0,\hat{\nu}^{(-k)},\hat{e}^{(-k)})-s(W_{ti}(z);\theta_0,\nu_0,e_{0,N}) \mid \cs^{(-k)},X_{ti}]\|.
\]
Taking conditional expectations of both sides yields
\begin{align*}
&\e\bigl[\|B_4^{(k),2}(z)\| \mid \cs^{(-k)},\cs_t^{X,(k)}\bigr] \\
& \leq N_k^{1/2} \sum_{t=1}^T \frac{n_{t,k}}{N_k}\frac{1}{n_{t,k}} \sum_{i:(t,i) \in \ci_k} \bigl|\hat{e}_t^{(k)}(X_{ti})-e_t(X_{ti})\bigr| \\
&\qquad\qquad\qquad\times \bigl\Vert\e[s(W_{ti}(z);\theta_0,\hat{\nu}^{(-k)},\hat{e}^{(-k)})-s(W_{ti}(z);\theta_0,\nu_0,e_{0,N}) \mid \cs^{(-k)},X_{ti}]\bigr\Vert \\
& \leq N_k^{1/2} \sum_{t=1}^T \frac{n_{t,k}}{N_k} \cdot S_t^{(-k)}(z) \cdot \sqrt{\frac{1}{n_{t,k}} \sum_{i:(t,i) \in \ci_k} (\hat{e}_t^{(k)}(X_{ti})-e_t(X_{ti}))^2} 
\end{align*}
by Cauchy-Schwarz.
We have by~\eqref{eq:csbae_limit} that
\[
\sqrt{\frac{1}{n_{t,k}} \sum_{i:(t,i) \in \ci_k} (\hat{e}_t^{(k)}(X_{ti})-e_t(X_{ti}))^2} = O_p(N^{-1/4}),
\]
Then by equation~\eqref{eq:S_t_neg_k},
we get $B_4^{(k),2}(z)=o_p(1)$ for $z=0,1$,
and so $B_4^{(k),2} = o_p(1)$ as well.

Finally, we write
\[
B_4^{(k),3} = \sum_{t=1}^T \sqrt{\frac{n_{t,k}}{N_k}} B_{4,t}^{(k),3}
\]
for
\begin{align*}
B_{4,t}^{(k),3} & = \frac{1}{\sqrt{n_{t,k}}} \sum_{i:(t,i) \in \ci_k} (\tilde{Z}_{ti}-Z_{ti})\Delta_{ti}^{(-k)},\\
\intertext{where}
\Delta_{ti}^{(-k)} & = (s(W_{ti}(1);\theta_0,\hat{\nu}^{(-k)},\hat{e}^{(-k)})-s(W_{ti}(0);\theta_0,\hat{\nu}^{(-k)},\hat{e}^{(-k)})) \\
&\phe\ - (s(W_{ti}(1);\theta_0,\nu_0,e_{0,N})-s(W_{ti}(0);\theta_0,\nu_0,e_{0,N})) \\
&\phe\ - \e[s(W_{ti}(1);\theta_0,\hat{\nu}^{(-k)},\hat{e}^{(-k)})-s(W_{ti}(0);\theta_0,\hat{\nu}^{(-k)},\hat{e}^{(-k)}) \mid \cs^{(-k)},X_{ti}] \\
&\phe\ - \e[s(W_{ti}(1);\theta_0,\nu_0,e_0)-s(W_{ti}(0);\theta_0,\nu_0,e_0) \mid  \cs^{(-k)},X_{ti}].
\end{align*}
For each $(t,i) \in \ci_k$,
we have
\begin{equation}
\label{eq:W_cond_ind}
(W_{ti}(0),W_{ti}(1)) \indep \cs_t^{X,(k)} \mid \cs^{(-k)},X_{ti}
\end{equation}
since given $\cs^{(-k)}$ and $X_{ti}$,
the only remaining randomness in $(W_{ti}(0),W_{ti}(1))$ is in the potential outcomes $(Y_{ti}(0),Y_{ti}(1))$.
These are independent of both $\cs^{(-k)}$ and $\{X_{tj} \mid j \neq i\}$,
the covariates of the other subjects in batch $t$.
Thus $\e[\Delta_{ti}^{(-k)} \mid \cs^{(-k)},\cs_t^{X,(k)}] = 0$,
and by~\eqref{eq:Z_cond_ind} we have
\[
\e[B_{4,t}^{(k),3} \mid \cs^{(-k)},\cs_t^{X,(k)}] = 0.
\]
Then applying Lemma~\ref{lemma:mean_norm_sq} and equations
\eqref{eq:Z_cond_ind} and~\eqref{eq:basic_var_inequality},
\begin{align*}
\e[\|B_{4,t}^{(k),3}\|^2 \mid \cs^{(-k)},\cs_t^{X,(k)}] & = \frac{1}{n_{t,k}} \sum_{i:(t,i) \in \ci_k} \e[(\tilde{Z}_{ti}-Z_{ti})^2\|\Delta_{ti}^{(-k)}\|^2 \mid \cs^{(-k)},\cs_t^{X,(k)}] \\
& = \frac{1}{n_{t,k}} \sum_{i:(t,i) \in \ci_k} |\hat{e}_t^{(k)}(X_{ti})-e_t(X_{ti})| \e[\|\Delta_{ti}^{(-k)}\|^2 \mid \cs^{(-k)},\cs_t^{X,(k)}] \\
& \leq \frac{1}{n_{t,k}} \sum_{i:(t,i) \in \ci_k} |\hat{e}_t^{(k)}(X_{ti})-e_t(X_{ti})| \e[\|\tilde{\Delta}_{ti}^{(-k)}\|^2 \mid \cs^{(-k)},\cs_t^{X,(k)}] \\
& \leq \biggl(\frac{1}{n_{t,k}} \sum_{i:(t,i) \in \ci_k} |\hat{e}_t^{(k)}(X_{ti})-e_t(X_{ti})|^p\biggr)^{1/p} \Gamma_t^{(k)} 
\end{align*}
by Holder's inequality, where $p$ is as above for
\begin{align*}
\Gamma_t^{(k)} & = \biggl(\frac{1}{n_{t,k}} \sum_{i:(t,i) \in \ci_k} (\e[\|\tilde{\Delta}_{ti}^{(-k)}\|^2 \mid \cs^{(-k)},\cs_t^{X,(k)}])^{q/2}\biggr)^{2/q},\quad\text{with} \\
\tilde{\Delta}_{ti}^{(-k)} & =  (s(W_{ti}(1);\theta_0,\hat{\nu}^{(-k)},\hat{e}^{(-k)})-s(W_{ti}(0);\theta_0,\hat{\nu}^{(-k)},\hat{e}^{(-k)}))\\ 
&\phe\ - (s(W_{ti}(1);\theta_0,\nu_0,e_{0,N})-s(W_{ti}(0);\theta_0,\nu_0,e_{0,N})).
\end{align*}
By Jensen's inequality
\[
(\Gamma_t^{(k)})^{q/2} \leq \frac{1}{n_{t,k}} \sum_{i:(t,i) \in \ci_k}\e\bigl[\|\tilde{\Delta}_{ti}^{(-k)}\|^q \mid \cs^{(-k)},\cs_t^{X,(k)}\bigr].
\]
Taking conditional expectations, we get
\[
\e\bigl[(\Gamma_t^{(k)})^{q/2} \mid \cs^{(-k)}\bigr] = \frac{1}{n_{t,k}} \sum_{i:(t,i) \in \ci_k} \e\bigl[\|\tilde{\Delta}_{ti}^{(-k)}\|^q \mid \cs^{(-k)}\bigr].
\]
Then
\begin{align*}
&\phe\ \left(\e[(\Gamma_t^{(k)})^{q/2} \mid \cs^{(-k)}]\right)^{1/q}\bm{1}(\ce_{N,k}) \\
& \leq \sup_{(\nu,e) \in \ct_N} \left(\e_t[\|(s(W(1);\theta_0,\nu,e)-s(W(0);\theta_0,\nu,e)) -(s(W(1);\theta_0,\nu_0,e_{0,N})-s(W(0);\theta_0,\nu_0,e_{0,N}))\|^q]\right)^{1/q} \\
& \leq \sup_{(\nu,e) \in \ct_N} \sum_{z \in \{0,1\}} (\e_t[\|s(W(z);\theta_0,\nu,e)\|^q])^{1/q} + (\e_t[\|s(W(z);\theta_0,\nu_0,e_{0,N})\|^q])^{1/q} \\
& \leq 4C\kappa_t^{-1/q} 
\end{align*}
by~\eqref{eq:dP_t_dP_0} and moment boundedness.
We conclude that $\e[(\Gamma_t^{(k)})^{q/2} \mid \cs^{(-k)}]\bm{1}(\mathcal{E}_{N,k})$ is uniformly bounded.
Recalling that $\Pr(\mathcal{E}_{N,k}) \rightarrow 1$,
Markov's inequality then ensures $\Gamma_t^{(k)} = O_p(1)$.
In view of~\eqref{eq:e_hat_t_j_p_convergence},
we then have $\e[\|B_{4,t}^{(k),3}\|^2 \mid \cs^{(-k)},\cs_t^{X,(k)}] = o_p(1)$.
This implies $B_{4,t}^{(k),3} = o_p(1)$ by Lemma~\ref{lemma:conditional_oh_pee},
and hence $B_4^{(k),3} = o_p(1)$.
This establishes that $\hat{\theta}=\tilde{\theta}+o_p(N^{-1/2})$
and completes the proof.

\subsection{Proof of Corollary~\ref{cor:AIPW}: CLT for feasible $\hat{\theta}_{\aipw}$ in a CSBAE}
\label{proof:cor:AIPW}
Corollary~\ref{cor:AIPW}, which shows the feasible estimator $\hat{\theta}_{\aipw}$ satisfies a CLT for estimating $\hat{\theta}_{0,\ate}$ in a CSBAE,
holds under the numbered generalizations at the end of Appendix~\ref{app:nonstationary},
subject to the additional requirement that the mean functions are stationary,
as in our generalization of Corollary~\ref{cor:ate_oracle_clt} to nonstationary batches.
We prove this more general result.
The proof of our generalized Corollary~\ref{cor:ate_oracle_clt} shows that Assumption~\ref{assump:gen_identification} is satisfied under the moment bounds in Assumption~\ref{assump:ate_regularity} with
$\theta_0=\theta_{0,\ate}$,
$s(\cdot)=s_{\aipw}(\cdot)$,
$\nu_0=\nu_{0,\aipw} \in \cn=\cn_{\aipw}$,
and any $\gamma \in (0,1/2)$.
It remains to show that the further
conditions of Assumption~\ref{assump:dml}, 
namely (a), (b), (c) and equations~\eqref{eq:neyman_orthogonality} through~\eqref{eq:s_moment_condition},
are satisfied.
Then Corollary~\ref{cor:AIPW} follows by Theorem~\ref{thm:batch_clt}.

\begin{description}
\item[Condition~\ref{cond:regularity}] Invertibility of $\e_0[\theta_0(W;\nu_0,e_0)]$ and existence of $\e_0[s(W;\theta_0,\nu_0,e_0)^2]$
follow from Assumption~\ref{assump:ate_regularity},
as shown in the proof of Corollary~\ref{cor:ate_oracle_clt}.
\item[Condition~\ref{cond:score_differentiability}]
Fix $\lambda \in [0,1]$ and $(\nu,e) \in \ct=\cn \times \cf_{\gamma}$,
where $\nu(\cdot)=(m(0,\cdot),m(1,\cdot))$. Consider
\begin{align*}
f_N(\lambda) & := s(W;\theta_0,\nu_0+\lambda(\nu-\nu_0),e_{0,N}+\lambda(e-e_{0,N})) \\
& = m_{\lambda}(1,X)-m_{\lambda}(0,X) + \frac{Z(Y-m_{\lambda}(1,X))}{e_{\lambda,N}(x)} - \frac{(1-Z)(Y-m_{\lambda}(0,X))}{1-e_{\lambda,N}(x)}
\end{align*}
where $m_{\lambda}(z,x) := m_0(z,x) + \lambda(m(z,x)-m_0(z,x))$ and $e_{\lambda,N}(x) := e_{0,N}(x) + \lambda(e(x)-e_{0,N}(x))$.
Taking two derivatives with respect to $\lambda$ we get
\begin{align*}
 f_N'(\lambda)&  = (m(1,X)-m_0(1,X))-(m(0,X)-m_0(0,X)) \\
&\phe\ + Z\left(\frac{m_0(1,X)-m(1,X)}{e_{\lambda,N}(X)}-\frac{(Y-m_{\lambda}(1,X))(e(X)-e_{0,N}(X))}{e_{\lambda,N}(X)^2}\right) \\
&\phe\ - (1-Z)\left(\frac{m_0(0,X)-m(0,X)}{1-e_{\lambda,N}(X)} - \frac{(Y-m_{\lambda}(0,X))(e_{0,N}(X)-e(X))}{(1-e_{\lambda,N}(X))^2}\right),\quad\text{and} \\
f_N''(\lambda) & = 2Z\left(\frac{(Y-m_{\lambda}(1,X))(e(X)-e_{0,N}(X))^2}{e_{\lambda,N}(X)^3}-\frac{(m_0(1,X)-m(1,X))(e(X)-e_{0,N}(X))}{e_{\lambda,N}(X)^2}\right) \\
&\phe\ + 2(1-Z)\left(\frac{m_0(0,X)-m(0,X)(e_{0,N}(X)-e(X))}{(1-e_{\lambda,N}(X))^2} - \frac{(Y-m_{\lambda}(0,X))(e_{0,N}(X)-e(X))^2}{(1-e_{\lambda,N}(X))^3}\right).
\end{align*}

We now show that $f_N(\cdot)$, $f_N'(\cdot)$, and $f_N''(\cdot)$ are upper bounded by an integrable random variable on an open interval containing $[0,1]$,
so that by the Leibniz rule,
we can swap both first and second derivatives with expectations to conclude the function $\lambda \mapsto \e_{0,N}[s(W;\theta_0,\nu_0+\lambda(\nu-\nu_0),e_{0,N}+\lambda(e-e_{0,N}))]$ has second derivative $\e_{0,N}[f_N''(\lambda)]$
for $\lambda \in [0,1]$.
Note $e_{\lambda,N} \in \cf_{\gamma}$.
Then by repeated application of the triangle inequality,
we see that for sufficiently small $\epsilon>0$,
we must have
\begin{align*}
\sup_{\lambda \in (-\epsilon,1+\epsilon)} |f_N(\lambda)| & \leq 2(|m(1,X)| + |m_0(1,X)| + |m(0,X)| + |m_0(0,X)|) \\
& + 2\gamma^{-1}(|Y(1)|+|Y(0)|+|m(1,X)| + |m_0(1,X)| + |m(0,X)| + |m_0(0,X)|) \\
\sup_{\lambda \in (-\epsilon,1+\epsilon)} |f_N'(\lambda)| & \leq 2(|m(1,X)| + |m_0(1,X)| + |m(0,X)| + |m_0(0,X)|) \\
& + 2\gamma^{-2}(|Y(1)|+|Y(0)|+|m(1,X)| + |m_0(1,X)| + |m(0,X)| + |m_0(0,X)|) \\
\sup_{\lambda \in (-\epsilon,1+\epsilon)} |f_N''(\lambda)| & \leq 2\gamma^{-3}(|Y(1)|+|Y(0)|+|m(1,X)| + |m_0(1,X)| + |m(0,X)| + |m_0(0,X)|).
\end{align*}
The right-hand sides above are clearly integrable under $P_{0,N}$.
Hence, the mapping $\lambda \mapsto \e_{0,N}[s(W;\theta_0,\nu_0+\lambda(\nu-\nu_0),e_{0,N}+\lambda(e-e_{0,N}))]$ is indeed twice differentiable
with second derivative $\e_{0,N}[f_N''(\lambda)]$
for $\lambda \in [0,1]$.
This second derivative is continuous for such $\lambda$ by continuity of $f_N''(\cdot)$ and dominated convergence.
\item[Condition~\ref{cond:potential_score_equations}] 
Fix $e \in \cf_{\gamma}$. 
By stationarity of the mean functions $m_0(z,\cdot)$, $z=0,1$, for each $t=1,\ldots,T$, $i=1,\ldots,N_t$ we have 
(letting $\nu_0=\nu_{0,\aipw}$ for brevity)
\begin{align*}
\e[s(W_{ti}(1);\theta_0,\nu_0,e) \mid X_{ti}] & = m_0(1,X)-m_0(0,X) - \theta_0+ \e_t\left[\frac{Y(1)-m_0(1,X)}{e(X)} \Bigm| X\right] \\
& = m_0(1,X)-m_0(0,X) - \theta_0 + (e(X))^{-1} \e_t[Y(1)-m_0(1,X) \giv X] \\
& = m_0(1,X)-m_0(0,X) - \theta_0
\end{align*}
and similarly 
\begin{align*}
\e[s(W_{ti}(0);\theta_0,\nu_0,e) \mid X_{ti}] & = m_0(1,X)-m_0(0,X)-\theta_0-\e_t\left[\frac{Y(0)-m_0(0,X)}{1-e(X)} \Bigm| X\right] \\
& = m_0(1,X)-m_0(0,X)-\theta_0-(1-e(X))^{-1} \e_t[Y(0)-m_0(0,X) \giv X] \\
& = m_0(1,X)-m_0(0,X)-\theta_0.
\end{align*}
Subtracting shows~\eqref{eq:potential_score_diff}.
\end{description}
It remains to show that
the out-of-fold estimators $(\hat{m}^{(-k)}(0,\cdot),\hat{m}^{(-k)}(1,\cdot),\hat{e}^{(-k)})$ lie in some set $\ct_N$ with high probability for all sufficiently large $N$,
where this set $\ct_N$ satisfies equations~\eqref{eq:neyman_orthogonality} through~\eqref{eq:s_moment_condition}.
To construct this $\ct_N$,
we see that
by the rate conditions on the nuisance estimators along with equation~\eqref{eq:e_0_N_vs_e_0},
there exists a sequence $\tilde{\delta}_N \downarrow 0$ so that $\|e_{0,N}-e_0\|_{2,P_0^X} \leq \tilde{\delta}_N$ for all $N$.
Furthermore,
with probability approaching 1 as $N\rightarrow \infty$,
these four conditions hold for $z=0,1$ and all folds $k=1,\ldots,K$:
\begin{align*}
\|\hat{m}^{(-k)}(z,\cdot)-m_0(z,\cdot)\|_{2,P_0^X} + \|\hat{e}^{(-k)}-e_{0,N}\|_{2,P_0^X} & \leq \tilde{\delta}_N,\\
\|\hat{m}^{(-k)}(z,\cdot)\|_{2,P_0^X}
\times \|\hat{e}^{(-k)}-e_{0,N}\|_{2,P_0^X} & \leq N^{-1/2} \tilde{\delta}_N,\\
\|\hat{m}^{(-k)}(z,\cdot)-m_0(z,\cdot)\|_{q,P_0^X} & \leq C, \\
\hat{e}^{(-k)}(\cdot) & \in \cf_{\gamma}.
\end{align*}
We then define $\ct_N$ be the set of functions $(m(0,\cdot),m(1,\cdot),e(\cdot))$ in $\ct=\cn_{\aipw} \times \ce$ 
obeying these conditions:
\begin{align*}
    \|m(z,\cdot)-m_0(z,\cdot)\|_{2,P_0^X} + \|e-e_{0,N}\|_{2,P_0^X} & \leq \tilde{\delta}_N, \\
\|m(z,\cdot)\|_{2,P_0^X} \times \|e-e_{0,N}\|_{2,P_0^X} & \leq N^{-1/2} \tilde{\delta}_N, \\
\|m(z,\cdot)-m_0(z,\cdot)\|_{q,P_0^X} & \leq C, \\
e(\cdot) & \in \cf_{\gamma}.
\end{align*}
By construction, 
for all $k=1,\ldots,K$ we have $\Pr((\hat{m}^{(-k)},\hat{e}^{(-k)}) \in \ct_N) \rightarrow 1$ as $N \rightarrow \infty$.

\medskip 
For the remainder of the proof,
we take
$N$ large enough so that $1/2 \leq {\mrd P_{0,N}^X}/{\mrd P_0^X} \leq 2$.
For such $N$, $\|f\|_{2,P_{0,N}^X} \leq \sqrt{2}\|f\|_{2,P_0^X}$ holds for all $f \in L^2(P_0^X)$.
We will show equations~\eqref{eq:neyman_orthogonality} through~\eqref{eq:s_consistency} hold
for a sequence $\delta_N$ that is some constant multiple of $\tilde{\delta}_N$.

\begin{description}
\item[Equation~\eqref{eq:neyman_orthogonality}] Fix $(\nu,e) \in \ct_N$. 
By the calculations and notation
above in the proof of condition~\ref{cond:score_differentiability}
and interchanging differentiation with expectation, we can verify using unconfoundedness that
\[
\frac{\partial}{\partial \lambda} \e_{0,N}[s(W;\theta_0,\nu_0,\lambda(\nu-\nu_0),e_{0,N}+\lambda(e-e_{0,N})]\Big|_{\lambda=0} = \e_{0,N}[f_N'(0)] = 0
\]
so the left-hand side of the Neyman orthogonality condition~\eqref{eq:neyman_orthogonality} is 0. The full calculation is shown in the proof of Theorem 5.1 of~\citet{chernozhukov2018double}.
\item[Equation~\eqref{eq:vanishing_second_derivatives}] Once again, we fix $(\nu,e) \in \ct_N$ and recall the calculations and notations in the proof of condition~\ref{cond:score_differentiability} above.
We see that
\begin{align*}
\Big|\frac{\partial^2}{\partial \lambda^2} \e_{0,N}[s(W;\theta_0,\nu_0,\lambda(\nu-\nu_0),e_{0,N}+\lambda(e-e_{0,N})]\Big| 
=\bigl|\e_{0,N}[f_N''(\lambda)]\bigr|
\end{align*}
and hence
\begin{align*}
\bigl|\e_{0,N}[f_N''(\lambda)]\bigr|
&\le\frac{2}{\gamma^3}\Bigl(\bigl|\e_{0,N}[(Y(1)-m_{\lambda}(1,X))(e(X)-e_{0,N}(X))^2]\bigr| \\
&\phe\ +\bigl|\e_{0,N}[(Y(0)-m_{\lambda}(0,X))(e_{0,N}(X)-e(X))^2]\bigr|\Bigr) \\
&\phe\ + \frac{2}{\gamma^2}\Bigl(|\e_{0,N}[(m_0(1,X)-m(1,X))(e(X)-e_0(X))]|\\
&\phe\ + \bigl|\e_{0,N}[(m_0(0,X)-m(0,X))(e(X)-e_0(X))]\bigr|\Bigr).
\end{align*}
By Cauchy-Schwarz and the definition of $\ct_N$ we have for $z=0,1$ that
\begin{align*}
&\phe\ \e_{0,N}\bigl[|m_0(z,X)-m(z,X)| \times |e(X)-e_{0,N}(X)|\bigr]\\
&\leq \|m_0(z,\cdot)-m(z,\cdot)\|_{2,P_{0,N}^X} \times \|e-e_{0,N}\|_{2,P_{0,N}^X} \\
&\leq 2N^{-1/2}\tilde{\delta}_N.
\end{align*}
Furthermore for $z=0,1$ we have $\e_{0,N}\bigl[(Y(z)-m_0(z,X))(e(X)-e_{0,N}(X))^2\bigr]=0$ by conditioning on $X$.
Hence, for all $\lambda \in [0,1]$, 
\begin{align*}
&\phe\ \bigl|\e_{0,N}[(Y(z)-m_{\lambda}(z,X))(e(X)-e_{0,N}(X))^2]\bigr| \\
& = \bigl|\e_{0,N}[\lambda(m(z,X)-m_0(z,X))(e(X)-e_{0,N}(X))^2]\bigr| \\
& \leq \lambda \e_{0,N}\bigl[|m(z,X)-m_0(z,X)||e(X)-e_{0,N}(X)|\bigr]  \\
& \leq 2N^{-1/2}\tilde{\delta}_N
\end{align*}
where the first inequality uses the fact that $|e(X)-e_{0,N}(X)| \leq 1$.
Taking suprema over $(\nu,e) \in \ct_N$ and $\lambda \in [0,1]$, we get
\[
\sup_{(\nu,e) \in \ct_N} \sup_{\lambda \in [0,1]} \Big|\frac{\partial^2}{\partial \lambda^2} \e_{0,N}[s(W;\theta_0,\nu_0,\lambda(\nu-\nu_0),e_{0,N}+\lambda(e-e_{0,N})]\Big| \leq N^{-1/2}\left(\frac{8}{\gamma^3}+\frac{8}{\gamma^2}\right)\tilde{\delta}_N
\]
which shows equation~\eqref{eq:vanishing_second_derivatives}.

\item[Equation~\eqref{eq:s_a_consistency}] For any $(\nu,e) \in \ct_N$, trivially
\[
\e_0\bigl[|s_a(W;\nu,e)-s_a(W;\nu_0,e_0)|^2\bigr]=\e_0\bigl[|-1-(-1)|^2\bigr]=0.
\]
\item[Equation~\eqref{eq:s_consistency}]
We fix $(\nu,e) \in \ct_N$ and write
\begin{align*}
s(W;\theta_0,\nu,e)-s(W;\theta_0,\nu_0,e_0) & = \left(1-\frac{Z}{e_0(X)}\right)(m(1,X)-m_0(1,X))\\
&\phe\ +\left(1-\frac{1-Z}{1-e_0(X)}\right)(m_0(0,X)-m(0,X)) \\
&\phe\ + Z(Y-m(1,X))(e(X)^{-1}-e_0(X)^{-1}) \\
&\phe\ - (1-Z)(Y-m(0,X))((1-e(X))^{-1}-(1-e_0(X))^{-1}).
\end{align*}
It now suffices to show that each of the summands has asymptotically vanishing second moment.
The definition of $\ct_N$ ensures that
\begin{align*}
\e_0\biggl[\biggl(1-\frac{Z}{e_0(X)}\biggr)^2(m(1,X)-m_0(1,X))^2\biggr] & \leq (1+\gamma^{-1})^2\tilde{\delta}_N^2\quad\text{and} \\
\e_0\biggl[\biggl(1-\frac{1-Z}{1-e_0(X)}\biggr)^2(m_0(0,X)-m(0,X))^2\biggr] & \leq (1+\gamma^{-1})^2\tilde{\delta}_N^2.
\end{align*}
Next 
\begin{align*}
&\phe\ \,\e_0\bigl[Z^2(Y(1)-m(1,X))^2(e(X)^{-1}-e_0(X)^{-1})^2\bigr] \\
& \leq \gamma^{-4} \e_0\bigl[(Y(1)-m(1,X))^2(e(X)-e_0(X))^2\bigr] \\
& = \gamma^{-4}\e_0\bigl[(e(X)-e_0(X))^2v_0(1,X)\bigr] \\
&\phe\ + \gamma^{-4}\e_0\bigl[(e(X)-e_0(X))^2(m(1,X)-m_0(1,X))^2\bigr] \\
& \leq C\gamma^{-4}\|e-e_0\|_{2,P_0^X}^2 + \gamma^{-4}\|m(1,\cdot)-m_0(1,\cdot)\|_{2,P_0^X}^2 \\
& \leq (2C+1)\gamma^{-4}\tilde{\delta}_N^2.
\end{align*}
By a similar computation
\[
\e_0[(1-Z)^2(Y(0)-m(0,X))^2((1-e(X))^{-1}-(1-e_0(X))^{-1})^2] \leq (2C+1)\gamma^{-4}\tilde{\delta}_N^2.
\]
This completes the proof of~\eqref{eq:s_consistency}.

\item[Equation~\eqref{eq:s_a_moment_condition}]
For any $(\nu,e) \in \ct_N$, trivially $\e\bigl[|s_a(W(z);\nu,e)|^q\bigr]=1$ for $z=0,1$. 
\item[Equation~\eqref{eq:s_moment_condition}]
Fix $(\nu,e) \in \ct_N$. We write
\begin{align}\label{eq:sw1parts}
s(W(1);\theta_0,\nu,e) = m(1,X)-m(0,X)-\theta_0 + \frac{Y(1)-m(1,X)}{e(X)}.
\end{align}
Since $(\e_0[|Y(z)|^q])^{1/q} \leq C$ for $z=0,1$,
\[
\e_0\bigl[|m_0(z,X)|^q\bigr] = \e\bigl[|\e_{0}[Y(z) \mid X]|^q\bigr] 
\leq \e_0\bigl[|Y(z)|^q\bigr] \leq C^q
\]
and then by the definition of $\theta_0$
\begin{align*}
|\theta_0| &= \bigl(|\e_0[m_0(1,X)-m_0(0,X)]|^q\bigr)^{1/q} \\
&\leq \bigl(\e_0\bigl[|m_0(1,X)-m_0(0,X)|^q\bigr]\bigr)^{1/q} \\
&\leq 2C.
\end{align*}
With $\|m(z,\cdot)-m_0(z,\cdot)\|_{q,P_0^X} \leq C$ by definition of $\ct_N$,
we have 
$$\|m(z,\cdot)\|_{q,P_0^X} \leq \|m(z,\cdot)-m_0(z,\cdot)\|_{q,P_0^X} + \|m_0(z,\cdot)\|_{q,P_0^X} \leq 2C$$ 
for $z=0,1$.
Therefore
\begin{align*}
\sup_{(\nu,e) \in \ct_N}\bigl(\e_0\bigl[|s(W(1);\theta_0,\nu,e)|^q\bigr]\bigr)^{1/q} & \leq \sup_{(\nu,e) \in \ct_N} \|m(1,\cdot)\|_{q,P_0^X} + \|m(0,\cdot)\|_{q,P_0^X} + |\theta_0| \\
&\phe\ + \sup_{(\nu,e) \in \ct_N} \gamma^{-1}\bigl((\e_0[|Y(1)|^q])^{1/q}+\|m(1,\cdot)\|_{q,P_0^X}\bigr) \\
& \leq 6C + \frac{3C}{\gamma}.
\end{align*}
Similarly,
$\sup_{(\nu,e) \in \ct_N}\bigl(\e_0\bigl[|s(W(0);\theta_0,\nu,e)|^q\bigr]\bigr)^{1/q} \leq 6C+{3C}/{\gamma}$ 
as well.
\item[Equation~\eqref{eq:S_t_neg_k}] 
We begin by considering $S_t^{(-k)}(z)$ for $z=1$.
For each $k=1,\ldots,K$ we have
\begin{align*}
s(W_{ti}(1);\hat{\nu}^{(-k)},\hat{e}^{(-k)})-s(W_{ti}(1);\nu_0,e_{0,N}) & = \bigl(1-\hat{e}^{(-k)}(X)^{-1}\bigr)\bigl(\hat{m}^{(-k)}(1,X_{ti})-m_0(1,X_{ti})\bigr) \\
&\phe\ +\bigl(Y_{ti}(1)-m_0(1,X_{ti})\bigr)\bigl(\hat{e}^{(-k)}(X_{ti})^{-1}-e_0(X_{ti})^{-1}\bigr) \\
&\phe\ -\bigl(\hat{m}^{(-k)}(0,X_{ti})-m_0(0,X_{ti})\bigr).
\end{align*}
Taking the conditional expectation given $\cs^{(-k)}$ and $X_{ti}$ yields
\begin{align*}
&\phe\ \e\bigl[s(W_{ti}(1);\hat{\nu}^{(-k)},\hat{e}^{(-k)})-s(W_{ti}(1);\nu_0,e_{0,N}) \mid \cs^{(-k)},X_{ti}\bigr] \\
 &= \bigl(1-\hat{e}^{(-k)}(X)^{-1}\bigr)(\hat{m}^{(-k)}(1,X_{ti})-m_0(1,X_{ti})) - (\hat{m}^{(-k)}(0,X_{ti})-m_0(0,X_{ti})).
\end{align*}
Then for $z=1$ we get
\begin{align*}
S_t^{(-k)}(1) & \leq \sqrt{\frac{1}{n_{t,k}}\sum_{i:(t,i) \in \ci_k} \bigl(1-\hat{e}^{(-k)}(X_{ti})^{-1}\bigr)^2\bigl(\hat{m}^{(-k)}(1,X_{ti})-m_0(1,X_{ti})\bigr)^2} \\
&\phe\ + \sqrt{\frac{1}{n_{t,k}}\sum_{i:(t,i) \in \ci_k}(\hat{m}^{(-k)}(0,X_{ti})-m_0(0,X_{ti}))^2} \\
& = O_p\bigl(\|\hat{m}^{(-k)}(1,\cdot)-m(1,\cdot)\|_{2,P_0^X} + \|\hat{m}^{(-k)}(0,\cdot)-m(0,\cdot)\|_{2,P_0^X}\bigr) \\
& = o_p(N^{-1/4}).
\end{align*}
The first equality above follows from
$\left(1-\hat{e}^{(-k)}(X_{ti})\right)^2 \leq (1+\gamma^{-1})^2 < \infty$,
followed by 
an application of the conditional Markov inequality using
\begin{align*}
\e\biggl[\frac{1}{n_{t,k}}\sum_{i:(t,i) \in \ci_k}(\hat{m}^{(-k)}(z,X_{ti})-m_0(z,X_{ti}))^2 \Bigm| \cs^{(-k)} \biggr] & = \|\hat{m}^{(-k)}(z,\cdot)-m_0(z,\cdot)\|_{2,P_t^X}^2 \\
& \leq \kappa_t^{-1}  \|\hat{m}^{(-k)}(z,\cdot)-m_0(z,\cdot)\|_{2,P_0^X}^2
\end{align*}
for $z=0,1$.
By an identical argument, $S_t^{(-k)}(0)=o_p(N^{-1/4})$,
establishing~\eqref{eq:S_t_neg_k}.
\end{description}
Having shown all conditions of Assumption~\ref{assump:dml},
the conclusion of Corollary~\ref{cor:AIPW} follows by Theorem~\ref{thm:batch_clt}.
\subsection{Proof of Corollary~\ref{cor:pl}:
CLT for feasible $\hat{\theta}_{\epl}$ in a CSBAE}
\label{proof:cor:pl}

Corollary~\ref{cor:pl},
which shows the feasible estimator $\hat{\theta}_{\epl}$ satisfies a CLT for estimating $\theta_{0,\pl}$ in a CSBAE,
holds under the numbered generalizations at the end of Appendix~\ref{app:nonstationary},
subject to the additional requirement that the mean and variance functions are stationary across batches,
as in Appendix~\ref{proof:cor:pl_oracle_clt}.
We prove this more general result using the same structure as the proof of Corollary~\ref{cor:AIPW} in Appendix~\ref{proof:cor:AIPW}.

The proof of Corollary~\ref{cor:pl_oracle_clt} in Appendix~\ref{proof:cor:pl_oracle_clt} shows that Assumption~\ref{assump:gen_identification} is satisfied with
$\theta_0=\theta_{0,\pl}$,
$s(\cdot)=s(\cdot)_{\epl}$,
$\nu_0=\nu_{0,\epl} \in \cn=\cn_{\epl}$,
and $\gamma=0$.
Then it suffices to show the remaining conditions of Assumption~\ref{assump:dml}
to complete the proof,
in view of Theorem~\ref{thm:batch_clt}.
Throughout the remainder of this proof we let $C_0$ be a generic positive finite constant;
possibly depending on $c$ and $C$ in Assumption~\ref{assump:pl_regularity};
different appearances of $C_0$ may correspond to different constants.

\begin{description}
\item[Condition~\ref{cond:regularity}] Invertibility of $\e_0[s_a(W;\nu_0,e_0)]$ and existence of $\e_0[\|s(W;\theta_0,\nu_0,e_0)\|^2]$ follow from Assumption~\ref{assump:pl_regularity},
as shown in the proof of Corollary~\ref{cor:pl_oracle_clt}.
\item[Condition~\ref{cond:score_differentiability}] Fix $(\nu,e) \in \ct$.
Recall that the weight 
function in our estimator $\hat{\theta}_{\epl}$
is $w(X,\nu,e)= (v(0,x)e(x)+v(1,x)(1-e(x)))^{-1}$.
We compute
\begin{align*}
\e_{0,N}&[s(W;\theta_0,\nu_0+\lambda(\nu-\nu_0),e_{0,N}+\lambda(e-e_{0,N}))] \\
& = \e_{0,N}[w(X;\nu_{\lambda},e_{\lambda,N})(Z-e_{\lambda,N}(X))(Y-m_{\lambda}(0,X)-Z\psi(X)^{\top}\theta_0)\psi(X)] \\
& = \lambda \e_{0,N}[w(X;\nu_{\lambda},e_{\lambda,N})(Z-e_{\lambda,N}(X))(m_0(0,X)-m(0,X))\psi(X)] \\
& = \e_{0,N}[f_N(\lambda)]
\end{align*}
where the last two equalities follow by conditioning on $(X,Z)$ then just $X$,
and we have defined
\begin{align*}
\nu_{\lambda}(\cdot) & =\nu_0(\cdot)+\lambda(\nu(\cdot)-\nu_0(\cdot)), \\
e_{\lambda,N}(\cdot) & = e_{0,N}(\cdot)+\lambda(e(\cdot)-e_{0,N}(\cdot))\quad\text{and} \\
f_N(\lambda) & = \lambda^2w(X;\nu_{\lambda},e_{\lambda,N})(e_{0,N}(X)-e(X))(m_0(0,X)-m(0,X))\psi(X)\\
&=\lambda^2w(X;\nu_\lambda,e_{\lambda,N})g_N(X)
\end{align*}
for $g_N(X)=(e_{0,N}(X)-e(X))(m_0(0,X)-m(0,X))\psi(X)$.
We compute
\begin{align*}
\frac{\partial}{\partial \lambda} w(X;\nu_{\lambda},e_{\lambda,N}) & = -w^2(X;\nu_{\lambda},e_{\lambda,N})\Delta_N(\lambda,X)  \\ 
\frac{\partial}{\partial \lambda} w^2(X;\nu_{\lambda},e_{\lambda,N}) & = -2w^3(X;\nu_{\lambda},e_{\lambda,N})\Delta_N(\lambda,X)
\end{align*} 
where
\begin{align*}
\Delta_N(\lambda,X) & = \frac{\partial}{\partial \lambda} \bigl(v_{\lambda}(0,X)e_{\lambda,N}(X)+v_{\lambda}(1,X)(1-e_{\lambda,N}(X))\bigr) \\
& = (1-e_{\lambda}(X))(v(1,X)-v_0(1,X))+e_{\lambda}(X)(v(0,X)-v_0(0,X)) \quad\text{and} \\
&\phe\ +(e(X)-e_{0,N}(X))(v_{\lambda}(0,X)-v_{\lambda}(1,X))
\end{align*}
Then
\begin{align*}
f_N'(\lambda) & =  (-\lambda^2\Delta_N(\lambda,X)w^2(X;\nu_{\lambda},e_{\lambda,N}) + 2\lambda w(X;\nu_{\lambda},e_{\lambda,N}))g_N(X) \\
f_N''(\lambda) & = (2\lambda^2(\Delta_N(\lambda,X))^2w^3(X;\nu_{\lambda},e_{\lambda,N})-\lambda^2\Delta_N^{(2)}(\lambda,X)w^2(X;\nu_{\lambda},e_{\lambda,N}))g_N(X) \\
& +(2w(X;\nu_{\lambda},e_{\lambda,N})-4\lambda\Delta_N(\lambda,X)w^2(X;\nu_{\lambda},e_{\lambda,N}))g_N(X)
\end{align*}
where
\begin{align*}
\Delta_N^{(2)}(\lambda,X) & = \frac{\partial}{\partial \lambda} \Delta_N(\lambda,X) \\
& = 2(e(X)-e_{0,N}(X))(v(0,X)-v_0(0,X)+v_0(1,X)-v(1,X))
\end{align*}
We conclude that for sufficiently small $\epsilon>0$,
\begin{align*}
\sup_{\lambda \in (-\epsilon,1+\epsilon)}   |\Delta_N(\lambda,X)| & \leq \sup_{\lambda \in (-\epsilon,1+\epsilon)} |1-e_{\lambda}(X))||v(1,X)-v_0(1,X)|+|e_{\lambda}(X)||v(0,X)-v_0(0,X)| \\
&\phe\ + \sup_{\lambda \in (-\epsilon,1+\epsilon)} |e(X)-e_{0,N}(X)||v_{\lambda}(0,X)-v_{\lambda}(1,X)| \\
& \leq C_0 \quad\text{and}\\
\sup_{\lambda \in (-\epsilon,1+\epsilon)} |\Delta_N^{(2)}(\lambda,X)| & \leq C_0.
\end{align*}
Taking $\epsilon$ smaller if necessary,
we can ensure
\[
\sup_{\lambda \in (-\epsilon,1+\epsilon)} w(X;\nu_{\lambda},e_{\lambda,N}) = \sup_{\lambda \in (-\epsilon,1+\epsilon)} 
 (v_{\lambda}(0,X)e_{\lambda,N}(X)+v_{\lambda}(1,X)(1-e_{\lambda,N}(X)))^{-1} < 2c^{-1}
\]
and then
\begin{align*}
\sup_{\lambda \in (-\epsilon,1+\epsilon)} \|f_N(\lambda)\| & \leq  C_0\|g_n(X)\| \leq C_0C|e_{0,N}(X)-e(X)||m_0(0,X)-m(0,X)| \\
\sup_{\lambda \in (-\epsilon,1+\epsilon)} \|f_N'(\lambda)\| & \leq  C_0\|g_n(X)\| \leq C_0C|e_{0,N}(X)-e(X)||m_0(0,X)-m(0,X)| \\
\sup_{\lambda \in (-\epsilon,1+\epsilon) } \|f_N''(\lambda)\| & \leq  K\|g_n(X)\| \leq C_0C|e_{0,N}(X)-e(X)||m_0(0,X)-m(0,X)|
\end{align*} 
By the Leibniz integral rule,
the mapping $\lambda \mapsto \e_{0,N}[s(W;\theta_0,\nu_0+\lambda(\nu-\nu_0),e_{0,N}+\lambda(e-e_{0,N}))]$
then has second derivative $\e_{0,N}[f_N''(\lambda)]$ on $[0,1]$.
With $f_N''(\lambda)$ continuous on $[0,1]$,
we conclude by dominated convergence that 
$\e_{0,N}[f''_N(\lambda)]$
is continuous as well.
\item[Condition~\ref{cond:potential_score_equations}]
For $z=0,1$,
\[
\e_0[Y(z) \giv X]=\e_0[Y \giv X,Z=z]=m_0(0,X)+z\psi(X)^{\top}\theta_0. \]
Fix $e(\cdot) \in \cf_0$. 
For each $t=1,\ldots,T$, $i=1,\ldots,N_t$ we have
\begin{align*}
\e&[s(W_{ti}(1);\theta_0,\nu_0,e) \mid X_{ti}] \\
& = w(X_{ti};\nu_0,e)(1-e(X))\bigl(\e_0[Y(1) \giv X=X_{ti}]-m_0(0,X_{ti})-\psi(X_{ti})^{\top}\theta_0\bigr)\psi(X_{ti}) \\
& = 0 
\intertext{and similarly}
\e&[s(W_{ti}(0);\theta_0,\nu_0,e) \mid X_{ti}] \\
& = -w(X_{ti};\nu_0,e)e(X)\bigl(\e_0[Y(0) \giv X=X_{ti}]-m_0(0,X_{ti})\bigr)\psi(X_{ti}) \\
& = 0.
\end{align*}
These are equal so their difference is 0
as required for~\eqref{eq:potential_score_diff}.
\end{description}

Now we show the out-of-fold estimates $(\hat{m}^{(-k)}(0,\cdot), \hat{v}^{(-k)}(0,\cdot),\hat{v}^{(-k)}(1,\cdot),\hat{e}^{(-k)}(\cdot))$ lie in a set $\ct_N$ with high probability for all sufficiently large $N$,
where this $\ct_N$ satisfies equations~\eqref{eq:neyman_orthogonality} through~\eqref{eq:s_moment_condition}.
By the rate and regularity conditions on the nuisance estimators in Corollary~\ref{cor:pl},
there exists a sequence $\tilde{\delta}_N \downarrow 0$ and constants $0<c<C<\infty$ so that with probability approaching 1 as $N \rightarrow \infty$, we have
\begin{align*}
\|\hat{m}^{(-k)}(0,\cdot) - m_0(0,\cdot)\|_{2,P_0^X} + \|\hat{v}(z,\cdot)-v_0(z,\cdot)\|_{2,P_0^X} + \|\hat{e}^{(-k)}-e_{0,N}\|_{2,P_0^X} & \leq \tilde{\delta}_N \\
\|\hat{m}^{(-k))}(0,\cdot)-m_0(0,\cdot\|_{2,P_0^X}\|\hat{e}^{(-k)}-e_{0,N}\|_{2,P_0^X} & \leq N^{-1/2}\tilde{\delta}_N \\
\|\hat{m}^{(-k)}(0,\cdot)-m_0(0,\cdot)\|_{q,p_0^X} & \leq C \\
\hat{v}^{(-k)}(z,x) & \geq c, \quad z=0,1 
\end{align*}
for all folds $k=1,\ldots,K$.
Then let $\ct_N$ be the set of functions $(m(0,\cdot),v(0,\cdot),v(1,\cdot),e(\cdot))$ in $\cn_{\epl} \times \cf_0$ for which 
\begin{align*}
\|m(0,\cdot) - m_0(0,\cdot)\|_{2,P_0^X} + \|v(z,\cdot)-v_0(z,\cdot)\|_{2,P_0^X} + \|e-e_{0,N}\|_{2,P_0^X} & \leq \tilde{\delta}_N \\
\|m(0,\cdot)-m_0(0,\cdot\|_{2,P_0^X}\|e-e_{0,N}\|_{2,P_0^X} & \leq N^{-1/2}\tilde{\delta}_N \\
\|m(0,\cdot)-m_0(0,\cdot)\|_{q,p_0^X} & \leq C \\
v(z,x) & \geq c, \quad z=0,1. 
\end{align*}
By construction,
$\Pr((\hat{m}^{(-k)}(0,\cdot),\hat{v}^{(-k)}(0,\cdot),\hat{v}^{(-k)}(1,\cdot),\hat{e}^{(-k)}(\cdot)) \in \mathcal{T}_N) \rightarrow 1$ as $N \rightarrow \infty$ for all $k=1,\dots,K$.
Now we show that $\ct_N$ satisfies equations~\eqref{eq:neyman_orthogonality} through~\eqref{eq:s_moment_condition}
for all $N$ large enough to ensure that $1/2 \leq {\mrd P_{0,N}^X}/{\mrd P_0^X} \leq 2$.
For such $N$,  $\|f\|_{2,P_{0,N}^X} \leq \sqrt{2}\|f\|_{2,P_0^X}$ 
holds for all $f \in L^2(P_0^X)$.
As in the proof of Corollary~\ref{cor:AIPW},
we will show equations~\eqref{eq:neyman_orthogonality} through~\eqref{eq:s_consistency} hold
for a sequence $\delta_N$ that is some constant multiple of $\tilde{\delta}_N$.
\begin{description}
\item[Equation~\eqref{eq:neyman_orthogonality}]
Using the notation from our proof of condition~\ref{cond:score_differentiability} above, we see $f_N'(0)=0$ with probability 1, so for any $(\nu,e) \in \ct_N$ we have
\[
\frac{\partial}{\partial \lambda} \e_{0,N}[s(W;\theta_0,\nu_0+\lambda(\nu-\nu_0),e_{0,N}+\lambda(e-e_{0,N}))]\Big|_{\lambda=0} = \e_{0,N}[f_N'(0)]=0
\]
which shows the left-hand side of~\eqref{eq:neyman_orthogonality} is identically zero.
\item[Equation~\eqref{eq:vanishing_second_derivatives}] Once again we recall the notation and calculations from the proof of condition~\ref{cond:score_differentiability} above.
For each $\lambda \in [0,1]$ and $(\nu,e) \in \ct_N$ we see
\begin{align*}
\bigg\|\frac{\partial^2}{\partial \lambda^2} \e_{0,N}\bigl[s& \bigl(W;\theta_0,\nu_0+\lambda(\nu-\nu_0),e_{0,N}+\lambda(e-e_{0,N})\bigr)\bigr]\bigg\| = \bigl\|\e_{0,N}[f_N''(\lambda)]\bigr\|
\end{align*}
which is no larger than
\begin{align*}
\e_{0,N}\biggl[\,\sup_{\lambda \in [0,1]} \|f_N''(\lambda)\|\biggr]
\leq C_0\|e-e_{0,N}\|_{2,P_0^X}\|m(0,\cdot)-m_0(0,\cdot)\|_{2,P_0^X} 
\leq C_0N^{-1/2}\tilde{\delta}_N
\end{align*}
by the definition of $\ct_N$.
\item[Equation~\eqref{eq:s_a_consistency}]
Recall that $s_{\epl,a}(W;\nu,e)=-w(X;\nu,e(X))Z(Z-e)\psi(X)\psi(X)^\top$.
Fix $(\nu,e) \in \ct_N$.
We compute
\begin{align*}
s_a(W;\nu,e)-s_a(W;\nu_0,e_0) & = \bigl(w(X;\nu_0,e_0)(Z-e_0(X))-w(X;\nu,e)(Z-e(X))\bigr)Z\psi(X)\psi(X)^{\top} \\
& = w(X;\nu_0,e_0)(e(X)-e_0(X))Z\psi(X)\psi(X)^{\top} \\
&\phe\ + (w(X;\nu_0,e_0)-w(X;\nu,e))(Z-e(X))Z\psi(X)\psi(X)^{\top}.
\end{align*}
Hence, using $\|\psi(X)\psi(X)^{\top}\| \leq C^2$,
\begin{align*}
&\phe\ \e_0\bigl[\|s_a(W;\nu,e)-s_a(W;\nu_0,e_0)\|^2\bigr]^{1/2} \\
& \leq C^2\e_0\bigl[(w(X;\nu_0,e_0)Z(e(X)-e_0(X)))^2\bigr]^{1/2} \\
&\phe\ 
+ C^2\e_0\bigl[(w(X;\nu_0,e_0)-w(X;\nu,e))^2(Z-e(X))^2Z^2\bigr]^{1/2}.
\end{align*}.
To bound the right-hand side,
first note that
\[
\Big(\e_0\bigl[(w(X;\nu_0,e_0)Z(e(X)-e_0(X)))^2\bigr]\Big)^{1/2}  
\leq c^{-1}\|e-e_0\|_{2,P_0^X} \leq c^{-1}\tilde{\delta}_N.
\]
Next
\begin{align*}
&\phe\ \, |w(X;\nu,e)-w(X;\nu_0,e_0)| \\
& \leq c^{-2}\bigl|v(0,X)e(X)+v(1,X)(1-e(X))-v_0(0,X)e_0(X)-v_0(1,X)(1-e_0(X))\bigr| \\
& \leq c^{-2}\bigl(|v(0,X)-v_0(0,X)| \cdot e(X) + v_0(0,X) \cdot |e(X)-e_0(X)|\bigr) \\
&\phe\ + c^{-2}\bigl(|v(1,X)-v_0(1,X)| \cdot (1-e(X)) + v_0(1,X) \cdot |e(X)-e_0(X)|\bigr) \\
& \leq c^{-2}\bigl(|v(0,X)-v_0(0,X)| + |v(1,X)-v_0(1,X)| + 2C|e(X)-e_0(X)|\bigr)
\end{align*}
and so
\begin{align*}
&\phe\ \,\e_0 \bigl[(w(X;\nu_0,e_0)-w(X;\nu,e))^2(Z-e(X))^2Z^2\bigr]^{1/2} \\
& \leq c^{-2}\bigl(\|v(0,\cdot)-v_0(0,\cdot)\|_{2,P_0^X} + \|v(1,\cdot)-v_0(1,\cdot)\|_{2,P_0^X} + 2C\|e-e_0\|_{2,P_0^X}\bigr) \\
& \leq C_0\tilde{\delta}_N,
\end{align*}
which shows~\eqref{eq:s_a_consistency}.
\item[Equation~\eqref{eq:s_consistency}]
Once again, fix $(\nu,e) \in \ct_N$. Then
\begin{align*}
&\phe\ s(W;\theta_0,\nu,e)-s(W;\theta_0,\nu_0,e_0) \\
& = [w(X;\nu,e)(Z-e(X))-w(X;\nu_0,e_0)(Z-e_0(X))][Y-m_0(0,X)-Z\psi(X)^{\top}\theta_0]\psi(X) \\
&\phe\ + w(X;\nu,e)(Z-e(X))(m_0(0,X)-m(0,X))\psi(X)
\end{align*}
so that
\begin{align*}
&\phe\ \,\e_0\bigl[\|s(W;\theta_0,\nu,e)-s(W;\theta_0,\nu_0,e_0)\|^2\bigr]^{1/2} \\
& \leq C\Bigl(\e_0\bigl[\bigl(w(X;\nu,e)(Z-e(X))-w(X;\nu_0,e_0)(Z-e_0(X))\bigr)^2\bigl(Y-m_0(0,X)-Z\psi(X)^{\top}\theta_0\bigr)^2\bigr]\Bigr)^{1/2} \\
&\phe\ + C\bigl(\e_0\bigl[(w(X;\nu,e)(Z-e(X)))^2(m_0(0,X)-m(0,X))^2\bigr]\bigr)^{1/2}.
\end{align*}
By conditioning on $(X,Z)$ we see
that
\begin{align*}
\e_0&\bigl[|w(X;\nu,e)(Z-e(X))-w(X;\nu_0,e_0)(Z-e_0(X))|^2|Y-m_0(0,X)-Z\psi(X)^{\top}\theta_0|^2\bigr]^{1/2} \\
& = \e_0\bigl[v_0(Z,X)|w(X;\nu,e)(Z-e(X))-w(X;\nu_0,e_0)(Z-e_0(X))|^2\bigr]^{1/2} \\
& \leq C^{1/2}\e_0\bigl[|w(X;\nu,e)(Z-e(X))-w(X;\nu_0,e_0)(Z-e_0(X))|^2\bigr]^{1/2}.
\end{align*}
Now we write
\begin{align*}
\bigl|&w(X;\nu,e)(Z-e(X)) -w(X;\nu_0,e_0)(Z-e_0(X))\bigr| 
 \\
 & \leq |w(X;\nu,e)-w(X;\nu_0,e_0)||Z-e(X)| + |w(X;\nu_0,e_0)||e_0(X)-e(X)| \\
 & \leq c^{-2}(|v(0,X)-v_0(0,X)| + |v(1,X)-v_0(1,X)| + 2C|e(X)-e_0(X)|) + c^{-1}|e_0(X)-e(X)|
\end{align*}
so that
\begin{align*}
\e_0&\bigl[|w(X;\nu,e)(Z-e(X))-w(X;\nu_0,e_0)(Z-e_0(X))|^2|Y-m_0(0,X)-Z\psi(X)^{\top}\theta_0|^2\bigr]^{1/2} \\
& \leq C^{1/2}\bigl[c^{-2}(\|v(0,\cdot)-v_0(0,\cdot)\|_{2,P_0^X} + \|v(1,\cdot)-v_0(1,\cdot)\|_{2,P_0^X}+2C\|e-e_0\|_{2,P_0^X})+c^{-1}\|e-e_0\|_{2,P_0^X}\bigr] \\
& \leq C_0\tilde{\delta}_N
\end{align*}
With $(\nu,e) \in \ct_N$ arbitrary and
\[
\e_0\bigl[|w(X;\nu,e)(Z-e(X))|^2|m_0(0,X)-m(0,X)|^2\bigr]^{1/2} \leq c^{-1}\|m(0,\cdot)-m_0(0,\cdot)\|_{2,P_0^X} \leq C_0\tilde{\delta}_N
\]
we have shown~\eqref{eq:s_consistency}.

\item[Equation~\eqref{eq:s_a_moment_condition}] Fix $(\nu,e) \in \ct_N$.
For $z=0,1$ we have
\[
\sup_{z \in \{0,1\}} \|s_a(W(z);\nu,e)\|^q = \sup_{z \in \{0,1\}} \|w(X;\nu,e)z(z-e(X))\psi(X)\psi(X)^{\top}\|^q \leq \left(\frac{C^2}{c}\right)^q
\]
which immediately shows that
\[ 
\e_0\bigl[\|s_a(W(z);\nu,e)\|^q\bigr]^{1/q} \leq C_0.
\]

\item[Equation~\eqref{eq:s_moment_condition}] For any $(\nu,e) \in \ct_N$ we have
\[
\e_0\bigl[\|s(W(z);\theta_0,\nu,e)\|^q\bigr]^{1/q} \leq \frac{C}{c}\e\bigl[|Y(z)-m(0,X)-z\psi(X)^{\top}\theta_0|^q\bigr]^{1/q} \leq C_0
\]
since
\begin{align*}
\e_0&\bigl[|Y(z)-m(0,X)-z\psi(X)^{\top}\theta_0|^q\bigr]^{1/q} \\
& \leq \e_0\bigl[|Y(z)|^q\bigr]^{1/q} + \e_0\bigl[|m(0,X)|^q\bigr]^{1/q} + \e_0\bigl[|m(1,X)|^q\bigr]^{1/q}.
\end{align*}
We have $(\e_0[|Y(z)|^q])^{1/q} \leq C_0$ by Assumption~\ref{assump:pl_regularity},
hence 
for $z=0,1$
\begin{align*}
(\e_0[|m(z,X)|^q])^{1/q} & \leq (\e_0[|m_0(z,X)|^q])^{1/q} + \|m(z,\cdot)-m_0(z,\cdot)\|_{q,P_0^X} \\
& \leq (\e_0[|Y(z)|^q])^{1/q} + \|m(z,\cdot)-m_0(z,\cdot)\|_{q,P_0^X} \\
& \leq C_0
\end{align*}
which shows~\eqref{eq:s_moment_condition}.
\item[Equation~\eqref{eq:S_t_neg_k}]
Fix a fold $k \in \{1,\ldots,K\}$.
For each $z=0,1$ and $(t,i) \in \ci_k$ we have
\begin{align*}
s&(W_{ti}(z);\hat{\nu}^{(-k)},\hat{e}^{(-k)})-s(W_{ti}(z);\nu_0,e_{0,N}) \\
& = \bigl[w(X_{ti};\hat{\nu}^{(-k)},\hat{e}^{(-k)})(z-\hat{e}^{(-k)}(X_{ti}))-w(X_{ti};\nu_0,e_{0,N})(z-e_{0,N}(X_{ti}))\bigr] \\
& \phe\times [Y_{ti}(z)-m_0(0,X_{ti})-z\psi(X_{ti})^{\top}\theta_0]\psi(X_{ti}) \\
& \phe+ w(X_{ti};\hat{\nu}^{(-k)},\hat{e}^{(-k)})(z-\hat{e}^{(-k)}(X_{ti}))(m_0(0,X_{ti})-\hat{m}^{(-k)}(0,X_{ti}))\psi(X_{ti}).
\end{align*}
Taking the conditional expectation given $\cs^{(-k)},X_{ti}$ gives
\begin{align*}
\Big\|\e&\bigl[s(W_{ti}(z);\hat{\nu}^{(-k)},\hat{e}^{(-k)})-s(W_{ti}(z);\nu_0,e_{0,N}) \giv \cs^{(-k)},X_{ti}\bigr]\Big\|^2 \\
& = w^2(X_{ti};\hat{\nu}^{(-k)},\hat{e}^{(-k)})(z-\hat{e}^{(-k)}(X_{ti}))^2(\hat{m}^{(-k)}(0,X_{ti})-m_0(0,X_{ti}))^2\|\psi(X_{ti})\|^2 \\
& \leq \frac{C^2}{c^2}(\hat{m}^{(-k)}(0,X_{ti})-m_0(0,X_{ti}))^2
\end{align*}
and so
\[
S_t^{(-k)}(z) \leq C_0 \sqrt{\frac{1}{n_{t,k}} \sum_{(t,i) \in \ci_k} (\hat{m}^{(-k)}(0,X_{ti})-m_0(0,X_{ti}))^2} = o_p(N^{-1/4})
\]
in view of conditional Markov's inequality, as
\begin{align*}
\e\biggl[\frac{1}{n_{t,k}} \sum_{(t,i) \in \ci_k} (\hat{m}^{(-k)}(0,X_{ti})-m_0(0,X_{ti}))^2 \Bigm| \cs^{(-k)}\biggr] & = \|\hat{m}^{(-k)}(0,\cdot)-m_0(0,\cdot)\|_{2,P_t^X} \\
& \leq \kappa_t^{-1}\|\hat{m}^{(-k)}(0,\cdot)-m_0(0,\cdot)\|_{2,P_0^X} \\
& = o_p(N^{-1/4}) 
\end{align*}
by assumption.
\end{description}
Having shown all conditions of Assumption~\ref{assump:dml}, we can apply Theorem~\ref{thm:batch_clt} to complete the proof of Corollary~\ref{cor:pl}.

\subsection{Proof of Lemma~\ref{lemma:Psi_cond}}\label{proof:lemma:Psi_cond}
Here we show that the information functions $\Psi_d(\cdot)$ and $\Psi_a(\cdot)$ for
D-optimality and A-optimality, respectively, both satisfy conditions (a) through (d)
of Assumption~\ref{assump:Psi}.
That both $\Psi_d(\cdot)$ and $\Psi_a(\cdot)$ are continuous, concave, and non-decreasing on $\ess_+^p$ is well known.

If $M$ is singular then $\Psi_d(M)=\Psi_a(M)=-\infty$;
however if $M \in \ess_{++}^p$ then $\Psi_d(M)$ and $\Psi_a(M)$ are finite.
Thus condition (a) holds with $\Psi_0 = -\infty$.

Next, we recall that for $M \in \ess_{++}^p$, 
we have $\nabla \Psi_d(M) = M^{-1}$ and $\nabla \Psi_a(M) = M^{-2}$.
Now we fix $0<k<K$ and $A, B \in \ess_{++}^p$ such that
$KI \succeq A \succeq kI$ and $KI \succeq B \succeq kI$.
Then
\begin{align*}
\|\nabla \Psi_d(A) - \nabla \Psi_d(B)\| & = \|A^{-1} - B^{-1} \| = \|A^{-1}(B-A)B^{-1}\| \\
& \leq \|A^{-1}\|\|B^{-1}\|\|B-A\| \leq k^{-2} \|A-B\|\quad\text{and} \\
\|\nabla \Psi_a(A) - \nabla \Psi_a(B)\| & = \|A^{-2}-B^{-2} \| = \|A^{-2}(B^2-A^2)B^{-2} \| \\
& \leq \|A^{-2}\|\|B^{-2}\|(\|B\|\|B-A\| + \|B-A\|\|A\|) \\
& \leq Kk^{-4} \|A-B\|
\end{align*}
which shows condition (b).

We also have $K^{-1} I \preceq \nabla \Psi_d(A) \preceq k^{-1} I$ and $K^{-2} I \preceq \nabla \Psi_a(A) \preceq k^{-2} I$.
Therefore condition (c) holds.

Finally fix $\tilde{\Psi}_0 > -\infty$,
and suppose $0 \preceq A \preceq KI$ with $\Psi_d(A) \geq \tilde{\Psi}_0$.
Letting $\lambda_1 \geq \ldots \geq \lambda_p$ be the eigenvalues of $A$,
we have $\lambda_j \leq K$ for all $j$ and so
\[
(p-1)\log(K) + \lambda_p \geq \Psi_d(A) = \sum_{j=1}^p \log \lambda_j \geq \tilde{\Psi}_0
\]
so that $\lambda_p \geq \exp(\tilde{\Psi}_0-(p-1)\log K) > 0$,
showing condition (d) for $\Psi_d(\cdot)$.
Similarly if $\Psi_a(A) \geq \tilde{\Psi}_0$ then
\[
-\frac{p-1}{K}-\lambda_p^{-1} \geq \Psi_a(A) = -\sum_{j=1}^p \lambda_j^{-1} \geq \tilde{\Psi}_0
\]
which implies $-\tilde{\Psi}_0-(p-1)/K \geq \lambda_p^{-1} > 0$. Therefore
$\lambda_p \geq (-\tilde{\Psi}_0 - {(p-1)}/{K})^{-1} > 0$,
showing condition (d) for $\Psi_a(\cdot)$ as well.

\subsection{Proof of Lemma~\ref{lemma:concave_maximization}: Convergence of generic concave maximization routine}
\label{proof:lemma:concave_maximization}

Here we prove Lemma~\ref{lemma:concave_maximization} about the convergence rates of our generic concave maximization routine. 

\begin{proof}
Let $\tilde{\delta} = \delta/2$,
with $\delta < 0$ as in Assumption~\ref{assump:info_matrix}.
We will repeatedly use the fact that since $[\tilde{\delta},1-\tilde{\delta}] \times \cw \subseteq \real^{r+1}$ is compact, 
Assumption~\ref{assump:info_matrix} implies that $f=f(\cdot,\cdot)$ and all of its partial derivatives up to second order are uniformly bounded above in norm 
on that set.
Then WLOG we can make $C$ from Assumption~\ref{assump:info_matrix} larger so that
\begin{equation}
\label{eq:uniform_bound}
    \sup_{(k,w) \in [\tilde{\delta},1-\tilde{\delta}] \times \cw} \|h(k,w)\| \leq C, \quad \forall h=h(\cdot,\cdot) \in \{f,f',f'',f_w,f_{ww},f_w'\}
\end{equation}
where $f_w=f_w(\cdot,\cdot)$, $f_{ww}=f_{ww}(\cdot,\cdot)$, and $f_w'=f_w'(\cdot,\cdot)$ are tensors with $f_w$ the partial derivative of $f$ with respect to the second argument,
$f_{ww}$ the second partial derivative of $f$ with respect to the second argument,
and $f_w'$ the partial derivative of $f'$ with respect to the second argument.

First, we show the existence of $e^*(\cdot)$ satisfying~\eqref{eq:e_star}.
For each propensity $e=e(\cdot) \in \ce$ define
$\phi(e) = \Psi(M(e))$ and $\hat{\phi}_n(e) = \Psi(\hat{M}_n(e))$ where
\[
M(e) = \int_{\cx} f(e(x),\eta(x)) \mrd P(x) \quad\text{and}\quad \hat{M}_n(e) = \frac{1}{n} \sum_{i=1}^n f(e(X_i),\hat{\eta}(X_i)).
\]
For any $e_1=e_1(\cdot),e_2=e_2(\cdot) \in \ce$ and $\lambda \in [0,1]$ we have
\begin{align*}
M(\lambda e_1 + (1-\lambda)e_2) & = \int_{\cx} f\bigl(\lambda e_1(x) + (1-\lambda)e_2(x),\eta(x)\bigr)\, \mrd P(x) \\
& \succeq \int_{\cx} \bigl( \lambda f(e_1(x),\eta(x)) + (1-\lambda) f(e_2(x),\eta(x)) \bigr) \, \mrd P(x) \\
& = \lambda M(e_1) + (1-\lambda)M(e_2)
\end{align*}
where the matrix inequality follows from the fact that the function $u \mapsto f(u,\eta(x))$ is a concave matrix-valued function on $[0,1]$ since its second derivative is globally negative semidefinite by~\eqref{eq:f_strong_concavity}.
Thus $M=M(\cdot)$ is also a concave matrix-valued function on $\ce$.
In fact, $M$ is also Lipschitz continuous in the sense that 
\begin{align}
\|M(e_1)-M(e_2)\| & \leq \int_{\cx} \|f(e_1(x),\eta(x))-f(e_2(x),\eta(x))\|\, \mrd P(x)\nonumber \\
& \leq C \int_{\cx} |e_1(x)-e_2(x)|\, \mrd P(x)\nonumber \\
& \leq C \|e_1-e_2\|_{2,P}, \quad \forall e_1,e_2 \in \ce \label{eq:M_lipschitz}
\end{align}
where the second inequality uses~\eqref{eq:uniform_bound} with $h=f'$
and Taylor's theorem with the Lagrange form of the remainder.
With the information function $\Psi=\Psi(\cdot)$ continuous, concave, and increasing in the semidefinite ordering by Assumption~\ref{assump:Psi},
we conclude that $\phi=\phi(\cdot)$ is continuous and concave on $\ce$ (e.g., by Section 3.6 of~\citet{boyd2004convex}).

We now consider an
extension $\bar{\phi}=\bar{\phi}(\cdot)$ of $\phi$ to 
$L^2(P)$,
\[
\bar{\phi}(e) = 
\begin{cases}
\phi(e), & e \in \ce \\
-\infty, & \text{ otherwise.}
\end{cases}
\]
Since $\phi(\cdot)$ is continuous and concave on $\ce$,
it is straightforward to show that the extension $\bar{\phi}$ is concave and upper semicontinuous on $L^2(P)$;
the latter means that
$\bar{\phi}(e_0) = \limsup_{e \rightarrow e_0} \bar{\phi}(e)$ for all $e_0=e_0(\cdot) \in L^2(P)$.
The function $\bar{\phi}$ is also ``proper" in that it 
never equals $+\infty$.
Since $\cf_*$ is a closed, bounded, and convex
subset of the Hilbert space $L^2(P)$, 
by Proposition 1.88 and Theorem 2.11 of~\citet{barbu2012convexity} we conclude $\bar{\phi}(\cdot)$,
and hence $\phi(\cdot)$,
attains its maximum on $\cf_*$.
This shows the existence of $e^* = \argmax_{e \in \cf} \phi(e)$.
We will show uniqueness ($P$-almost surely) later.

Next, note that compactness of $F_n \subseteq \real^n$ 
and continuity of the map $(e_1,\ldots,e_n) \mapsto \Psi\left(n^{-1} \sum_{i=1}^n f(e_i,\hat{\eta}(X_i))\right)$,
which follows from continuity of the information function $\Psi$ 
and of $e \mapsto f(e,\hat{\eta}(x))$ on $[0,1]$ 
for each $x \in \cx$,
ensure the existence of a vector $(\hat{e}_1,\ldots,\hat{e}_n)$ satisfying~\eqref{eq:e_hat}.
This shows the second claim of Lemma~\ref{lemma:concave_maximization}.
Then condition~\ref{cond:transform_En} of Assumption~\ref{assump:info_matrix} ensures that there exists a propensity $\hat{e}=\hat{e}(\cdot) \in \ce$ with $\hat{e}(X_i)=\hat{e}_i$ for all $i=1,\ldots,n$.
So all that remains is to show that any such propensity $\hat{e}(\cdot)$ satisfies the rate conditions $\|\hat{e}-e^*\|_{2,P_n} + \|\hat{e}-e^*\|_{2,P} = O_p(n^{-1/4}+\alpha_n)$,
along with uniqueness of $e^*(\cdot)$ satisfying~\eqref{eq:e_star} ($P$-almost surely).

We have now established existence of $e^*(\cdot)$, $\hat e_1,\dots,\hat e_n$ and $\hat e(\cdot)$.
It remains to show
the desired convergence rates of $\hat e(\cdot)$ to $e^*(\cdot)$ in $L^2$ under both $P$ and $\hat P_n$.
Before proceeding further, we list 
a few useful facts.
By~\eqref{eq:uniform_bound} with $h=f$ 
and recalling that the Frobenius norm of a matrix upper bounds its spectral norm,
we have
\begin{equation}
\label{eq:M_bounded}
M(e) \preceq C I \quad\text{ and }\quad \hat{M}_n(e) \preceq C I, \quad \forall e \in \ce.
\end{equation}
Next, we claim that
\begin{equation}
\label{eq:M_pos_def}
M(e^*) \succeq k^*I
\end{equation}
for some $k^*>0$.
To see this, note that $\Psi(M(e^*)) = \phi(e^*) \geq \phi(e_0) =\Psi(M(e_0)) > \Psi_0$ by Assumption~\ref{assump:Psi}(a) and condition~\ref{cond:positivity} of Assumption~\ref{assump:info_matrix}.
Then Assumption~\ref{assump:Psi}(d) shows that the $k_*>0$ we need exists.

Now, we define the function class
\[
\cf_n=\biggl\{e \in \ce \Bigm| m_L \leq \frac1n\sum_{i=1}^n e(X_i) \leq m_H\biggr\}.
\]
This class contains any propensity $\hat{e}(\cdot)$ derived as above by solving~\eqref{eq:e_hat} and interpolating within the base propensity class $\ce$.

We complete the proof by the following 4 steps:
\begin{enumerate}
    \item Show that $\phi$ is strongly concave on $\cf_*$ with respect to the $\|\cdot\|_{2,P}$ norm and that $\hat{\phi}_n$ is strongly concave on $\cf_n$ with respect to the $\|\cdot\|_{2,P_n}$ norm.
    This means there exist nonrandom positive constants $c_0$, $r_0$, and $k_0$ such that
    \begin{align}
    &c_0 \min\bigl(r_0, \|e^*-e\|_{2,P}\bigr)\times\|e^*-e\|_{2,P} \leq \phi(e^*)-\phi(e), \quad \forall e \in \cf_*     \label{eq:strong_concavity} \\
    \intertext{and}
    &c_0 \min\bigl(r_0, \|\hat{e}-e\|_{2,P_n}\bigr)\times\|\hat{e}-e\|_{2,P_n}\indic(A_n) \leq (\hat{\phi}_n(\hat{e})-\hat{\phi}_n(e))\indic(A_n), \quad \forall e \in \cf_n \label{eq:strong_concavity_n}
    \end{align}
    where $A_n$ is the event that $\hat{M}_n(\hat{e}) \succeq k_0I$. 
    Equation~\eqref{eq:strong_concavity} shows that $e^*(\cdot)$ is unique $P$-a.s.,
    because any $e$ that maximizes~\eqref{eq:e_star} makes the right hand size of~\eqref{eq:strong_concavity} equal $0$ which then makes $\|e-e^*\|_{2,P}=0$.
    \item Conclude by the previous step that with probability approaching 1,
    \begin{align*}
    &c_0\min(r_0,\|e^*-\hat{e}_{\cf}\|_{2,P})\times\|e^*-\hat{e}_{\cf}\|_{2,P}  \leq \phi(e^*)-\phi(\hat{e}_{\cf})+\hat{\phi}_n(\hat{e})-\hat{\phi}_n(e_n^*),\\
\intertext{and}
    &c_0 \min(r_0, \|\hat{e}-e_n^*\|_{2,P_n})\times\|\hat{e}-e_n^*\|_{2,P_n}\indic(A_n)  \leq (\phi(e^*)-\phi(\hat{e}_{\cf})+\hat{\phi}_n(\hat{e})-\hat{\phi}_n(e_n^*))\indic(A_n)
    \end{align*}
    for all $\hat{e}_{\cf} \in \cf$, $e_n^* \in \cf_n$.
    \item Show that with probability approaching 1, there exist $\hat{e}_{\cf} \in \cf$ and $e_n^* \in \cf_n$ converging at the rate $O_p(n^{-1/2})$ in sup-norm to $\hat{e}$ and $e^*$, respectively, so that by empirical process arguments we can argue that
    \begin{equation}
        \label{eq:step_3}
        \phi(e^*)-\phi(\hat{e}_{\cf})+\hat{\phi}_n(\hat{e})-\hat{\phi}_n(e_n^*) = O_p(n^{-1/2}) + O_p(\alpha_n).
    \end{equation}
    Show that $\Pr(A_n) \rightarrow 1$ as $n \rightarrow \infty$ and conclude by the previous step that $\|e^*-\hat{e}_{\cf}\|_{2,P} + \|\hat{e}-e_n^*\|_{2,P_n} = O_p(n^{-1/4}) + O_p(\alpha_n^{-1/2})$.
    Then by the definitions of $\hat{e}_{\cf}$ and $e_n^*$ we can conclude that
    \[
    \|\hat{e}-e^*\|_{2,P} + \|\hat{e}-e^*\|_{2,P_n} = O_p(n^{-1/4}) + O_p(\alpha_n^{-1/2})
    \]
    as well.
    In particular, $\hat{e}$ is mean square consistent for $e^*$ both in-sample and out-of-sample.
    \item Apply a ``peeling" argument, similar to Theorem 3.2.5 in~\citet{van1996weak}, to show that $\|\hat{e}-\tilde{e}\|_{2,P_n} = O_p(\alpha_n)$, where 
    \[
    \tilde{e} \in \argmax_{e \in \cf_n} \Psi\biggl(n^{-1} \sum_{i=1}^n f(e(X_i),\eta(X_i))\biggr)
    \]
    is the propensity score we'd estimate with knowledge of $\eta$,
    i.e., by taking $\hat{\eta}=\eta$.
    Conclude by the previous step that $\|\hat{e}-e^*\|_{2,P_n} = O_p(n^{-1/4}) + O_p(\alpha_n)$,
    and show the same convergence rate holds for $\|\hat{e}-e^*\|_{2,P}$ by empirical process arguments.
\end{enumerate}

\subsection*{Step 1}
Strong concavity will be proven using calculus along with Assumptions~\ref{assump:info_matrix} and~\ref{assump:Psi}.
First, we notice that by Assumption~\ref{assump:Psi}(c) and~\eqref{eq:M_pos_def},
we know $\nabla \Psi(A)|_{A=M(e^*)} \succeq m^*I$
for some $m^*>0$.
Then by continuity of the smallest eigenvalue function $\lambda_{\min}(\cdot)$ and of $\nabla \Psi(\cdot)$,
there exists $r_0>0$ such that if $e \in \ce$ satisfies $\|e-e^*\|_{2,P} \leq r_0$ (which implies $\|M(e)-M(e^*)\| \leq 
Cr_0$ by~\eqref{eq:M_lipschitz}),
then $M(e) \succeq (k^*/2)I$ and $\nabla \Psi(A)|_{A=M(e)} \succeq (m^*/2)I$.
We now extend this argument to provide a high probability eigenvalue lower bound on $\hat{M}_n(e)$
for $e$ sufficiently close to $\hat{e}$ in $L^2(\cx,P_n)$:
\begin{lemma}
\label{lemma:stochastic_pos_def}
Suppose that all conditions of Lemma~\ref{lemma:concave_maximization} hold.
Fix any $k_0>0$ and define $A_n$ to be the event that $\hat{M}_n(\hat{e}) \succeq k_0I$.
Then there exist $\tilde{r}>0$ and $0<\tilde{k}<\tilde{K}$
such that  whenever $A_n$ holds,
for all $e \in \ce$ with $\|e-\hat{e}\|_{2,P_n} \leq \tilde{r}$
we have $\hat{M}_n(e) \succeq (k_0/2)I$ and $\tilde{K}I \succeq \nabla \Psi(\hat{M}_n(e)) \succeq \tilde{k}I$.
\end{lemma}
\begin{proof}[Proof of Lemma~\ref{lemma:stochastic_pos_def}]
The function $\lambda_{\min}(\cdot)$ is uniformly continuous on the compact subset 
$\cm = \{A:0 \preceq A \preceq CI\}$ of $\real^{p \times p}$.
Hence, when the event $A_n$ holds,
we know that there exists (nonrandom) $\tilde{\delta}>0$ such that $A \succeq (k_0/2)I$ for all $A \in \cm$ with $\|A-\hat{M}_n(\hat{e})\| \leq \tilde{\delta}$.
By~\eqref{eq:M_bounded},
$\cm$ contains both $\{M(e):e \in \ce\}$ and $\{\hat{M}_n(e):e \in \ce\}$.
Noting that
\begin{align}
\|\hat{M}_n(e_1)-\hat{M}_n(e_2)\| & \leq \frac{1}{n} \sum_{i=1}^n \|f(e_1(X_i),\hat{\eta}(X_i))-f(e_2(X_i),\hat{\eta}(X_i)\| \nonumber \\
& \leq \frac{C}{n} \sum_{i=1}^n |e_1(X_i)-e_2(X_i)| \nonumber \\
& \leq C \|e_1-e_2\|_{2,P_n}. \label{eq:M_n_lipschitz}
\end{align}
we see that
whenever $e \in \ce$ with $\|e-\hat{e}\|_{2,P_n} \leq \tilde{r} := \tilde{\delta}/C$,
we have $\|\hat{M}_n(e)-\hat{M}_n(\hat{e})\| \leq \tilde{\delta}$
and hence $\hat{M}_n(e) \succeq (k_0/2)I$.
The conclusion of Lemma~\ref{lemma:stochastic_pos_def} follows immediately by Assumption~\ref{assump:Psi}(c).
\end{proof}
Next, we bound directional derivatives of $\phi$ and $\hat{\phi}_n$.
\begin{lemma}
\label{lemma:phi_prime}
For any $e_1,e_2 \in \cf$ with $M(e_1)$ and $M(e_2)$ nonsingular,
the inequality
\begin{equation}
\label{eq:phi_prime_prime}
\frac{\mrd^2}{\mrd t^2} \phi(e_1+t(e_2-e_1)) \leq \tr\left[\nabla \Psi(M(e_1+t(e_2-e_1)))^{\top} \left(\frac{\mrd^2}{\mrd t^2} M(e_1+t(e_2-e_1))\right)\right]
\end{equation}
holds for each $t \in (0,1)$.
Similarly,
for any $e_1,e_2 \in \cf_n$ with $\hat{M}_n(e_1)$ and $\hat{M}_n(e_2)$ nonsingular,
we have the inequality
\begin{equation}
\label{eq:phi_hat_n_prime_prime}
\frac{\mrd^2}{\mrd t^2} \hat{\phi}_n(e_1+t(e_2-e_1)) \leq \tr\biggl[\nabla \Psi(\hat{M}_n(e_1+t(e_2-e_1)))^{\top}\biggl(\frac{\mrd^2}{\mrd t^2} \hat{M}_n(e_1+t(e_2-e_1))\biggr)\biggr].
\end{equation}
\end{lemma}
\begin{proof}[Proof of Lemma~\ref{lemma:phi_prime}]
Fix $e_1=e_1(\cdot)$ and $e_2=e_2(\cdot) \in \cf$ with $M(e_1)$ and $M(e_2)$ invertible.
Define
\[
\tilde{M}(t) = M(e_{(t)})
\]
where $e_{(t)}=e_1+t(e_2-e_1)$.
We first show that
\begin{align}
\tilde{M}'(t) & = \int_{\cx} (e_2(x)-e_1(x))f'(e_{(t)}(x),\eta(x)) \,\mrd P(x),\quad \forall t \in [0,1],\quad\text{and} \label{eq:M_prime} \\
\tilde{M}''(t) & = \int_{\cx} (e_2(x)-e_1(x))^2f''(e_{(t)}(x),\eta(x)) \,\mrd P(x),\quad \forall t \in (0,1). \label{eq:M_prime_prime}
\end{align}
We include the endpoints $t=0,1$ in~\eqref{eq:M_prime} so that we can apply Taylor's theorem with the Lagrange form of the second order remainder to complete the proof of Lemma~\ref{lemma:phi_prime}.
We could also strengthen~\eqref{eq:M_prime_prime} to include those endpoints,
but this will not be needed.

Consider the difference quotient
\[
D_1(t,h;X) = \frac{f(e_{(t+h)}(X),\eta(X))-f(e_{(t)}(X),\eta(X))}{h}
\]
By the chain rule we know that 
\[
\lim_{h \rightarrow 0} D_1(t,h;X) = \frac{\mrd}{\mrd t} f(e_{(t)}(X),\eta(X)) = (e_2(X)-e_1(X))f'(e_{(t)}(X),\eta(X)).
\]
Furthermore the fact that $e_1(x),e_2(x) \in [0,1]$ for all $x \in \cx$ indicates
\[
\tilde{\delta} \leq -h \leq e_{(t)}(x) \wedge e_{(t+h)}(x) \leq e_{(t)}(x) \vee e_{(t+h)}(x) \leq 1+h \leq 1-\tilde{\delta}
\]
for all $t \in [0,1]$ and $|h| \leq -\tilde{\delta}$.
By uniform boundedness of $f'$ on $[\tilde{\delta},1-\tilde{\delta}] \times \cw$ and Taylor's theorem,
and noting $e_{(t+h)}(x)-e_{(t)}(x) = h(e_2(x)-e_1(x))$
we conclude that 
\[
\sup_{0<|h| \leq -\tilde{\delta}} \|D_1(t,h;X)\| \leq |e_2(X)-e_1(X)| \times\sup_{e \in [\tilde{\delta},1-\tilde{\delta}]}  \|f'(e,\eta(X))\|  \leq C
\]
so by dominated convergence
\begin{align*}
\lim_{h \rightarrow 0} \frac{\tilde{M}(t+h)-\tilde{M}(t)}{h} &= \lim_{h \rightarrow 0} \int_{\cx} D_1(t,h;x)
\,\mrd P(x)\\
&= \int_{\cx} (e_2(x)-e_1(x))f'(e_{(t)}(x),\eta(x))
\,\mrd P(x)
\end{align*}
which establishes~\eqref{eq:M_prime}.

Similarly we define the second difference quotient 
\[
D_2(t,h;X) = \frac{f'(e_{(t+h)}(X),\eta(X))-f'(e_{(t)}(X),\eta(X))}{h}.
\]
By the chain rule we once again have
\[
\lim_{h \rightarrow 0} D_2(t,h;X) = \frac{\mrd}{\mrd t} f'(e_{(t)}(X),\eta(X)) = (e_2(X)-e_1(X))f''(e_{(t)}(X),\eta(X)).
\]
By uniform boundedness of $f''$ on  $[\tilde{\delta},1-\tilde{\delta}] \times \cw$ we get
\[
\sup_{0<|h| \leq -\tilde{\delta}} |D_2(t,h;X)| \leq |e_2(X)-e_1(X)| \times\sup_{e \in [\tilde{\delta},1-\tilde{\delta}]}  \|f''(e,\eta(X))\|  \leq C.
\]
Then in view of~\eqref{eq:M_prime} we can apply dominated convergence to conclude that
\begin{align*}
\lim_{h \rightarrow 0} \frac{\tilde{M}'(t+h)-\tilde{M}'(t)}{h} & = \lim_{h \rightarrow 0} \int_{\cx} D_2(t,h;x)(e_2(x)-e_1(x))
\,\mrd P(x) \\
& = \int_{\cx} (e_2(x)-e_1(x))^2f''(e_{(t)}(x),\eta(x)) 
\,\mrd P(x)
\end{align*}
establishing~\eqref{eq:M_prime_prime}.

Now we differentiate $\tilde{\phi}(t) := \phi(e_{(t)})$.
Note $M(e_{(t)}) \succ 0$ for all $t \in [0,1]$ by concavity of $M(\cdot)$,
shown previously.
Then using~\eqref{eq:M_prime} and~\eqref{eq:M_prime_prime},
we apply the chain rule to get
\begin{align*}
\tilde{\phi}'(t) & = \tr\bigl[\nabla \Psi(\tilde{M}(t))^{\top}\tilde{M}'(t)\bigr],\quad t\in[0,1].
\end{align*}
Similarly for all $t\in(0,1)$
\begin{align*}
\tilde{\phi}''(t) & = D^2 \Psi(\tilde{M}(t))(\tilde{M}'(t),\tilde{M}'(t)) + \tr[\nabla \Psi(\tilde{M}(t))^{\top} \tilde{M}''(t)]\\ & \leq \tr\bigl[\nabla \Psi(\tilde{M}(t))^{\top} \tilde{M}''(t)\bigr],
\end{align*}
which shows~\eqref{eq:phi_prime_prime}.
Here $D^2\Psi(\tilde{M}(t))$ is the second derivative mapping of $\Psi$ evaluated at $\tilde{M}(t)$,
viewed as a bilinear function from $\real^{p \times p} \times \real^{p \times p}$ to $\real$;
the inequality in the preceding display follows from concavity of $\Psi$.

Equation~\eqref{eq:phi_hat_n_prime_prime} follows by a very similar calculation,
though the argument is simplified,
since dominated convergence is no longer needed as we are dealing with finite sums instead of integrals.
Instead we immediately perform term-by-term differentiation to conclude
\begin{align}
\tilde{M}_n'(t) & = \frac1n \sum_{i=1}^n (e_2(X_i)-e_1(X_i))f'(e_{(t)}(X_i),\hat{\eta}(X_i))
\label{eq:M_n_prime}, \quad\forall t \in [0,1],\quad\text{and} \\
\tilde{M}_n''(t) & = \frac1n \sum_{i=1}^n (e_2(X_i)-e_1(X_i))^2f''(e_{(t)}(X_i),\hat{\eta}(X_i)) 
\label{eq:M_n_prime_prime}, \quad\forall t \in (0,1)
\end{align}
where $\tilde{M}_n(t) := \hat{M}_n(e_{(t)})$,
and then use the chain rule as above.
\end{proof}

We are now ready to prove~\eqref{eq:strong_concavity}.
We apply Lemma~\ref{lemma:phi_prime} with $e_1=e^*$ and any $e_2\in \cf$ with $\|e_2-e^*\|_{2,P} \leq r_0$.
Note that our definition of $r_0$
(in the paragraph before the statement of Lemma~\ref{lemma:stochastic_pos_def})
along with~\eqref{eq:M_pos_def} ensures that $M(e_1)$ and $M(e_2)$ are nonsingular. 
Also, note that $\tilde{\phi}'(0) \leq 0$ by optimality of $e^*$.
Lemma~\ref{lemma:phi_prime} along with Taylor's theorem with the Lagrange form of the remainder then enables us to conclude
\[
\phi(e_2) = \tilde{\phi}(1) = \tilde{\phi}(0) + \tilde{\phi}'(0) + \frac{1}{2}\tilde{\phi}''(t) \leq \phi(e^*) + \frac{1}{2}\tr\bigl[\nabla \Psi(\tilde{M}(t))^{\top}\tilde{M}''(t)\bigr]
\]
for some $t \in (0,1)$,
where $\tilde{M}(t)=M(e_1+t(e_2-e_1))$ as in the proof of Lemma~\ref{lemma:phi_prime}.
By~\eqref{eq:f_strong_concavity} and~\eqref{eq:M_prime_prime},
we know that $\tilde{M}''(t) \preceq 0$.
Recalling that the trace of the product of two symmetric positive semidefinite matrices is nonnegative,
we have
\[
0 \geq \tr[(\nabla \Psi(\tilde{M}(t))-(m^*/2)I)^{\top}\tilde{M}''(t)] = \tr[\nabla \Psi(\tilde{M}(t))^{\top}\tilde{M}''(t)] - \frac{m^*}{2}\tr(\tilde{M}''(t))
\]
since $\nabla \Psi(\tilde{M}(t)) \succeq (m^*/2)I$.
Then
\begin{align*}
\tr[\nabla \Psi(\tilde{M}(t))^{\top}\tilde{M}''(t)] &\leq \frac{m^*}{2} \tr(\tilde{M}''(t)) \\
&\leq -\frac{cm^*}{2}\int_{\cx}(e^*(x)-e_2(x))^2 \,\mrd P(x) \\
&= -\frac{cm^*}{2} \|e^*-e_2\|_{2,P}^2
\end{align*}
where the second inequality follows by~\eqref{eq:f_strong_concavity} and~\eqref{eq:M_prime_prime} once again.
We conclude that whenever $e_2 \in \cf$ with $\|e_2-e^*\|_{2,P} \leq r_0$ we have
\[
\phi(e^*) \geq \phi(e_2) + \frac{cm^*}{4} \|e_2-e^*\|_{2,P}^2.
\]

Now take any $e_2 \in \cf$ with $\|e_2-e^*\|_{2,P} > r_0$.
Define $t=1-r_0/\|e_2-e^*\|_{2,P} \in (0,1)$ and consider $\tilde{e}_2 = te^*+(1-t)e_2$ so that $\|\tilde{e}_2-e^*\|_{2,P} = r_0$.
Note $\tilde{e}_2 \in \cf$ by convexity of $\cf_*$, so by the preceding display
\[
\phi(e^*) \geq \phi(\tilde{e}_2) + \frac{cm^*r_0^2}{4} \geq t\phi(e^*) + (1-t)\phi(e_2) + \frac{cm^*r_0^2}{4}
\]
where the second inequality is by concavity of $\phi$. Rearranging we have
\[
\phi(e^*) \geq \phi(e_2) + \frac{cm^*r_0^2}{4(1-t)} = \phi(e_2) + \frac{cm^*r_0}{4}\|e^*-e_2\|_{2,P}.
\]
Letting $c_0=cm^*/4$, we conclude that for all $e \in \cf$ we have
\[
\phi(e^*) \geq \phi(e) + c_0\min(r_0,\|e^*-e\|_{2,P})\|e^*-e\|_{2,P}
\]
which shows~\eqref{eq:strong_concavity}.

The proof of~\eqref{eq:strong_concavity_n} is quite similar.
Take $k_0>0$ such that
with $k^*$ as in~\eqref{eq:M_pos_def},
whenever $0 \preceq A \preceq 
CI$ with $\Psi(A) \geq \inf_{B \succeq (k^*/4)I} \Psi(B)$,
we have $A \succeq k_0I$.
Such a $k_0$ exists by Assumptions~\ref{assump:Psi}(a) and~\ref{assump:Psi}(d).
By Lemma~\ref{lemma:stochastic_pos_def},
on the truncation event $A_n$ that
$\hat{M}_n(\hat{e}) \succeq k_0I$,
we have $\tilde{K}I \geq \nabla \Psi(\hat{M}_n(e)) \geq \tilde{k}I$ whenever $\|e-\hat{e}\|_{2,P_n} \leq \tilde{r}$,
for some $\tilde{r}>0$ and $0<\tilde{k} < \tilde{K}$.
Now we apply~\eqref{eq:phi_hat_n_prime_prime} with $e_1 = \hat{e}$ and any $e_2 \in \cf_n$ with $\|e_2-\hat{e}\|_{2,P_n} \leq \tilde{r}$
(note we must have $\hat{M}_n(e_1)$ and $\hat{M}_n(e_2)$ nonsingular).
Defining $\tilde{\phi}_n(t) := \hat{\phi}_n(e_{(t)})$ for $e_{(t)} = e_1+t(e_2-e_1)$,
by optimality of $\hat{e}$ we must have $\tilde{\phi}_n'(0) \leq 0$.
Taylor's theorem and the second part of Lemma~\ref{lemma:phi_prime} then allow us to conclude that
\begin{align}
\label{eq:taylor_phi_n}
\hat{\phi}_n(e_2)\indic(A_n) = \tilde{\phi}_n(1)\indic(A_n) & = \Bigl(\tilde{\phi}_n(0) + \tilde{\phi}_n'(0) + \frac{1}{2} \tilde{\phi}_n''(t)\Bigr)\indic(A_n) \nonumber \\
& \leq \Bigl(\hat{\phi}_n(\hat{e}) + \frac{1}{2} \tr\bigl[\nabla \Psi(\tilde{M}_n(t))^{\top} \tilde{M}_n''(t)\bigr]\Bigr)\indic(A_n)
\end{align}
for some $t \in (0,1)$,
where $\tilde{M}_n(t)=\hat{M}_n(e_{(t)})$ as in the proof of Lemma~\ref{lemma:phi_prime}.
With $\nabla \Psi(\tilde{M}_n(t)) \succeq \tilde{k}I$ for all $t \in (0,1)$ whenever $A_n$ holds,
we have by~\eqref{eq:M_n_prime_prime} that
\[
0 \geq \tr\bigl[(\nabla \Psi(\tilde{M}_n(t))-\tilde{k}I)^{\top}\tilde{M}_n''(t)\bigr]\indic(A_n) = \bigl(\tr[(\nabla \Psi(\tilde{M}_n(t))^{\top}\tilde{M}_n''(t)] -\tilde{k}\tr[\tilde{M}_n''(t)]\bigr)\indic(A_n).
\]
Hence by~\eqref{eq:f_strong_concavity}
\[
\tr[(\nabla \Psi(\tilde{M}_n(t))^{\top}\tilde{M}_n''(t)]\indic(A_n) \leq \tilde{k}\tr[\tilde{M}_n''(t)]\indic(A_n) \leq -\tilde{k}c\|\hat{e}-e\|_{2,P_n}^2 \indic(A_n).
\]
Then by~\eqref{eq:taylor_phi_n} we conclude that whenever $\|e-\hat{e}\|_{2,P_n} \leq \tilde{r}$ we have
\[
(\hat{\phi}_n(\hat{e})-\hat{\phi}_n(e))\indic(A_n) \geq c\frac{\tilde{k}}{2} \|\hat{e}-e\|_{2,P_n}^2 \indic(A_n).
\]
Redefining $r_0$ to be the minimum of the $r_0$ appearing in the proof of~\eqref{eq:strong_concavity} and $\tilde{r}$,
since $\hat{\phi}_n$ is always concave on $\cf_n$.
we can repeat the argument at the end of the proof of~\eqref{eq:strong_concavity}
to conclude that
\[
(\hat{\phi}_n(\hat{e})-\hat{\phi}_n(e)) \indic(A_n) \geq \left(c_0\min(r_0,\|\hat{e}-e\|_{2,P_n})\|\hat{e}-e\|_{2,P_n}\right)\indic(A_n)
\]
for all $e \in \cf_n$;
here $c_0 := c\tilde{k}/2 > 0$.

\subsection*{Step 2}
Fix $\hat{e}_{\cf} \in \cf_*$ and $e_n^* \in \cf_n$.
The result of Step 1 shows that for some positive constants $c_0$ and $r_0$,
we have
\begin{align*}
c_0 \min(r_0, \|e^*-\hat{e}_{\cf}\|_{2,P})\|e^*-\hat{e}_{\cf}\|_{2,P} & \leq \phi(e^*)-\phi(\hat{e}_{\cf}),\quad\text{and} \\
c_0 \min(r_0, \|\hat{e}-e_n^*\|_{2,P_n})\|\hat{e}-e_n^*\|_{2,P_n}\indic(A_n) & \leq (\hat{\phi}_n(\hat{e})-\hat{\phi}_n(e_n^*))\indic(A_n).
\end{align*}
With $\phi(e^*)-\phi(\hat{e}_{\cf}) \geq 0$ and $\hat{\phi}_n(\hat{e})-\hat{\phi}_n(e_n^*) \geq 0$
by the definitions of $e^*$ and $\hat{e}$,
we can further upper bound the right-hand sides by
\begin{align*}
\phi(e^*)-\phi(\hat{e}_{\cf}) & \leq \phi(e^*)-\phi(\hat{e}_{\cf}) + \hat{\phi}_n(\hat{e})-\hat{\phi}_n(e_n^*),\quad\text{and} \\
(\hat{\phi}_n(\hat{e})-\hat{\phi}_n(e_n^*))\indic(A_n) & \leq \phi(e^*)-\phi(\hat{e}_{\cf}) +  \hat{\phi}_n(\hat{e})-\hat{\phi}_n(e_n^*).
\end{align*}

\subsection*{Step 3}
For brevity, in this section
we introduce the empirical process notation
\[
Pe = \int_{\cx} e(x) \,\mrd P(x)
\]
for all $e \in \ce$.
For instance, with $P_n$ the empirical measure induced by $X_1,\ldots,X_n$, we have $P_n e = n^{-1} \sum_{i=1}^n X_i$ for all $e \in \ce$.

We first show that we can choose particular $\hat{e}_{\cf}=\hat{e}_{\cf}(\cdot) \in \cf$ and $e_n^*(\cdot) \in \cf_n$ that are very close in sup norm to $\hat{e} \in \cf_n$ and $e^* \in \cf$, respectively,
with high probability.
\begin{lemma}
\label{lemma:e_0_e_n_star}
Under the conditions of Lemma~\ref{lemma:concave_maximization},
there exist $\hat{e}_{\cf}=\hat{e}_{\cf}(\cdot) \in \cf$ and $e_n^*=e_n^*(\cdot) \in \cf_n$ such that
$\sup_{x \in \cx} |\hat{e}(x)-\hat{e}_{\cf}(x)| + \sup_{x \in \cx} |e^*(x)-e_n^*(x)| = O_p(n^{-1/2})$.
\end{lemma}
\begin{proof}[Proof of Lemma~\ref{lemma:e_0_e_n_star}]
With the functions in $\ce$ uniformly bounded by 1,
by Lemma~\ref{lemma:empirical_process}
we conclude
\[
\e_P\left[\sup_{e \in \ce} \Big|P_ne - Pe\Big|\right] \leq KC n^{-1/2}
\]
so that $\sup_{e \in \ce} \big|P_ne - Pe\big| = O_p(n^{-1/2})$ by Markov's inequality.
In view of the fact that $m_L \leq P_n \hat{e} \leq m_H$ and $m_L \leq P e^* \leq m_H$ since $\hat{e} \in \cf_n$ and $e^* \in \cf$,
we have 
\begin{align}m_L -\sup_{e \in \ce} \big|P_ne - Pe\big| & \leq P_n e^* \leq m_H+\sup_{e \in \ce} \big|P_ne - Pe\big|,\quad\text{and}\label{eq:P_n_e_star} \\
m_L - \sup_{e \in \ce} \big|P_ne - Pe\big| & \leq P\hat{e} \leq m_H+\sup_{e \in \ce} \big|P_ne - Pe\big|. \label{eq:P_hat_e}
\end{align}

Next, with $e_L=e_L(\cdot)$ and $e_H=e_H(\cdot)$ as in the assumptions of Lemma~\ref{lemma:concave_maximization}, 
define
\[
e_n^*(x) = 
\begin{cases}
e^*(x), & m_L \leq P_n e^* \leq m_H \\
e^*(x) + \lambda_n(e_L(x)-e^*(x)), & P_n e^* < m_L \\
e^*(x) + \lambda_n(e_H(x)-e^*(x)), & P_n e^* > m_H
\end{cases}
\]
where 
\[
\lambda_n = 
\begin{cases}
\dfrac{m_L-P_n e^*}{P_n[e_L-e^*]}, & P_n e^* < m_L \\[2ex]
\dfrac{P_n e^*-m_H}{P_n[e^*-e_H]}, & P_n e^* > m_H \\[2ex]
0 & \text{otherwise.}
\end{cases}
\]
On the event 
\[
A_n = \left\{P_n e_L \geq \frac{P e_L + m_L}{2}, P_n e_H \leq \frac{Pe_H+m_H}{2} \right\}
\]
we must have $0 \leq \lambda_n \leq 1$ since $Pe_L > m_L$ and $Pe_H < m_H$,
and so when $A_n$ holds
we know $e_n^* \in \ce$ by convexity of $\ce$.
Furthermore we have $P_n e_n^* = \max(m_L,\min(m_H,\e_{P_n}[e^*(X)]))$ so that in fact $e_n^* \in \cf_n$.
But with $M = 2\max\left((Pe_L-m_L)^{-1},(m_H-Pe_H)^{-1}\right)$ we have
\begin{align*}
0 \leq \lambda_n\indic(A_n) & \leq M\indic(A_n)[(m_L-P_n e^*)\indic(P_ne^* < m_L) + (P_ne^*-m_H)\indic(P_ne^* > m_H)] \\
& \stackrel{\eqref{eq:P_n_e_star}}{\leq} 2M\sup_{e \in \ce}|P_ne-P_e| \indic(A_n).
\end{align*}
With $\Pr(A_n) \rightarrow 1$ as $n \rightarrow \infty$ by the law of large numbers,
we get $\lambda_n = O_p(n^{-1/2})$,
and hence for each $x \in \cx$ we have
\[
|e_n^*(x)-e^*(x)| \leq \lambda_n(|e_L(x)-e^*(x)| \vee |e_H(x)-e^*(x)|) \leq \lambda_n = O_p(n^{-1/2}).
\]

Next, define
\[
\hat{e}_{\cf}(x) = 
\begin{cases}
\hat{e}(x), & m_L \leq P\hat{e} \leq m_H \\
\hat{e}(x) + \tilde{\lambda}_n(e_L(x)-\hat{e}(x)), & P\hat{e} < m_L \\
\hat{e}(x) + \tilde{\lambda}_n(e_H(x)-\hat{e}(x)), & P\hat{e} > m_H
\end{cases}
\]
with
\[
\tilde{\lambda}_n = 
\begin{cases}
\dfrac{m_L-P \hat{e}}{m_H-\hat{e}}, & P\hat{e} < m_L \\[2ex]
\dfrac{P \hat{e}-m_H}{P\hat{e}-m_L}, & P\hat{e} > m_H \\[2ex]
0, & \text{ otherwise}
\end{cases}
\]
so that $0 \leq \tilde{\lambda}_n \leq 1$ and $\hat{e}_{\cf} \in \cf$ always.
In fact
\begin{align*}
\tilde{\lambda}_n &\leq \frac{M}{2}\bigl[(m_L-P\hat{e})\indic(P\hat{e}<m_L)+(P\hat{e}-m_H)\indic(P\hat{e} > m_H)\bigr] \\
&\leq M\sup_{e \in \ce}|P_n e - Pe|
\end{align*}
so that $\tilde{\lambda}_n = O_p(n^{-1/2})$ as well and the lemma follows since for all $x \in \cx$
\[
|\hat{e}(x)-\hat{e}_{\cf}(x)| \leq \tilde{\lambda}_n(|e_L(x)-\hat{e}(x)| \vee |e_H(x)-\hat{e}(x)|) \leq \tilde{\lambda}_n.
\]
\end{proof}
We are now ready to prove consistency.
Taking $\hat{e}_{\cf}$ as in Lemma~\ref{lemma:e_0_e_n_star},
we upper bound the right-hand side of the inequalities in step 2:
\begin{align}\begin{split}
\label{eq:basic_inequality}
\phi(e^*)-\phi(\hat{e}_{\cf})+\hat{\phi}_n(\hat{e})-\hat{\phi}_n(e_n^*) &\leq |\phi(e^*)-\phi(e_n^*)| + |\phi(e_n^*)-\hat{\phi}_n(e_n^*)| \\
&\phe\ + |\hat{\phi}_n(\hat{e})-\hat{\phi}_n(\hat{e}_{\cf})| + |\hat{\phi}_n(\hat{e}_{\cf})-\phi(\hat{e}_{\cf})|.
\end{split}\end{align}
In view of~\eqref{eq:M_lipschitz} and~\eqref{eq:M_n_lipschitz},
Lemma~\ref{lemma:e_0_e_n_star} shows that 
\begin{align}
\|M(e^*)-M(e_n^*)\| & = O_p(n^{-1/2}),\quad\text{and} \label{eq:M_star_diff} \\
\|\hat{M}_n(\hat{e})-\hat{M}_n(\hat{e}_{\cf})\| & = O_p(n^{-1/2}). \label{eq:M_n_hat_diff}
\end{align}
We now use~\eqref{eq:M_star_diff} to show that
\begin{equation}
\label{eq:e_star_e_n_star}
|\phi(e^*)-\phi(e_n^*)|=O_p(n^{-1/2}).
\end{equation}
As shown at the start of step 1,
whenever $\|e_n^*-e^*\|_{2,P} \leq r_0$
we have $M(e_n^*) \succeq (k^*/2)I$,
so that
\[
tM(e^*) + (1-t)M(e_n^*) \succeq (k^*/2)I\quad \forall t \in [0,1].
\]
Applying Taylor's theorem to $\Psi(\cdot)$ we have for some $K<\infty$ that
\begin{align}
|\phi(e^*)-\phi(e_n^*)|\indic(\tilde{A}_n) & = |\Psi(M(e^*))-\Psi(M(e_n^*))|\indic(\tilde{A}_n)  \nonumber \\
& \leq \sup_{t \in [0,1]}\Big|\tr[\nabla \Psi(tM(e^*) + (1-t)M(e_n^*))]^{\top}[M(e^*)-M(e_n^*)])\Big|\indic(\tilde{A}_n)  \nonumber \\
& \leq \sup_{t \in [0,1]} \|\nabla \Psi(tM(e^*) + (1-t)M(e_n^*))\| \cdot \|M(e^*)-M(e_n^*)\|\indic(\tilde{A}_n) \nonumber \\
& \leq K\sqrt{p} \cdot \|M(e^*)-M(e_n^*)\|\indic(\tilde{A}_n) \label{eq:taylor} 
\end{align}
where $\tilde{A}_n$ is the event $\|e_n^*-e^*\|_{2,P} \leq r_0$ and the last inequality follows from Assumption~\ref{assump:Psi}(c)
and the fact that $\|A\| \leq \sqrt{p}\lambda_{\max}(A)$ for any $A \in \ess_+^p$.
Here $\lambda_{\max}(A)$ denotes the largest eigenvalue of $A$.
Hence $|\phi(e^*)-\phi(e_n^*)|\indic(A_n)=O_p(n^{-1/2})$ by~\eqref{eq:M_star_diff}.
But $\Pr(\tilde{A}_n) \rightarrow 1$ by~\eqref{eq:M_star_diff},
so indeed $|\phi(e^*)-\phi(e_n^*)| = O_p(n^{-1/2})$.

Convergence of the remaining three terms in~\eqref{eq:basic_inequality} depends on the following result:
\begin{equation}
\label{eq:supremum}
\sup_{e \in \ce} \|\hat{M}_n(e)-M(e)\| = O_p(n^{-1/2}) + O_p(\alpha_n).
\end{equation}
To show this, define
\[
M_n(e) = \frac{1}{n} \sum_{i=1}^n f(e(X_i),\eta(X_i))
\]
which replaces the estimated nuisance function $\hat{\eta}$ in the definition of $\hat{M}_n$ with the true $\eta$.
First note that by uniform boundedness of $\|f_w\|_2$,
\[
\frac{1}{n} \sum_{i=1}^n \sup_{e \in [0,1]}\|f(e,\hat{\eta}(X_i))-f(e,\eta(X_i))\| \leq Cn^{-1} \sum_{i=1}^n \|\hat{\eta}(X_i)-\eta(X_i)\|_2 \leq C\|\hat{\eta}-\eta\|_{2,P_n}.
\]
Then by~\eqref{eq:eta_n} we have
\begin{align}
\sup_{e \in \ce} \|\hat{M}_n(e)-M_n(e)\| & = \sup_{e \in \ce} \bigg\|\frac{1}{n} \sum_{i=1}^n f(e(X_i),\hat{\eta}(X_i))-f(e(X_i),\eta(X_i))\bigg\| \nonumber \\
& \leq \frac{1}{n} \sum_{i=1}^n \sup_{e \in [0,1]} \|f(e,\hat{\eta}(X_i))-f(e,\eta(X_i))\| \nonumber \\
& = O_p(\alpha_n). \label{eq:f_hat_minus_f}
\end{align}
Now for $i,j \in \{1,\ldots,p\}$ define the class of functions
\[
\cg_{ij} = \{x \mapsto f_{ij}(e(x),\eta(x)) \mid e \in \ce\} 
\]
where $f_{ij}$ denotes the $(i,j)$-th entry of the function $f$.
By Lemma~\ref{lemma:covering_bracketing}
\begin{align*}
n^{-1/2} \int_0^1 \sqrt{\log \cn(r,\cg_{ij},L^2(P_n))}\,\mrd r & = C n^{-1/2} \int_0^{C^{-1}} \sqrt{\log \cn(C\epsilon,\cg_{ij},L^2(P_n))} \,\mrd\epsilon  \\
 & \leq Cn^{-1/2} \int_0^{C^{-1}} \sqrt{\log \cn(\epsilon,\ce,L^2(P_n))} \,\mrd\epsilon.
\end{align*}
Let
\[
D(e) = \frac{1}{n} \sum_{i=1}^n f(e(X_i),\eta(X_i)) 
- \int_{\cx} f(e(x),\eta(x))
\,\mrd P(x) = M_n(e)-M(e).
\]
Then Lemma~\ref{lemma:empirical_process} and Assumption~\ref{assump:info_matrix} indicate that $\sup_{e \in \ce} |D_{ij}(e)| = O_p(n^{-1/2})$ and so
\begin{equation}
\label{eq:D_e}
\sup_{e \in \ce} \|D(e)\| = \sup_{e \in \ce} \biggl(\,\sum_{i=1}^p\sum_{j=1}^p D_{ij}(e)^2\biggr)^{1/2} \leq \biggl(\,\sum_{i=1}^p \sum_{j=1}^p \sup_{e \in \ce} |D_{ij}(e)|^2\biggr)^{1/2} = O_p(n^{-1/2}).
\end{equation}
The result~\eqref{eq:supremum} follows by the triangle inequality in view of~\eqref{eq:f_hat_minus_f} and~\eqref{eq:D_e}.

We are finally ready to bound the remaining terms on the right-hand side of~\eqref{eq:basic_inequality},
and show that $\Pr(A_n) \rightarrow 1$ as $n \rightarrow \infty$
where $A_n$ is the event $\hat{M}_n(\hat{e}) \succeq k_0I$ for $k_0>0$ defined in Step 1.
Choose $\delta>0$ so that for any $A_1$, $A_2$ in $\cm = \{A \in \ess_+^p: 0 \preceq A \preceq CI\}$,
with $\|A_1-A_2\| \leq \delta$,
we have $|\lambda_{\min}(A_1)-\lambda_{\min}(A_2)| \leq \min(k_0/2,k^*/4)$,
where $k^*$ satisfies~\eqref{eq:M_pos_def}.
Such $\delta$ exists by uniform continuity of $\lambda_{\min}(\cdot)$ on the compact subset $\cm$ of $\real^{p \times p}$ (cf. the proof of Lemma~\ref{lemma:stochastic_pos_def}).
Also define the event $B_n$ that all of the following are true:
\begin{align}
\|e_n^*-e^*\|_{2,P} & \leq r_0, \label{eq:e_n_star}  \\
\sup_{e \in \ce} \|\hat{M}_n(e)-M(e)\| & \leq \delta\quad \text{and} \label{eq:M_n_minus_M} \\
\|\hat{M}_n(\hat{e})-\hat{M}_n(\hat{e}_{\cf})\| & \leq \delta. \label{eq:M_n_hat} 
\end{align}
We claim that $B_n$ implies the following conditions:
$$M(e_n^*) \succeq \frac{k^*}{2}I,\quad
\hat{M}_n(e_n^*) \succeq \frac{k^*}{4}I,\quad
\hat{M}_n(\hat{e}) \succeq k_0I\quad\text{and}\quad
\hat{M}_n(\hat{e}_{\cf}) \succeq \frac{k_0}{2}I.$$
We prove these statements briefly.
Assume $B_n$ holds.
First, note that $M(e_n^*) \succeq (k^*/2)I$ holds by definition of $r_0$ and~\eqref{eq:e_n_star}.
Next, the definition of $\delta$ immediately ensures by~\eqref{eq:M_n_minus_M} that
$\hat{M}_n(e_n^*) \succeq (k^*/4)I$.
But then
\[
\Psi(\hat{M}_n(\hat{e})) \geq \Psi(\hat{M}_n(e_n^*)) \geq \inf_{B \succeq (k^*/4)I} \Psi(B)
\]
so that $\hat{M}_n(\hat{e}) \succeq k_0I$ by the definition of $k_0$,
and in particular we have shown $B_n \subseteq A_n$.
Finally~\eqref{eq:M_n_hat} shows $\hat{M}_n(\hat{e}_{\cf}) \succeq (k_0/2)I$.

Now take $\tilde{K}<\infty$ to be as derived from Assumption~\ref{assump:Psi}(c) with $k=\min(k^*/4,k_0/2)$ and 
$K=C$.
Then repeated applications of arguments analogous to~\eqref{eq:taylor} show that
\begin{align*}
|\phi(e_n^*)-\hat{\phi}_n(e_n^*)|\indic(B_n) &= |\Psi(M(e_n^*))-\Psi(\hat{M}_n(e_n^*))|\indic(B_n)\\
& \leq \sup_{0 \preceq A \preceq \tilde{K}I} |\tr(A^{\top}[M(e_n^*)-\hat{M}_n(e_n^*)])| \\ 
& \leq \tilde{K}\sqrt{p} \cdot \|M(e_n^*)-\hat{M}_n(e_n^*)\| \\
& = O_p(n^{-1/2}) + O_p(\alpha_n)  
\end{align*}
by~\eqref{eq:supremum} and similarly
\begin{align*}
|\hat{\phi}_n(\hat{e})-\hat{\phi}_n(\hat{e}_{\cf})|\indic(B_n) \leq \tilde{K}\sqrt{p} \cdot \|\hat{M}_n(\hat{e})-\hat{M}_n(\hat{e}_{\cf})\| = O_p(n^{-1/2})
\end{align*}
by~\eqref{eq:M_n_hat_diff}. Also
\begin{align*}
|\hat{\phi}_n(\hat{e}_{\cf})-\phi(\hat{e}_{\cf})|\indic(B_n) \leq \tilde{K}\sqrt{p} \cdot \|\hat{M}_n(\hat{e}_{\cf})-M(\hat{e}_{\cf})\|  = O_p(n^{-1/2}) + O_p(\alpha_n) 
\end{align*}
by~\eqref{eq:supremum}.

However, by~\eqref{eq:M_star_diff},~\eqref{eq:M_lipschitz},~\eqref{eq:M_n_hat_diff}, and~\eqref{eq:supremum}
we know that $\Pr(B_n) \rightarrow 1$ as $n \rightarrow \infty$.
Since $B_n \subseteq A_n$ we also have $\Pr(A_n) \rightarrow 1$.
We conclude from the preceding displays that
\[
|\phi(e_n^*)-\hat{\phi}_n(e_n^*)| + |\hat{\phi}_n(\hat{e})-\hat{\phi}_n(\hat{e}_{\cf})| + |\hat{\phi}_n(\hat{e}_{\cf})-\phi(\hat{e}_{\cf})| = O_p(n^{-1/2}) + O_p(\alpha_n).
\]
Then by step 2,~\eqref{eq:basic_inequality}, and~\eqref{eq:e_star_e_n_star},
we conclude that 
\begin{align*}
\|e^*-\hat{e}_{\cf}\|_{2,P} & = O_p(n^{-1/4}) + O_p(\alpha_n^{-1/2})\quad\text{and} \\
\|\hat{e}-e_n^*\|_{2,P_n} & = O_p(n^{-1/4}) + O_p(\alpha_n^{-1/2}).
\end{align*}
But by Lemma~\ref{lemma:e_0_e_n_star} we know that $\|\hat{e}-\hat{e}_{\cf}\|_{2,P}+\|e^*-e_n^*\|_{2,P_n}=O_p(n^{-1/2})$.
So by the triangle inequality we conclude $\|e^*-\hat{e}\|_{2,P} + \|e^*-\hat{e}\|_{2,P_n} = O_p(n^{-1/4}) + O_p(\alpha_n^{-1/2})$ as well.

\subsection*{Step 4}
The final step in the argument to derive our best convergence rates is a variation of a standard ``peeling" argument used in deriving convergence rates of $M$-estimators.
The main argument requires deriving a bound on the ``locally centered empirical process" as in our next result.
\begin{lemma}
\label{lemma:centered_empirical_process}
For each $e \in \ce$, let $\phi_n(e) = \Psi(M_n(e))$
and take $\tilde{e} \in \argmax_{e \in \cf_n} \phi_n(e)$.
Then there exists $\beta>0$ and a universal constant $C_0 < \infty$ such that for all $u \leq \beta$
\begin{multline}\label{eq:centered_empirical_process}
\sup_{e \in \ce: \|e-\tilde{e}\|_{2,P_n} \leq u} \bigl[(\hat{\phi}_n(e)-\phi_n(e))-(\hat{\phi}_n(\tilde{e})-\phi_n(\tilde{e}))\bigr]\indic(B_n) 
\leq C_0\bigl(u\|\hat{\eta}-\eta\|_{2,P_n}+\|\hat{\eta}-\eta\|_{2,P_n}^2\bigr)
\end{multline}
for some sequence of events $B_n$ with $\Pr(B_n) \rightarrow 1$ as $n \rightarrow \infty$.
\end{lemma}
\begin{proof}[Proof of Lemma~\ref{lemma:centered_empirical_process}]
Let $U_n(e) = \hat{\phi}_n(e)-\phi_n(e)$ for each $e \in \ce$.
By Taylor's theorem, for each $e \in \ce$
\[
U_n(e) = \Psi(\hat{M}_n(e))-\Psi(M_n(e)) = \nabla \Psi(R_n(e))^{\top}(\hat{M}_n(e)-M_n(e))
\]
for some $R_n(e)$ lying on the line segment between $M_n(e)$ and $\hat{M}_n(e)$.
Write
\begin{align}
U_n(e)-U_n(\tilde{e}) & = (\nabla \Psi(R_n(e)) - \nabla \Psi(R_n(\tilde{e})))^{\top}(\hat{M}_n(e)-M_n(e)) \nonumber \\
&\phe\ + \nabla \Psi(R_n(\tilde{e}))^{\top}[(\hat{M}_n(e)-M_n(e))-(\hat{M}_n(\tilde{e})-M_n(\tilde{e}))] \label{eq:U_n_diff}.
\end{align}
Because $\hat{\eta}=\eta$ trivially satisfies the conditions of Lemma~\ref{lemma:concave_maximization},
all of our results in steps 1\fromart{--}3 apply if we replace $\hat{M}_n$ with $M_n$.
In particular with $k_0>0$ as chosen in step 1,
we have $\Pr(M_n(\tilde{e}) \succeq k_0I) \rightarrow 1$ as $n \rightarrow \infty$
by the argument at the end of step 3.
Take $\tilde{r}$, $\tilde{k}$, and $\tilde{K}$ derived from Lemma~\ref{lemma:stochastic_pos_def} with this choice of $k_0$.
By the proof of Lemma~\ref{lemma:stochastic_pos_def},
we know that for all $0 \preceq A \preceq CI$ with $\|A-M_n(\tilde{e})\| \leq C\tilde{r}$,
we must have $A \succeq (k_0/2)I$ and $\tilde{K}I \succeq \nabla \Psi(A) \succeq \tilde{k}I$.
Let $B_n$ be the intersection of the events $M_n(\tilde{e}) \succeq k_0I$ and 
$\sup_{e \in \ce} \|\hat{M}_n(e)-M_n(e)\| \leq C\tilde{r}/3$.
Now $\Pr(B_n) \rightarrow 1$ as $n \rightarrow \infty$ by~\eqref{eq:f_hat_minus_f}.
Then for any $e \in \ce$ with $\|e-\tilde{e}\|_{2,P_n} \leq \tilde{r}/3$,
in view of~\eqref{eq:M_n_lipschitz} we have
\begin{align*}
\|\hat{M}_n(e)-M_n(\tilde{e})\|\indic(B_n) & \leq \|\hat{M}_n(e)-\hat{M}_n(\tilde{e})\| + \|\hat{M}_n(\tilde{e})-M_n(\tilde{e})\|\indic(B_n) \\
&\leq \frac{2C\tilde{r}}{3}\quad\text{and}  \\
\|M_n(e)-M_n(\tilde{e})\|\indic(B_n) & \leq \|M_n(e)-\hat{M}_n(e)\|\indic(B_n) + \|\hat{M}_n(e)-M_n(\tilde{e})\|\indic(B_n) \\
&\leq C\tilde{r}.
\end{align*}
We conclude by the preceding display
that whenever $B_n$ holds,
for all $e \in \ce$ with $\|e-\tilde{e}\|_{2,P_n} \leq \tilde{r}/3$ we have 
$(k_0/2)I \preceq \hat{M}_n(e) \preceq C
I$ and $(k_0/2)I \preceq M_n(e) \preceq C
I$,
and thus $(k_0/2)I \preceq R_n(e) \preceq C
I$ along with $\tilde{K}I \succeq \nabla \Psi(R_n(e)) \succeq \tilde{k}I$.
Then by Assumption~\ref{assump:Psi}(b), 
for all $u \leq \tilde{r}/3$ there exists a constant $K_0<\infty$ (independent of $u$) for which
\begin{align*}
\|\nabla \Psi(R_n(e)) - \nabla \Psi(R_n(\tilde{e}))\|\indic(B_n) & \leq K_0\|R_n(e)-R_n(\tilde{e})\|\indic(B_n) \\
& \leq K_0(\|R_n(e)-\hat{M}_n(e)\| + \|\hat{M}_n(e)-\hat{M}_n(\tilde{e})\| + \|\hat{M}_n(\tilde{e})-R_n(\tilde{e})\|)
\end{align*}
The preceding inequality holds for all $e \in \ce$ with $\|e-\tilde{e}\|_{2,P_n} \leq u$.
Taking a supremum over such $e(\cdot)$,
another application of~\eqref{eq:M_n_lipschitz}
and the fact that $\|R_n(e)-\hat{M}_n(e)\| \leq \|M_n(e)-\hat{M}_n(e)\|$ for all $e \in \ce$
show that
\[
\sup_{e \in \ce: \|e-\tilde{e}\|_{2,P_n} \leq u} \|\nabla \Psi(R_n(e)) - \nabla \Psi(R_n(\tilde{e}))\|\indic(B_n)  \leq K_0\bigl(Cu+2S_n(\ce)\bigr)
\]
where $S_n(\ce) = \sup_{e \in \ce} \|\hat{M}_n(e)-M_n(e)\|$.
Then by Cauchy-Schwarz
\begin{equation}\begin{split}
\label{eq:peeling_sup_1}
\phe\ &\sup_{e \in \ce: \|e-\tilde{e}\|_{2,P_n} \leq u} \bigl|(\nabla \Psi(R_n(e)) - \nabla \Psi(R_n(\tilde{e})))^{\top}(\hat{M}_n(e)-M_n(e))\bigr|\indic(B_n) \\
&\leq K_0\bigl(Cu+2S_n(\ce)\bigr)S_n(\ce)
\end{split}
\end{equation}
for all $u \leq \tilde{r}/3$.

Next, define 
\[
c(x;e,\tilde{e},\eta,\hat{\eta}) = \bigl[f(e(X_i),\hat{\eta}(X_i))-f(e(X_i),\eta(X_i))\bigr]-\bigl[f(\tilde{e}(X_i),\hat{\eta}(X_i))-f(\tilde{e}(X_i),\eta(X_i))\bigr]
\]
for all $x \in \cx$ so that
\[
\bigl(\hat{M}_n(e)-M_n(e)\bigr)-\bigl(\hat{M}_n(\tilde{e})-M_n(\tilde{e})\bigr) = \frac1n \sum_{i=1}^n c(X_i;e,\tilde{e},\eta,\hat{\eta}).
\]
Fix $k,\ell \in \{1,\ldots,p\}$.
By Taylor's theorem with the Lagrange form of the remainder,
\begin{align*}
|c(X_i;e,\tilde{e},\eta,\hat{\eta})| & = |[f(e(X_i),\hat{\eta}(X_i))-f(e(X_i),\eta(X_i))]-[f(\tilde{e}(X_i),\hat{\eta}(X_i))-f(\tilde{e}(X_i),\eta(X_i))]| \\
& = |f_w(e(X_i),\eta_1(X_i))^{\top}(\hat{\eta}(X_i)-\eta(X_i)) - f_w(\tilde{e}(X_i),\eta_2(X_i))^{\top}(\hat{\eta}(X_i)-\eta(X_i))| \\
& \leq \|\hat{\eta}(X_i)-\eta(X_i)\|_2\|f_w(e(X_i),\eta_1(X_i))-f_w(e(X_i),\eta(X_i))\|_2 \\
&\phe\ + \|\hat{\eta}(X_i)-\eta(X_i)\|_2 \|f_w(e(X_i),\eta(X_i))-f_w(\tilde{e}(X_i),\eta(X_i))\|_2\\
&\phe\ + \|\hat{\eta}(X_i)-\eta(X_i)\|_2\|f_w(\tilde{e}(X_i),\eta(X_i))-f_w(\tilde{e}(X_i),\eta_2(X_i))\|_2 \\
& \leq 2C\|\hat{\eta}(X_i)-\eta(X_i)\|_2^2 + C\|\hat{\eta}(X_i)-\eta(X_i)\|_2|\tilde{e}(X_i)-e(X_i)|.
\end{align*}
We have omitted subscripts $k\ell$ on $c$ and $f$ everywhere in the preceding display for brevity
(i.e., $c$ above denotes $c_{k\ell}$ and $f$ above denotes $f_{k\ell}$).
The functions
$e_1(x)$ and $e_2(x)$ are somewhere on the line segment between $e(x)$ and $\tilde{e}(x)$,
and he functions $\eta_1(x)$ and $\eta_2(x)$ are somewhere on the line segment between $\eta(x)$ and $\hat{\eta}(x)$,
and the final inequality follows from uniform boundedness of $f_{ww}$ and $f_w'$.
We conclude
\begin{align*}
\sup_{e \in \ce: \|e-\tilde{e}\|_{2,P_n} \leq u} \|(\hat{M}_n(e)-M_n(e))-(\hat{M}_n(\tilde{e})-M_n(\tilde{e}))\| & \leq 
\frac{1}{n}\sum_{i=1}^n \|c(X_i;e,\tilde{e},\eta,\hat{\eta})\|  \\
& \leq Cp^2
\bigl(2\|\hat{\eta}-\eta\|_{2,P_n}^2 + u\|\hat{\eta}-\eta\|_{2,P_n}\bigr)
\end{align*}
where the last inequality follows by Cauchy-Schwarz.
Recalling $\nabla \Psi(R_n(\tilde{e})) \preceq \tilde{K}I$ whenever $B_n$ holds,
we conclude
\begin{align*}
&   \sup_{e \in \ce: \|e-\tilde{e}\|_{2,P_n} \leq u} \nabla \Psi(R_n(\tilde{e}))^{\top}[(\hat{M}_n(e)-M_n(e))-(\hat{M}_n(\tilde{e})-M_n(\tilde{e}))]\indic(B_n)\\
&\leq \tilde{K}C
   p^2(2\|\hat{\eta}-\eta\|_{2,P_n}^2 + u\|\hat{\eta}-\eta\|_{2,P_n})
\end{align*}
for all $u \leq r/3$.
The preceding display and~\eqref{eq:peeling_sup_1} imply Lemma~\ref{lemma:centered_empirical_process} in view of the decomposition~\eqref{eq:U_n_diff}. 
\end{proof}

Continuing with the proof of Step 4,
let
$C_n=A_n \cap B_n$ and $r=r_0 \wedge \beta > 0$,
where $A_n$ and $r_0$ are as in~\eqref{eq:strong_concavity_n} with $\hat{e}=\tilde{e}$ and $B_n$ and $\beta$ are as in~\eqref{eq:centered_empirical_process}.
Fix $M > -\infty$ and an arbitrary sequence $a_n \uparrow \infty$.
For each $j>M$ with $2^ja_n\alpha_n \leq r$,
define the ``shell" $S_j = \{e \in \cf_n: 2^{j-1}a_n\alpha_n < \|e-\tilde{e}\|_{2,P_n} \leq 2^j a_n\alpha_n\}$.
It follows that for each such $j$, 
whenever $e \in S_j$ we have $r_0 \geq r \geq \|e-\tilde{e}\|_{2,P_n} \geq 2^{j-1}a_n\alpha_n$,
and so by~\eqref{eq:strong_concavity_n},
we have
\[
(\phi_n(e)-\phi_n(\tilde{e}))\indic(C_n) \leq -c_0\|e-\tilde{e}\|_{2,P_n}^2\indic(C_n) \leq -c_02^{2j-2}a_n^2\alpha_n^2\indic(C_n)
\]
for all $e \in S_j$.
Hence using the definition of $\hat e$
\begin{align*}
\indic(\hat{e} \in S_j)\indic(C_n) & \leq \indic\biggl(\,\sup_{e \in S_j}\hat{\phi}_n(e)-\hat{\phi}_n(\tilde{e}) \geq 0\biggr)\indic(C_n)
\\
& \leq \indic\biggl(\,\sup_{e:\|e-\tilde{e}\|_{2,P_n} \leq 2^ja_n\alpha_n}(\hat{\phi}_n(e)-\hat{\phi}_n(\tilde{e}))-(\phi_n(e)-\phi_n(\tilde{e})) \geq c_02^{2j-2}a_n^2\alpha_n^2\biggr)\indic(C_n) \\
& \leq \indic\bigl(2^ja_n\alpha_n \|\hat{\eta}-\eta\|_{2,P_n} + \|\hat{\eta}-\eta\|_{2,P_n}^2 \geq c_0C_0^{-1}2^{2j-2}a_n^2\alpha_n^2\bigr)\indic(C_n) \\
& \leq \frac{2^ja_n\alpha_n\|\hat{\eta}-\eta\|_{2,P_n} + \|\hat{\eta}-\eta\|_{2,P_n}^2}{c_0C_0^{-1}2^{2j-2}a_n^2\alpha_n^2}\indic(C_n) \\
& = \frac{C_0}{c_0} \biggl(\frac{\|\hat{\eta}-\eta\|_{2,P_n}}{2^{j-2}a_n\alpha_n} + \frac{\|\hat{\eta}-\eta\|_{2,P_n}^2}{2^{2j-2}a_n^2\alpha_n^2}\biggr) \indic(C_n)
\end{align*}
where the last inequality follows from Lemma~\ref{lemma:centered_empirical_process}.
Then
\begin{align*}
\indic\left(\frac{r}{2} \geq \|\hat{e}-\tilde{e}\|_{2,P_n} > 2^Ma_n\alpha_n\right)\indic(C_n) & \leq \sum_{j > M,2^ja_n\alpha_n \leq r} \indic(\hat{e} \in S_j) \indic(C_n) \\
& \leq \frac{C_0}{c_0} \sum_{j=M+1}^{\infty} \left(\frac{\|\hat{\eta}-\eta\|_{2,P_n}}{2^{j-2}a_n\alpha_n} + \frac{\|\hat{\eta}-\eta\|_{2,P_n}^2}{2^{2j-2}a_n^2\alpha_n^2}\right) \\
& = \frac{C_0}{c_0}\left(\frac{\|\hat{\eta}-\eta\|_{2,P_n}}{2^{M-2}a_n\alpha_n} + \frac{4}{3}\frac{\|\hat{\eta}-\eta\|_{2,P_n}^2}{2^{2M}a_n^2\alpha_n^2}\right).
\end{align*}
Since $\|\hat{\eta}-\eta\|_{2,P} = o_p(a_n\alpha_n)$,
we conclude
\[
\Pr(r/2 \geq \|\hat{e}-\tilde{e}\|_{2,P_n} > 2^Ma_n\alpha_n,C_n) = o(1)
\]
for each $M$.
Now, by step 3
we know that $\|\hat{e}-e^*\|_{2,P_n} = o_p(1)$.
Applying step 3 again but with $\hat{\eta}=\eta$ shows $\|\tilde{e}-e^*\|_{2,P_n} = O_p(n^{-1/4}) = o_p(1)$.
Thus
\[
\|\hat{e}-\tilde{e}\|_{2,P_n} \leq \|\hat{e}-e^*\|_{2,P_n} + \|e^*-\tilde{e}\|_{2,P_n} = o_p(1)
\]
so that $\Pr(\|\hat{e}-\tilde{e}\|_{2,P_n} > r/2) = o(1)$.
Since $\Pr(C_n) \rightarrow 1$ as $n \rightarrow \infty$,
we can conclude that 
\[
\Pr(\|\hat{e}-\tilde{e}\|_{2,P_n} > 2^Ma_n\alpha_n) = o(1).
\]
With $a_n \uparrow \infty$ arbitrary,
by the preceding display and Lemma~\ref{lemma:oh_pee} we have $\|\hat{e}-\tilde{e}\|_{2,P_n} = O_p(\alpha_n)$ and
\[
\|\hat{e}-e^*\|_{2,P_n} \leq \|\hat{e}-\tilde{e}\|_{2,P_n} + \|\tilde{e}-e^*\|_{2,P_n} = O_p(n^{-1/4}) + O_p(\alpha_n)
\]
as desired.

It remains to show the same convergence rate holds out of sample,
i.e., $\|\hat{e}-e^*\|_{2,P} = O_p(n^{-1/4})+O_p(\alpha_n)$.
We do this by showing that $\big| \|\hat{e}-e^*\|_{2,P_n} - \|\hat{e}-e^*\|_{2,P} \big| = O_p(n^{-1/4})$.
With $\ce_2^-$ as defined in Lemma~\ref{lemma:sq_bracket} in terms of the collection $\ce$ of Assumption~\ref{assump:info_matrix},
we know that 
\begin{align*}
\big| \|\hat{e}-e^*\|_{2,P_n}^2 - \|\hat{e}-e^*\|_{2,P}^2 \big| & = \Big|\frac1n\sum_{i=1}^n [(\hat{e}(X_i)-e^*(X_i))^2 - \int_{\cx} (\hat{e}(x)-e^*(x))^2\, \mrd P(x)] \Big| \\
& \leq \sup_{e \in \ce_2^-} |(P_n-P)e|.
\end{align*}
Furthermore,
by Lemma~\ref{lemma:empirical_process} and Lemma~\ref{lemma:sq_bracket} we know that for some $K_0<\infty$ we have
\begin{align*}
\sup_{e \in \ce_2^-} |(P_n-P)f| & \leq K_0 n^{-1/2} \int_0^1 \sqrt{\log \cn(\epsilon,\ce_2^-,L^2(P_n))} \,\mrd\epsilon \\
& \leq K_0 \sqrt{2} n^{-1/2} \int_0^1 \sqrt{\log\cn(\epsilon/4,\ce,L^2(P_n))} \,\mrd\epsilon \\
& = 4K_0\sqrt{2} n^{-1/2} \int_0^{1/4}\sqrt{\log \cn(\delta,\ce,L^2(P_n))} \,\mrd\delta.
\end{align*}
Then by~\eqref{eq:finite_entropy} we conclude $\sup_{e \in \ce_2^-} |(P_n-P)e| = O_p(n^{-1/2})$.
Since 
\[
|a-b| \leq a+b \implies \sqrt{|a-b|} \leq \sqrt{a+b} \implies |a-b| \leq \sqrt{a+b}\sqrt{|a-b|} = \sqrt{|a^2-b^2|}
\]
for any $a,b \geq 0$,
we have
\[
\big| \|\hat{e}-e^*\|_{2,P_n} - \|\hat{e}-e^*\|_{2,P} \big| = O_p(n^{-1/4})
\]
as desired.
\end{proof}

\subsection{Proof of Theorem~\ref{thm:concave_maximization}}
\label{proof:thm:concave_maximization}
Here we prove convergence of Algorithm~\ref{alg:csbae},
our concave maximization procedure for designing an optimal CSBAE.
We begin by proving we can design for $\hat{\theta}_{\aipw}$, as stated in the first numbered condition of
Theorem~\ref{thm:concave_maximization}.
The proof proceeds by showing that the objective defining $e_{t,\aipw}^*(\cdot)$ in~\eqref{eq:e_t_star_aipw} can be written in a form so that Assumption~\ref{assump:info_matrix} and~\eqref{eq:eta_n} are satisfied,
the latter with $\alpha_N=N^{-1/4}$.
Then we conclude by applying Lemma~\ref{lemma:concave_maximization}.

Many of our expressions will include the cumulative sum of batch frequencies $\sum_{u=1}^t\kappa_u$. We use $\kappa_{1:t}$ to denote this quantity below. 
We similarly abbreviate $\sum_{u=1}^tN_u$ to $N_{1:t}$.

With $V_{0:t,\aipw}$ scalar,
the information function $\Psi=\Psi(\cdot)$ is simply an increasing scalar-valued function by Assumption~\ref{assump:Psi},
and hence we have
\begin{align*}
e_{t,\aipw}^*(\cdot) = \argmax_{e_t(\cdot) \in \cf_{*,t}} (V_{0:t,\aipw})^{-1} & = \argmin_{e_t(\cdot) \in \cf_{*,t}} V_{0:t,\aipw} \\
& = \argmax_{e_t(\cdot) \in \cf_{*,t}} h(V_{0:t,\aipw}) \\
& = \argmax_{e_t(\cdot) \in \cf_{t,*}} \e_{P^X}\left[\frac{2C}{\gamma_0}-\frac{v_0(1,X)}{e_0^{(t)}(X)} - \frac{v_0(0,X)}{1-e_0^{(t)}(X)}\right]
\end{align*}
where in the final equality we dropped the additive term $\e[(\tau_0(X)-\theta_0)^2]$ which is independent of $e_t(\cdot)$,
and defined $h(x) = 2C/\gamma_0-x$ for
\[
\gamma_0 := \frac{\epsilon_1}{2}\frac1{\kappa_{1:t}}\min(\kappa_1,\kappa_t)  > 0.
\]
Evidently $h(\cdot)$ is decreasing.
We can now define the information matrix
\[
\ci = \ci(e_t,\eta_0) = \e_{P^X}\left[\frac{2C}{\gamma_0}-\frac{v_0(1,X)}{e_0^{(t)}(X)} - \frac{v_0(0,X)}{1-e_0^{(t)}(X)}\right]
\]
which is of the form~\eqref{eq:generic_info} with
\[
f(e,w) = \frac{2C}{\gamma_0} - \frac{w_1}{1-w_3-w_4e} - \frac{w_2}{w_3+w_4e}
\]
and
\[
\eta(x) = \biggl(v_0(0,x),\, v_0(1,x),\,\frac1{\kappa_{1:t}} \sum_{u=1}^{t-1} \kappa_u e_u(x),\, \frac1{\kappa_{1:t}}\kappa_t\biggr).
\]
Let $\cw=[c,C]^2 \times \cw_+$ where $\cw_+=\{(x,y) \in \real^2 \mid x \geq \gamma_0, y \geq \gamma_0, x+y \leq 1-\gamma_0\}$.
By~\eqref{eq:prop_asymp_limit} and the assumptions of the Theorem
(specifically the uniform bounds on the variance functions and the assumption that $\epsilon_1 \leq e_1(x)=\hat{e}_1^{(k)}(x) \leq 1-\epsilon_1$ for all $x \in \cx$ and folds $k=1,\ldots,K$),
we can verify that
$\eta(x) \in \cw$ and $\hat{\eta}^{(k)}(x) \in \cw$ for all $x \in \cx$, folds $k=1,\ldots,K$, and sufficiently large $N$,
where
\[
\hat{\eta}^{(k)}(x) = \biggl(\hat{v}^{(k)}(0,x),\,\hat{v}^{(k)}(1,x),\,\frac1{N_{1:t}} \sum_{u=1}^{t-1} N_u \hat{e}_u^{(k)}(x),\, \frac{N_t}{N_{1:t}}\biggr).
\]
Also we have
\[
1-\gamma_0 \geq w_3+w_4e \geq \gamma_0 \quad \text{and} \quad f(e,w) \geq 0, \quad \forall (e,w) \in [0,1] \times \cw.
\]
Evidently $\cw$ is closed and bounded, hence compact.
With $w_3+w_4e$ linear in $e$,
there exists $\delta < 0$ such that for all $(e,w)$ in a neighborhood containing $(\delta,1-\delta) \times \cw$,
$w_3+w_4e$ is uniformly bounded away from 0
and then $f(e,w)$ evidently has continuous second partial derivatives on this neighborhood.
Finally, we compute
\[
-f''(e,w) = \frac{2w_2w_4^2}{(w_3+w_4e)^3} + \frac{2w_1w_4^2}{(1-w_3-w_4e)^3} \geq \frac{2c\gamma_0^2}{(1-\gamma_0)^3} > 0, \quad \forall (e,w) \in [0,1] \times \cw
\]
which shows that all conditions of Assumption~\ref{assump:info_matrix} have been satisfied.
Equation~\eqref{eq:eta_n} holds with $\alpha_N=N^{-1/4}$ by~\eqref{eq:prop_asymp_conv} and~\eqref{eq:csbae_limit},
completing the proof of the first numbered condition of Theorem~\ref{thm:concave_maximization},
pertaining to design for $\hat{\theta}_{\aipw}$.

It remains to show the second numbered condition holds.
As above,
the proof proceeds by showing the objective for $\hat{e}_{t,\epl}^*(\cdot)$ in~\eqref{eq:e_t_star_aipw} can be written in a form so that Assumption~\ref{assump:info_matrix} and~\eqref{eq:eta_n} are satisfied,
the latter with $\alpha_N=N^{-1/4}$.
To that end,
we take $\ci=\ci(e_t,\eta_0) =V_{0:t,\epl}^{-1}$ which takes the form~\eqref{eq:generic_info} with
\[
f(e,w) = f(e,w_1,w_2,w_3,w_4,w_5) = \frac{(w_3+w_4e)(1-w_3-w_4e)}{w_1(w_3+w_4e) + w_2(1-w_3-w_4e)}w_5w_5^{\top}
\]
and 
\[
\eta(x) = \biggl(v_0(0,x),\, v_0(1,x),\, \frac1{\kappa_{1:t}} \sum_{u=1}^{t-1} \kappa_u e_u(x),\, \frac{\kappa_t}{\kappa_{1:t}},\,\psi(x)\biggr).
\]
Let $\cw = [c,C]^2 \times \cw_+ \times [-C,C]^p \subseteq \real^{4+p}$,
where once again $\cw_+=\{(x,y) \in \real^2 \mid x \geq \gamma_0, y \geq \gamma_0, x+y \leq 1-\gamma_0\}$
for
\[
\gamma_0 := \frac{\epsilon_1}{2\kappa_{1:t}}\min(\kappa_1,\kappa_t)  > 0.
\]
Evidently $\cw$ is closed and bounded, hence compact.
Then the assumptions of the Theorem (specifically the uniform bounds on the variance functions and the assumption that $\epsilon_1 \leq e_1(x)=\hat{e}_1^{(k)}(x) \leq 1-\epsilon_1$ for all $x \in \cx$ and folds $k=1,\ldots,K$) ensure that for each fold $k=1,\ldots,K$,
if we take 
\[
\hat{\eta}^{(k)}(x)=\biggl(\hat{v}^{(k)}(0,x),\,\hat{v}^{(k)}(1,x),\,\frac1{N_{1:t}} \sum_{u=1}^{t-1} N_u \hat{e}_u^{(k)}(x),\, \frac{N_t}{N_{1:t}},\, \psi(x)\biggr)
\]
then $\eta(x) \in \cw$ and $\hat{\eta}^{(k)}(x) \in \cw$ for all $x \in \cx$ whenever $N$ is sufficiently large.
We note that for each $e \in [0,1]$,
$$
f(e,w) = \frac{(w_3+w_4e)(1-w_3-w_4e)}{w_1(w_3+w_4e) + w_2(1-w_3-w_4e)}w_5w_5^{\top}
$$
is positive semidefinite because the lead constant above is nonnegative.
Furthermore,
note that the denominator $w_1(w_3+w_4e)+w_2(1-w_3-w_4e)$ is bounded below by $c$ for any $(e,w) \in [0,1] \times \cw$ and continuous on $(e,w_1,w_2,w_3,w_4) \in \real^5$.
This denominator is linear in $e$ (for fixed $w$) and so additionally
there exists $\delta < 0$ such that on some open neighborhood containing $(\delta,1-\delta) \times \cw$,
this denominator is strictly positive.
Therefore $f(e,w)$ has two continuous partial derivatives with respect to $e$.
Finally, we compute
\[
-f''(e,w) = \frac{2w_4^2(w_3+w_4e)(1-w_3-w_4e)}{(w_1(w_3+w_4e)+w_2(1-w_3-w_4e))^3}w_5w_5^{\top}. 
\]
This is positive semidefinite since as above,
$1-\gamma_0 \geq w_3+w_4e \geq \gamma_0$ for all $(e,w) \in [0,1] \times \cw$,
so
$$\frac{2w_4^2(w_3+w_4e)(1-w_3-w_4e)}{(w_1(w_3+w_4e)+w_2(1-w_3-w_4e))^3}>0$$ 
on $[0,1] \times \cw$.
Furthermore, as all diagonal entries of $w_5w_5^{\top}$ are nonnegative,
the inclusion of an intercept in $\psi(x)$ ensures that
\begin{align*}
\inf_{(e,w) \in [0,1] \times \cw} \tr(-f''(e,w)) &\geq \inf_{(e,w) \in [0,1] \times \cw}\, \frac{2w_4^2(w_3+w_4e)(1-w_3-w_4e)}{[w_1(w_3+w_4e)+w_2(1-w_3-w_4e)]^3} \\
& \geq \frac{2\gamma_0^3(1-\gamma_0)}{C^3}.
\end{align*}
As before,
equation~\eqref{eq:eta_n} holds with $\alpha_N=N^{-1/4}$ by~\eqref{eq:prop_asymp_conv} and~\eqref{eq:csbae_limit},
enabling us to apply Lemma~\ref{lemma:concave_maximization} and
completing the proof of the Theorem.

\section{Additional simulations}
\label{app:simulations}
We present results from some additional numerical simulations in the framework of Section~\ref{sec:simulations}.
\subsection{Unequal budget constraints}
\label{app:simulations_unequal}
Tables~\ref{table:ate_unequal_sim} and~\ref{table:epl_unequal_sim} reproduce Tables~\ref{table:ate_sim} and~\ref{table:epl_sim} using simulations with budget constraints $m_{L,2}=m_{H,2}=0.4$.
The results are qualitatively similar to those in the main text.
One notable difference is that there seem to be some additional gains to pooling in ATE estimation,
both asymptotically and in finite samples.
For example, in the homoskedastic DGPs,
we see about a 5\% asymptotic efficiency gain from pooling in Table~\ref{table:ate_unequal_sim} when either $d=1$ or $d=10$,
which translates well to finite sample gains,
particularly for the $d=10$ DGP.
By contrast there is no asymptotic gain for the homoskedastic DGPs in Table~\ref{table:ate_sim}.

\begin{table}[!htb]
\centering
\caption{Same as Table~\ref{table:ate_sim}, but for simulations with $m_{L,2}=m_{H,2}=0.4$ \\
}
\label{table:ate_unequal_sim}
\begin{tabular}{ccccc}
\toprule
DGP & Estimator & Design &  Sim. rel. eff. (90\% CI) & Asymp. rel. eff. \\
\midrule
\multirow{5}{0.15\linewidth}{\centering $d=1$, Homoskedastic} & $\hat{\theta}_{\aipw}$ & Flexible & 1.051 (1.015, 1.088) & 1.050 \\
& $\hat{\theta}_{\aipw}$ & Binned & 1.034 (0.995, 1.073) & 1.050 \\
& $\hat{\theta}_{\aipw}$ & Simple RCT & 1.050 (1.024, 1.077) & 1.050 \\
& $\hat{\theta}_{\aipw}^{(\la)}$ & Flexible & 0.993 (0.977, 1.008) & 1.000 \\
& $\hat{\theta}_{\aipw}^{(\text{bin})}$ & Binned & 0.916 (0.871, 0.962) & 0.919 \\
\midrule
\multirow{5}{0.15\linewidth}{\centering $d=1$, Heteroskedastic} & $\hat{\theta}_{\aipw}$ & Flexible & 1.024 (0.978, 1.071) & 1.073 \\
& $\hat{\theta}_{\aipw}$ & Binned & 1.099 (1.046, 1.154) & 1.067 \\
& $\hat{\theta}_{\aipw}$ & Simple RCT & 1.043 (1.019, 1.067) & 1.034 \\
& $\hat{\theta}_{\aipw}^{(\la)}$ & Flexible & 1.032 (1.003, 1.063) & 1.012 \\
& $\hat{\theta}_{\aipw}^{(\text{bin})}$ & Binned & 1.016 (0.961, 1.071) & 0.993 \\
\midrule
\multirow{5}{0.15\linewidth}{\centering $d=10$, Homoskedastic} & $\hat{\theta}_{\aipw}$ & Flexible & 1.084 (1.031, 1.139) & 1.050 \\
& $\hat{\theta}_{\aipw}$ & Binned & 1.081 (1.030, 1.133) & 1.029 \\
& $\hat{\theta}_{\aipw}$ & Simple RCT & 1.128 (1.096, 1.161) & 1.050 \\
& $\hat{\theta}_{\aipw}^{(\la)}$ & Flexible & 0.994 (0.970, 1.018) & 1.000 \\
& $\hat{\theta}_{\aipw}^{(\text{bin})}$ & Binned & 0.520 (0.478, 0.565) & 0.448 \\
\midrule
\multirow{5}{0.15\linewidth}{\centering $d=10$, Heteroskedastic} & $\hat{\theta}_{\aipw}$ & Flexible & 1.051 (0.999, 1.107) & 1.064 \\
& $\hat{\theta}_{\aipw}$ & Binned & 1.001 (0.952, 1.051) & 1.045 \\
& $\hat{\theta}_{\aipw}$ & Simple RCT & 1.062 (1.031, 1.094) & 1.035 \\
& $\hat{\theta}_{\aipw}^{(\la)}$ & Flexible & 0.981 (0.948,1.015) & 1.012 \\
& $\hat{\theta}_{\aipw}^{(\text{bin})}$ & Binned & 0.626 (0.582, 0.671) & 0.611 \\
\bottomrule
\end{tabular}
\end{table}

\begin{table}[!htb]
\centering
\caption{Same as Table~\ref{table:epl_sim}, but for simulations with $m_{L,2}=m_{H,2}=0.4$ \\
}
\label{table:epl_unequal_sim}
\begin{tabular}{ccccc}
\toprule
DGP & Estimator & Design &  Sim. rel. eff. (90\% CI) & Asymp. rel. eff. \\
\midrule
\multirow{4}{0.15\linewidth}{\centering $d=1$, Homoskedastic} & $\hat{\theta}_{\epl}$ & Flexible & 1.118 (1.076, 1.162) & 1.092 \\
& $\hat{\theta}_{\epl}$ & Binned & 1.084 (1.038, 1.131) & 1.051 \\
& $\hat{\theta}_{\epl}$ & Simple RCT & 1.099 (1.072, 1.126) & 1.050 \\
& $\hat{\theta}_{\epl}^{(\la)}$ & Flexible & 0.995 (0.976, 1.015) & 1.006\\
\midrule
\multirow{4}{0.15\linewidth}{\centering $d=1$, Heteroskedastic} & $\hat{\theta}_{\epl}$ & Flexible & 1.139 (1.035, 1.249) & 1.090 \\
& $\hat{\theta}_{\epl}$ & Binned & 1.050 (0.905, 1.188) & 1.013 \\
& $\hat{\theta}_{\epl}$ & Simple RCT & 1.074 (0.980, 1.172) & 1.034 \\
& $\hat{\theta}_{\epl}^{(\la)}$ & Flexible & 1.028 (0.970, 1.090) & 1.017 \\
\midrule
\multirow{4}{0.15\linewidth}{\centering $d=10$, Homoskedastic} & $\hat{\theta}_{\epl}$ & Flexible & 1.234 (1.203, 1.266) & 1.056 \\
& $\hat{\theta}_{\epl}$ & Binned & 1.204 (1.172, 1.237) & 0.994 \\
& $\hat{\theta}_{\epl}$ & Simple RCT & 1.262 (1.234, 1.292) & 1.050 \\
& $\hat{\theta}_{\epl}^{(\la)}$ & Flexible & 0.997 (0.982, 1.013) & 1.002 \\
\midrule
\multirow{4}{0.15\linewidth}{\centering $d=10$, Heteroskedastic} & $\hat{\theta}_{\epl}$ & Flexible & 1.104 (1.065, 1.144) & 1.069 \\
& $\hat{\theta}_{\epl}$ & Binned & 1.055 (1.016, 1.094) & 1.043 \\
& $\hat{\theta}_{\epl}$ & Simple RCT & 1.078 (1.042, 1.116) & 1.035 \\
& $\hat{\theta}_{\epl}^{(\la)}$ & Flexible & 0.993 (0.967,1.019) & 1.013 \\
\bottomrule
\end{tabular}
\end{table}

\subsection{Perfect nuisance estimation}
\label{app:oracle_sim}
To better isolate the performance effects of our specific choices of nuisance estimation methods in the numerical study of Section~\ref{sec:simulations},
in Tables~\ref{table:ate_oracle_sim} and~\ref{table:epl_oracle_sim}
we reproduce Tables~\ref{table:ate_sim} and~\ref{table:epl_sim},
respectively,
but assume all nuisance functions are 
known exactly at both the design and estimation stages.
For ATE estimation (Table~\ref{table:ate_oracle_sim}),
the flexible designs are aware of a perfectly constant variance function in the homoskedastic DGPs,
which induces them to always learn the (optimal) simple RCT in every simulation.
However, for the binned designs in the homoskedastic DGPs and both the binned and flexible designs in the heteroskedastic DGPs,
there is some cross-simulation variability in the propensity learned.
This stems from variation in the parts of the variance function being sampled due to variation in the covariates across simulations.
Consequently, the simulated finite sample efficiency gain ends up being somewhat lower than the asymptotic gain.
This suggests that the finite sample efficiency gains from pooling observed in Table~\ref{table:ate_sim} are due to improved use of nuisance function estimates by the pooled estimator $\hat{\theta}_{\aipw}$.
One reason we might expect this is that the pooled estimator uses nuisance estimates from observations pooled across both batches of the experiment,
while each component of the linearly aggregated estimator only uses nuisance estimates from a single batch.

\begin{table}[!htb]
\centering
\caption{Same as Table~\ref{table:ate_sim}, but for oracle simulations that assume knowledge of the true nuisance functions \\
}
\label{table:ate_oracle_sim}
\begin{tabular}{ccccc}
\toprule
DGP & Estimator & Design &  Sim. rel. eff. (90\% CI) & Asymp. rel. eff. \\
\midrule
\multirow{5}{0.15\linewidth}{\centering $d=1$, Homoskedastic} & $\hat{\theta}_{\aipw}$ & Flexible & 1.000 (1.000, 1.000) & 1.000 \\
& $\hat{\theta}_{\aipw}$ & Binned & 0.978 (0.960, 0.997) & 0.999 \\
& $\hat{\theta}_{\aipw}$ & Simple RCT & 1.000 (1.000, 1.000) & 1.000 \\
& $\hat{\theta}_{\aipw}^{(\la)}$ & Flexible & 1.000 (1.000, 1.000) & 1.000 \\
& $\hat{\theta}_{\aipw}^{(\text{bin})}$ & Binned & 0.828 (0.794, 0.863) & 0.875 \\
\midrule
\multirow{5}{0.15\linewidth}{\centering $d=1$, Heteroskedastic} & $\hat{\theta}_{\aipw}$ & Flexible & 1.005 (0.960, 1.051) & 1.051 \\
& $\hat{\theta}_{\aipw}$ & Binned & 0.977 (0.935, 1.020) & 1.044 \\
& $\hat{\theta}_{\aipw}$ & Simple RCT & 1.000 (1.000, 1.000) & 1.000 \\
& $\hat{\theta}_{\aipw}^{(\la)}$ & Flexible & 1.008 (0.977, 1.039) & 1.026 \\
& $\hat{\theta}_{\aipw}^{(\text{bin})}$ & Binned & 0.905 (0.859, 0.952) & 0.976 \\
\midrule
\multirow{5}{0.15\linewidth}{\centering $d=10$, Homoskedastic} & $\hat{\theta}_{\aipw}$ & Flexible & 1.000 (1.000, 1.000) & 1.000 \\
& $\hat{\theta}_{\aipw}$ & Binned & 0.968 (0.923, 1.013) & 0.965 \\
& $\hat{\theta}_{\aipw}$ & Simple RCT & 1.000 (1.000, 1.000) & 1.000 \\
& $\hat{\theta}_{\aipw}^{(\la)}$ & Flexible & 1.000 (1.000,1.000) & 1.000 \\
& $\hat{\theta}_{\aipw}^{(\text{bin})}$ & Binned & 0.402 (0.368, 0.439) & 0.432 \\
\midrule
\multirow{5}{0.15\linewidth}{\centering $d=10$, Heteroskedastic} & $\hat{\theta}_{\aipw}$ & Flexible & 1.023 (0.986, 1.060) & 1.043 \\
& $\hat{\theta}_{\aipw}$ & Binned & 1.010 (0.967, 1.054) & 1.019 \\
& $\hat{\theta}_{\aipw}$ & Simple RCT & 1.000 (1.000, 1.000) & 1.000 \\
& $\hat{\theta}_{\aipw}^{(\la)}$ & Flexible & 1.020 (0.990,1.051) & 1.026 \\
& $\hat{\theta}_{\aipw}^{(\text{bin})}$ & Binned & 0.618 (0.575, 0.663) & 0.612 \\
\bottomrule
\end{tabular}
\end{table}

\begin{table}[t!]
\centering
\caption{Same as Table~\ref{table:epl_sim}, but for oracle simulations that assume knowledge of the true nuisance functions
}
\label{table:epl_oracle_sim}
\begin{tabular}{ccccc}
\toprule
DGP & Estimator & Design &  Sim. rel. eff. (90\% CI) & Asymp. rel. eff. \\
\midrule
\multirow{4}{0.15\linewidth}{\centering $d=1$, Homoskedastic} & $\hat{\theta}_{\epl}$ & Flexible & 1.129 (1.080, 1.179) & 1.100 \\
& $\hat{\theta}_{\epl}$ & Binned & 1.026 (0.984, 1.069) & 1.021 \\
& $\hat{\theta}_{\epl}$ & Simple RCT & 1.009 (1.003, 1.016) & 1.000 \\
& $\hat{\theta}_{\epl}^{(\la)}$ & Flexible & 1.073 (1.041, 1.107) & 1.056\\
\midrule
\multirow{4}{0.15\linewidth}{\centering $d=1$, Heteroskedastic} & $\hat{\theta}_{\epl}$ & Flexible & 1.126 (1.070, 1.184) & 1.128 \\
& $\hat{\theta}_{\epl}$ & Binned & 0.966 (0.916, 1.017) & 0.969 \\
& $\hat{\theta}_{\epl}$ & Simple RCT & 1.008 (1.001, 1.015) & 1.000 \\
& $\hat{\theta}_{\epl}^{(\la)}$ & Flexible & 1.082 (1.043, 1.123) & 1.079 \\
\midrule
\multirow{4}{0.15\linewidth}{\centering $d=10$, Homoskedastic} & $\hat{\theta}_{\epl}$ & Flexible & 1.079 (1.068, 1.090) & 1.025 \\
& $\hat{\theta}_{\epl}$ & Binned & 0.997 (0.979, 1.015) & 0.966 \\
& $\hat{\theta}_{\epl}$ & Simple RCT & 1.047 (1.041, 1.053) & 1.000 \\
& $\hat{\theta}_{\epl}^{(\la)}$ & Flexible & 1.021 (1.013, 1.028) & 1.019 \\
\midrule
\multirow{4}{0.15\linewidth}{\centering $d=10$, Heteroskedastic} & $\hat{\theta}_{\epl}$ & Flexible & 1.132 (1.115, 1.148) & 1.075 \\
& $\hat{\theta}_{\epl}$ & Binned & 1.082 (1.058, 1.106) & 1.054 \\
& $\hat{\theta}_{\epl}$ & Simple RCT & 1.059 (1.052, 1.066) & 1.000 \\
& $\hat{\theta}_{\epl}^{(\la)}$ & Flexible & 1.055 (1.041,1.068) & 1.062 \\
\bottomrule
\end{tabular}
\end{table}

In Table~\ref{table:epl_oracle_sim}, however,
we still see some finite sample efficiency gains from pooling,
though the effect is not as large as in Table~\ref{table:epl_sim} in the main text.
We attribute this to the fact that for estimating $\theta_{0,\pl}$,
pooling allows the asymptotic variance to be approached more quickly as a function of the total sample size $N$.
Such an effect does not show up in ATE estimation with AIPW,
since for $\hat{\theta}_{\aipw}^*$ and $\hat{\theta}_{\epl}^*$ the oracle estimators of Section~\ref{sec:pooled_estimation},
we can see that $N\Var(\hat{\theta}_{\aipw}^*)=V_{0,\aipw}$ exactly for all $N$ while $N\Var(\hat{\theta}_{\epl}^*)$ only approaches $V_{0,\epl}$ asymptotically as $N \rightarrow \infty$.
So letting $A_N^*$ be the (finite sample) AMSE of $\hat{\theta}_{\epl}^*$ computed on a sample of size $N$,
we'd expect $NA_N^* > 2NA_{2N}^*$.
Then averaging two independent copies of $\hat{\theta}_{\epl}^*$ on $N$ observations yields an estimator with AMSE $A_N^*/2$,
while pooling would yield an estimator with AMSE $A_{2N^*}$.

\clearpage
\bibliographystyle{apalike}
\bibliography{semiparametric}

\begin{thebibliography}{}

\bibitem[Atkinson et~al., 2007]{atkinson2007optimum}
Atkinson, A., Donev, A., and Tobias, R. (2007).
\newblock {\em Optimum experimental designs, with SAS}, volume~34.
\newblock Oxford University Press, Oxford.

\bibitem[Barbu and Precupanu, 2012]{barbu2012convexity}
Barbu, V. and Precupanu, T. (2012).
\newblock {\em Convexity and optimization in Banach spaces}.
\newblock Springer Science \& Business Media, Dordrecht, fourth edition.

\bibitem[Blackwell et~al., 2022]{blackwell2022batch}
Blackwell, M., Pashley, N.~E., and Valentino, D. (2022).
\newblock Batch adaptive designs to improve efficiency in social science
  experiments.
\newblock Technical report, Harvard University.
\newblock \url{https://www.mattblackwell.org}.

\bibitem[Boyd and Vandenberghe, 2004]{boyd2004convex}
Boyd, S. and Vandenberghe, L. (2004).
\newblock {\em Convex optimization}.
\newblock Cambridge University Press, Cambridge.

\bibitem[Chamberlain, 1992]{chamberlain1992efficiency}
Chamberlain, G. (1992).
\newblock Efficiency bounds for semiparametric regression.
\newblock {\em Econometrica: Journal of the Econometric Society},
  60(3):567--596.

\bibitem[Che and Namkoong, 2023]{che2023adaptive}
Che, E. and Namkoong, H. (2023).
\newblock Adaptive experimentation at scale: Bayesian algorithms for flexible
  batches.
\newblock Technical report, arXiv:2303.11582.

\bibitem[Chernozhukov et~al., 2018]{chernozhukov2018double}
Chernozhukov, V., Chetverikov, D., Demirer, M., Duflo, E., Hansen, C., Newey,
  W., and Robins, J. (2018).
\newblock {Double/debiased machine learning for treatment and structural
  parameters}.
\newblock {\em The Econometrics Journal}, 21(1):C1--C68.

\bibitem[Cytrynbaum, 2021]{cytrynbaum2021designing}
Cytrynbaum, M. (2021).
\newblock Designing representative and balanced experiments by local
  randomization.
\newblock Technical report, arXiv:2111.08157.

\bibitem[Dai et~al., 2023]{dai2023clip}
Dai, J., Gradu, P., and Harshaw, C. (2023).
\newblock Clip-{OGD}: An experimental design for adaptive {Neyman} allocation
  in sequential experiments.
\newblock Technical report, arXiv:2305.17187.

\bibitem[Dudley, 1967]{dudley1967sizes}
Dudley, R.~M. (1967).
\newblock The sizes of compact subsets of {Hilbert} space and continuity of
  {Gaussian} processes.
\newblock {\em Journal of Functional Analysis}, 1(3):290--330.

\bibitem[Fu et~al., 2017]{fu2017cvxr}
Fu, A., Narasimhan, B., and Boyd, S. (2017).
\newblock {CVXR}: An {R} package for disciplined convex optimization.
\newblock Technical report, arXiv:1711.07582.

\bibitem[Gagnon-Bartsch et~al., 2023]{gagnon2023precise}
Gagnon-Bartsch, J., Sales, A., Wu, E., Botelho, A., Erickson, J., Miratrix, L.,
  and Heffernan, N. (2023).
\newblock Precise unbiased estimation in randomized experiments using auxiliary
  observational data.
\newblock {\em Journal of Causal Inference}, 11.

\bibitem[Hadad et~al., 2021]{hadad2021confidence}
Hadad, V., Hirshberg, D.~A., Zhan, R., Wager, S., and Athey, S. (2021).
\newblock Confidence intervals for policy evaluation in adaptive experiments.
\newblock {\em Proceedings of the National Academy of Sciences},
  118(15):e2014602118.

\bibitem[Hahn et~al., 2011]{hahn2011adaptive}
Hahn, J., Hirano, K., and Karlan, D. (2011).
\newblock Adaptive experimental design using the propensity score.
\newblock {\em Journal of Business \& Economic Statistics}, 29(1):96--108.

\bibitem[Hao et~al., 2020]{hao2020adaptive}
Hao, B., Lattimore, T., and Szepesvari, C. (2020).
\newblock Adaptive exploration in linear contextual bandit.
\newblock In {\em International Conference on Artificial Intelligence and
  Statistics}, pages 3536--3545. PMLR.

\bibitem[Hirano et~al., 2003]{hirano2003efficient}
Hirano, K., Imbens, G.~W., and Ridder, G. (2003).
\newblock Efficient estimation of average treatment effects using the estimated
  propensity score.
\newblock {\em Econometrica}, 71(4):1161--1189.

\bibitem[Kasy and Sautmann, 2021]{kasy2021adaptive}
Kasy, M. and Sautmann, A. (2021).
\newblock Adaptive treatment assignment in experiments for policy choice.
\newblock {\em Econometrica}, 89(1):113--132.

\bibitem[Kato et~al., 2020]{kato2020efficient}
Kato, M., Ishihara, T., Honda, J., and Narita, Y. (2020).
\newblock Efficient adaptive experimental design for average treatment effect
  estimation.
\newblock {\em arXiv:2002.05308}.

\bibitem[Kitagawa and Tetenov, 2018]{kitagawa2018should}
Kitagawa, T. and Tetenov, A. (2018).
\newblock Who should be treated{? E}mpirical welfare maximization methods for
  treatment choice.
\newblock {\em Econometrica}, 86(2):591--616.

\bibitem[Kluger and Owen, 2023]{kluger2023kernel}
Kluger, D.~M. and Owen, A.~B. (2023).
\newblock Kernel regression analysis of tie-breaker designs.
\newblock {\em Electronic Journal of Statistics}, 17(1):243--290.

\bibitem[Kosorok, 2008]{kosorok2008introduction}
Kosorok, M.~R. (2008).
\newblock {\em Introduction to empirical processes and semiparametric
  inference.}
\newblock Springer.

\bibitem[Li and Owen, 2023]{li2023general}
Li, H.~H. and Owen, A.~B. (2023).
\newblock A general characterization of optimal tie-breaker designs.
\newblock {\em Annals of Statistics}, 51(3):1030--1057.

\bibitem[Ma et~al., 2006]{ma2006efficient}
Ma, Y., Chiou, J.-M., and Wang, N. (2006).
\newblock Efficient semiparametric estimator for heteroscedastic partially
  linear models.
\newblock {\em Biometrika}, 93(1):75--84.

\bibitem[Montanari and Saeed, 2022]{montanari2022universality}
Montanari, A. and Saeed, B.~N. (2022).
\newblock Universality of empirical risk minimization.
\newblock In {\em Conference on Learning Theory}, pages 4310--4312. PMLR.

\bibitem[Morrison and Owen, 2022]{morrison2022optimality}
Morrison, T.~P. and Owen, A.~B. (2022).
\newblock Optimality in multivariate tie-breaker designs.
\newblock Technical report, arXiv:2202.10030.

\bibitem[{MOSEK ApS}, 2022]{mosek}
{MOSEK ApS} (2022).
\newblock {\em MOSEK Rmosek package 10.1.9}.

\bibitem[Neyman, 1934]{neyman1934}
Neyman, J. (1934).
\newblock On the {Two} {Different} {Aspects} of the {Representative} {Method}:
  {The} {Method} of {Stratified} {Sampling} and the {Method} of {Purposive}
  {Selection}.
\newblock {\em Journal of the Royal Statistical Society}, 97(4):558.

\bibitem[Owen and Varian, 2020]{owen2020optimizing}
Owen, A.~B. and Varian, H. (2020).
\newblock Optimizing the tie-breaker regression discontinuity design.
\newblock {\em Electronic Journal of Statistics}, 14:4004--4027.

\bibitem[Pukelsheim, 2006]{pukelsheim2006optimal}
Pukelsheim, F. (2006).
\newblock {\em Optimal design of experiments}.
\newblock SIAM.

\bibitem[Robinson, 1988]{robinson1988root}
Robinson, P.~M. (1988).
\newblock Root-n-consistent semiparametric regression.
\newblock {\em Econometrica: Journal of the Econometric Society}, pages
  931--954.

\bibitem[Rosenbaum and Rubin, 1983]{rosenbaum1983central}
Rosenbaum, P.~R. and Rubin, D.~B. (1983).
\newblock The central role of the propensity score in observational studies for
  causal effects.
\newblock {\em Biometrika}, 70(1):41--55.

\bibitem[Rosenman and Owen, 2021]{rosenman2021designing}
Rosenman, E.~T. and Owen, A.~B. (2021).
\newblock Designing experiments informed by observational studies.
\newblock {\em Journal of Causal Inference}, 9(1):147--171.

\bibitem[Russo et~al., 2018]{russo2018tutorial}
Russo, D.~J., Van~Roy, B., Kazerouni, A., Osband, I., Wen, Z., et~al. (2018).
\newblock A tutorial on {Thompson} sampling.
\newblock {\em Foundations and Trends{\textregistered} in Machine Learning},
  11(1):1--96.

\bibitem[Slud et~al., 2018]{slud2018combining}
Slud, E., Vonta, I., and Kagan, A. (2018).
\newblock Combining estimators of a common parameter across samples.
\newblock {\em Statistical Theory and Related Fields}, 2(2):158--171.

\bibitem[Tabord-Meehan, 2022]{tabord-meehan2022stratification}
Tabord-Meehan, M. (2022).
\newblock Stratification {Trees} for {Adaptive} {Randomisation} in {Randomised}
  {Controlled} {Trials}.
\newblock {\em The Review of Economic Studies}, 90(5):2646--2673.

\bibitem[van~der Vaart and Wellner, 1996]{van1996weak}
van~der Vaart, A.~W. and Wellner, J.~A. (1996).
\newblock {\em Weak {Convergence} and {Empirical} {Processes}}.
\newblock Springer {Series} in {Statistics}. Springer, New York.

\bibitem[Vapnik, 1991]{vapnik1991principles}
Vapnik, V. (1991).
\newblock Principles of risk minimization for learning theory.
\newblock {\em Advances in neural information processing systems}, 4.

\bibitem[Wainwright, 2019]{wainwright2019high}
Wainwright, M.~J. (2019).
\newblock {\em High-dimensional statistics: A non-asymptotic viewpoint},
  volume~48.
\newblock Cambridge University Press, Cambridge.

\bibitem[Wood, 2003]{wood2003thin-plate}
Wood, S.~N. (2003).
\newblock Thin-plate regression splines.
\newblock {\em Journal of the Royal Statistical Society (B)}, 65(1):95--114.

\bibitem[Wood, 2004]{wood2004stable}
Wood, S.~N. (2004).
\newblock Stable and efficient multiple smoothing parameter estimation for
  generalized additive models.
\newblock {\em Journal of the American Statistical Association},
  99(467):673--686.

\bibitem[Xu et~al., 2018]{xu2018fully}
Xu, L., Honda, J., and Sugiyama, M. (2018).
\newblock A fully adaptive algorithm for pure exploration in linear bandits.
\newblock In {\em International Conference on Artificial Intelligence and
  Statistics}, pages 843--851. PMLR.

\bibitem[Zhang et~al., 2020]{zhang2020inference}
Zhang, K., Janson, L., and Murphy, S. (2020).
\newblock Inference for batched bandits.
\newblock {\em Advances in Neural Information Processing Systems},
  33:9818--9829.

\bibitem[Zhang et~al., 2021]{zhang2021statistical}
Zhang, K., Janson, L., and Murphy, S. (2021).
\newblock Statistical inference with m-estimators on adaptively collected data.
\newblock {\em Advances in Neural Information Processing Systems},
  34:7460--7471.

\bibitem[Zhang et~al., 2007]{zhang2007asymptotic}
Zhang, L.-X., Hu, F., Cheung, S.~H., and Chan, W.~S. (2007).
\newblock Asymptotic properties of covariate-adjusted response-adaptive
  designs.
\newblock {\em Annals of Statistics}, 35(3):1166--1182.

\bibitem[Zhao, 2023]{zhao2023adaptive}
Zhao, J. (2023).
\newblock Adaptive {Neyman} allocation.
\newblock Technical report, SSRN \url{https://ssrn.com/abstract=4448249}.

\bibitem[Zhu and Zhu, 2023]{zhu2023covariate}
Zhu, H. and Zhu, H. (2023).
\newblock Covariate-adjusted response-adaptive designs based on semiparametric
  approaches.
\newblock {\em Biometrics}.

\end{thebibliography}

\end{document}